\documentclass[conference]{IEEEtran}
\IEEEoverridecommandlockouts 
\pdfoutput=1
\usepackage[todo=hide,theorems,tikz,mathtools=false,mnsymbol=false,colorlinks=false]{generic}
\usepackage{amsmath}
\usepackage{amsthm}
\usepackage[amsfonts]{logic-thm}
\usepackage{mathabx}
\usepackage{wrapfig}
\usetikzlibrary{patterns}
\usetikzlibrary{arrows,automata}
\usepackage[square,numbers,sort]{natbib}

\newcommand{\trans}[2][]{%
    \mathrel{\raisebox{-1pt}[10pt][0pt]{%
      $\overset{#2}{\underset{^{\raisebox{-2pt}[0pt][0pt]{$^{^{#1}}$}}}{\raisebox{0pt}[3pt][0pt]{%
        $\relbar\mspace{-8mu}\xrightarrow{\quad~}$%
       }}}$%
    }}}

\newcommand{\cmax}{c_{\mathsf{max}}}
\newcommand{\hmax}{h_{\mathsf{max}}}
\newcommand{\emax}{e_{\mathsf{max}}}
\newcommand{\bound}{\boldsymbol{B}}

\newcommand{\PR}[1]{{\normalfont\bfseries P#1}}

\let\oldvdash\vdash
\let\olddashv\dashv
\renewcommand{\vdash}{\mathop{\oldvdash}}
\renewcommand{\dashv}{\mathop{\olddashv}}
\newcommand{\Left}{\ensuremath{\mathsf{left}}}
\newcommand{\Right}{\ensuremath{\mathsf{right}}}

\newcommand{\tr}[1]{\ensuremath{\mathsf{tr}}(#1)}
\newcommand{\an}[1]{\ensuremath{\mathsf{an}}(#1)}
\newcommand{\out}[1]{\ensuremath{\mathsf{out}}(#1)}
\newcommand{\outb}[1]{\ensuremath{\mathsf{out}}\big(#1\big)}
\newcommand{\block}[1]{\ensuremath{\mathsf{block}}(#1)}

\renewcommand{\dom}{\ensuremath{\mathsf{dom}}}

\newcommand{\pump}{\ensuremath{\mathsf{pump}}}

\newcommand{\LL}{{\ensuremath{\mathsf{LL}}}}
\renewcommand{\RR}{{\ensuremath{\mathsf{RR}}}}
\newcommand{\LR}{{\ensuremath{\mathsf{LR}}}}
\newcommand{\RL}{{\ensuremath{\mathsf{RL}}}}

\let\oldlhd\lhd 
\let\oldrhd\rhd
\renewcommand{\lhd}{{\oldlhd}}
\renewcommand{\rhd}{{\oldrhd}}

\newcommand{\crossrel}{\ensuremath{\mathrel{\text{\sf S}}}}

\renewcommand{\simeq}{\crossrel^*}

\newcommand{\lesstime}{\mathrel{\lhd}}
\newcommand{\leqtime}{\mathrel{\unlhd}}

\newcommand{\geqtime}{\mathrel{\unrhd}}

\newcommand{\lesspair}{\sqsubset}

\providecommand{\underbracket}[2][]{\underbrace{#2}}
\providecommand{\overbracket}[2][]{\overbrace{#2}}

\makeatletter
\def\shortrightarrowfill@{\arrowfill@\relbar\relbar\shortrightarrow}
\def\shortleftarrowfill@{\arrowfill@\shortleftarrow\relbar\relbar}
\def\shortleftrightarrowfill@{\arrowfill@\leftrightarrow}
\newcommand{\ort}{\mathpalette{\overarrow@\shortrightarrowfill@}}
\newcommand{\olft}{\mathpalette{\overarrow@\shortleftarrowfill@}}
\newcommand{\olftrt}{\mathpalette{\overarrow@\shortleftrightarrowfill@}}
\makeatother

\makeatletter
\newcommand{\fixed@sra}{$\vrule height 2\fontdimen22\textfont2 width 0pt\shortrightarrow$}
\newcommand{\upperleft}{{\mspace{-2mu}\text{\rotatebox[origin=c]{\numexpr135}{\fixed@sra}}\mspace{-2mu}}}
\newcommand{\lowerright}{{\mspace{-2mu}\text{\rotatebox[origin=c]{\numexpr315}{\fixed@sra}}\mspace{-2mu}}}
\newcommand{\leftshort}{{\shortleftarrow}}
\newcommand{\rightshort}{{\shortrightarrow}}
\makeatother

\makeatletter
\let\pmb\boldsymbol
\makeatother

\newtheorem*{claim}{\bfseries Claim}
\newenvironment{claimproof}[1][\proofname]
  {%
    \proof[#1]%
  }
  {%
    \endproof%
  }

\renewcommand{\paragraph}[1]{\smallskip\par\noindent {\bf #1}}

\newcommand\customtheorem[1]{\let\oldthetheorem\thetheorem%
                             \renewcommand\thetheorem{#1}}
\newcommand\undocustomtheorem{\let\thetheorem\oldthetheorem\addtocounter{theorem}{-1}}

\begin{document}
\sloppy

\title{Untwisting two-way transducers in elementary time\thanks{This work was partially supported by the ANR projects ExStream (ANR-13-JS02-0010) and DeLTA (ANR-16-CE40-0007).}}

\author{\IEEEauthorblockN{F{\'e}lix Baschenis}
        \IEEEauthorblockA{Universit\'e de Bordeaux, LaBRI \\
                          felix.baschenis@labri.fr}
        \and
        \IEEEauthorblockN{Olivier Gauwin}
        \IEEEauthorblockA{Universit\'e de Bordeaux, LaBRI \\
                          olivier.gauwin@labri.fr}
        \and
        \IEEEauthorblockN{Anca Muscholl}
        \IEEEauthorblockA{Universit\'e de Bordeaux, LaBRI \\
                          anca@labri.fr}
        \and
        \IEEEauthorblockN{Gabriele Puppis}
        \IEEEauthorblockA{CNRS, LaBRI \\
                          gabriele.puppis@labri.fr}
       }

\maketitle

\thispagestyle{plain}
\pagestyle{plain}

\begin{abstract}
Functional transductions realized by two-way transducers 
(equivalently, by streaming transducers and by MSO transductions) 
are the natural and standard notion of ``regular'' mappings from words to words.
It was shown recently (LICS'13) that it is decidable if such a 
transduction can be implemented by some one-way transducer, 
but the given algorithm has non-elementary complexity. 
We provide an algorithm of different flavor solving the 
above question, that has double exponential space complexity.
We further apply our technique to decide whether the transduction 
realized by a two-way transducer can be implemented by a sweeping 
transducer, with either known or unknown number of passes.

\end{abstract}

\IEEEpeerreviewmaketitle

\section{Introduction}\label{sec:introduction}

Since the early times of computer science, transducers have been
identified as a fundamental notion of computation, where one
is interested how objects can be transformed into each other. Numerous
fields of computer science are ultimately concerned with
transformations, ranging from databases to image
processing, and an important issue is to perform transformations with
low costs, whenever possible. 

The most basic form of transformers are devices that process an input
and produce outputs during the processing, using  finite memory. Such
devices are called finite-state transducers. Word-to-word finite-state
transducers were considered in very early work in formal language
theory~\cite{sch61,ahu69,eil76}, and it was soon clear that they are much more
challenging than finite-state word acceptors - the classical finite-state
automata. One essential difference between transducers and automata
over words is that the capability to process the
input in both directions strictly increases the expressive power in the case
of transducers, whereas this does not for
automata~\cite{RS59,she59}. In other words, two-way word transducers
are strictly more expressive than one-way word transducers.

We consider in this paper  functional transducers, that compute
functions from words to words. Two-way word transducers capture very
nicely the notion of regularity in this setting. Regular word
functions, i.e. functions computed by functional two-way transducers, inherit many
of the characterizations and algorithmic properties of the robust
class of regular languages. Engelfriet and Hoogeboom~\cite{EH98} 
showed that monadic second-order definable graph transductions,
restricted to words, are equivalent to two-way transducers --- this
justifies the 
notation ``regular'' word functions, in the spirit of classical
results in automata theory and logic by B\"uchi, Elgot, Rabin and
others. Recently, Alur and Cern\'{y}~\cite{AlurC10} proposed an enhanced
version of one-way transducers called streaming transducers, and
showed that they are equivalent to the two previous models. A streaming
transducer processes the input word from left to right, and stores
(partial) output words in finitely many, write-only registers.

Two-way transducers raise challenging questions about resource
requirements. One crucial resource is the number of times the
transducer needs to re-process the input word. In particular, the case
where the input can be processed in a single pass, from left to right,
is very attractive as it corresponds to the setting of \emph{streaming},
where the (possibly very large) inputs do not need to be stored in order
to be processed. Recently, it was shown in~\cite{fgrs13} that it is
decidable whether the transduction  defined by a functional two-way
transducer can be implemented by a one-way transducer. However, the decision procedure
of~\cite{fgrs13} has non-elementary complexity, and it is very natural
to ask whether one can do better.  We gave
in~\cite{bgmp15,bgmp16} an exponential space algorithm  in the special case
of~\emph{sweeping} transducers: head reversals are only allowed at the
extremities of the input. However, sweeping transducers are known to be
strictly less expressive than two-way transducers. 

In this paper we provide an algorithm of elementary complexity for deciding 
whether the transduction defined by a functional two-way transducer can be
implemented by a one-way transducer: the decision algorithm has double
exponential space complexity, and an equivalent one-way transducer (if
it exists), can be constructed with triple exponential size.
The known lower bound~\cite{bgmp15} is double exponential size.
Our techniques can be further adapted to characterize definability of transductions
by other models of transducers, e.g.~to characterize sweeping 
transducers within the class of two-way transducers.

\emph{Related work.} Besides the papers mentioned above, there are
several recent results around the expressivity and the resources of
two-way transducers, or 
equivalently, streaming transducers. 
First-order definable transductions were shown to be equivalent to
transductions defined by aperiodic
streaming transducers~\cite{FiliotKT14} and to aperiodic two-way
transducers~\cite{CartonDartois15}. An effective characterization of
aperiodicity for one-way transducers was obtained in~\cite{FGL16}.

In~\cite{DRT16,bgmp16} the
minimization of the number of registers of deterministic streaming transducers,
resp., passes of functional sweeping transducers, was shown to be
decidable. An algebraic characterization of (not necessarily
functional) two-way transducers over unary alphabets was 
provided in~\cite{CG14mfcs}. It was shown that in this case 
sweeping transducers have the same expressivity. 
The expressivity of non-deterministic input-unary
or output-unary two-way
transducers was investigated in~\cite{Gui15}. 

\emph{Overview.} Section~\ref{sec:preliminaries} introduces basic
notations for two-way transducers, and Section~\ref{sec:mainres}
states the main result. Section~\ref{sec:effects} is devoted to the
effect of pumping runs on outputs, and Section~\ref{sec:inversions}
introduces the main tool for our characterization. 
Section~\ref{sec:decomposition} handles the construction of an
equivalent one-way transducer.
Finally, Section~\ref{sec:sweepingness} describes a procedure to decide
whether a functional transducer is equivalent to a sweeping transducer.

\section{Preliminaries}\label{sec:preliminaries}

\paragraph{Two-way automata and transducers.}
We start with some basic notations and definitions for two-way
automata (resp., transducers). We assume that every input word $u=a_1\cdots a_n$ 
has two special delimiting symbols $a_1 = \vdash$
and $a_n = \dashv$ that do not occur elsewhere: $a_i \notin \{\vdash,\dashv\}$ 
for all $i=2,\dots,n-1$. 

A \emph{two-way automaton} $\Aa=\tup{Q,\S,\vdash,\dashv,\de,q_0,F}$ has 
a finite state set $Q$, input alphabet $\S$, transition relation 
$\de \subseteq Q \times (\S \cup\set{\vdash,\dashv}) \times Q \times
\set{\Left,\Right}$, initial state $q_0\in Q$, and set of final states
$F\subseteq Q$. By convention, left transitions on $\vdash$ are not allowed.
A \emph{configuration} of $\Aa$ has the form $u\,q\,v$, 
with $uv \in \{\vdash\}\cdot\S^*\cdot\{\dashv\}$ 
and $q \in Q$. A configuration $u\,q\,v$ represents the 
situation where the current state of $\Aa$ is $q$ and its head reads the first 
symbol of $v$ (on input $uv$). If $(q,a,q',\Right) \in \de$,
then there is a transition from any configuration of the form
$u\,q\,av$ to the configuration $ua\,q'\,v$, which we denote 
$u\,q\,av \trans{a,\Right} ua\,q'\,v$. 
Similarly, if $(q,a,q',\Left) \in \de$,
then there is a transition from any configuration of the form
$ub\,q\,av$ to the configuration $u\,q'\,bav$, 
denoted as $ub\,q\,av \trans{a,\Left} u\,q'\,bav$.
A \emph{run} on $w$ 
is a sequence of transitions.
It is \emph{successful} if it starts in the initial configuration 
$q_0\, w$ and ends in a configuration $w\,q$ with $q \in F$
--- note that this latter configuration does not allow additional
transitions. The \emph{language} of $\Aa$ is the set of input words
that admit a successful run of $\Aa$.

The definition of \emph{two-way transducers} is similar to that of two-way automata,
with the only difference that now there is an additional output alphabet $\Gamma$ and the transition 
relation is a finite subset of $Q \times (\S \cup \{\vdash,\dashv\}) \times \G^* \times  Q \times \{\Left,\Right\}$,
which associates an output over $\Gamma$ with each transition of the underlying two-way automaton.
Formally, given a two-way transducer $\Tt=\tup{Q,\S,\vdash,\dashv,\Gamma,\de,q_0,F}$,
we have a transition of the form $ub\,q\,av \trans{a,d|w} u'\,q'\,v'$, outputting $w$, 
whenever $(q,a,w,q',d)\in\delta$ and either $u'=uba$, $v'=v$ or 
$u'=u$, $v'=bav$, depending on whether $d=\Right$ or $d=\Left$.
The \emph{output} associated with a  run 
$\rho = u_1\,q_1\,v_1 \trans{a_1,d_1|w_1} \dots \trans{a_n,d_n|w_n} u_{n+1}\,q_{n+1}\,v_{n+1}$
of $\Tt$ is the word $\out{\rho} = w_1\cdots w_n$. A transducer $\Tt$
defines a relation consisting of all pairs $(u,w)$ such that
$w=\out{\r}$, for some successful run $\r$ on $u$. 

The \emph{domain} of $\Tt$, denoted $\dom(\cT)$, is the set of input words that 
have a successful run. 
For transducers $\cT,\cT'$, we write $\cT' \subseteq \cT$ 
to mean that $\dom(\Tt') \subseteq \dom(\Tt)$ and the transductions 
computed by $\Tt,\Tt'$ coincide on $\dom(\Tt')$.

We say that $\Tt$ is \emph{functional} 
if for each input $u$, at most one output $w$ can be
produced by any possible successful run on $u$.
Finally, we say that $\Tt$ is \emph{one-way} if it does not have transition rules
of the form $(q,a,w,q',\Left)$.

\begin{figure}[!t]
\centering
\scalebox{.66}{%
\begin{tikzpicture}[->,>=stealth',shorten >=1pt,auto,node distance=2.3cm,semithick,xscale=1.3,yscale=0.75]

  \node[state,] (A)       at (0,2)             {$q_0  $};
  \node[state,] (B)   [right of=A]                 {$  q_1  $};
  \node[state,] (C)   [right of=B]                 {$  q_2   $};
  \node[state,] (D)   [above=1cm of C]                 {$  q_3  $};
   \node[state,] (E)   [left of=D]                 {$ q_4   $};
 
  \node[state,] (F)   [above=1cm of E]                 {$  q_5    $};
  \node[state,] (G)   [right of=F]                 {$ q_6    $};
   \node[state, ] (H)   [right of=G]                 {$ q_7  $};
   \node[state,] (I) [right of=H] { $q_8$} ;
 \coordinate[below=1cm of A] (d1);   
   
  \path (A) edge              node {$a_1 , \Right $ } (B) ;
   \path (B) edge              node {$a_2 , \Right $ } (C) ;
    \path (C) edge       [bend right=80]       node [swap] {$a_3 , \Left $ } (D) ;
    \path (D) edge              node {$a_2 , \Left $ } (E) ;
     \path (E) edge       [bend left=80]       node  {$a_1 , \Right $ } (F) ;
       \path (F) edge              node {$a_2 , \Right $ } (G) ;
         \path (G) edge              node {$a_3 , \Right $ } (H) ;
          \path (H) edge              node {$a_4 , \Right $ } (I) ;
    
    \node[state,rectangle,minimum size=3mm] (X) at (1,-1) {$a_1$};
    \node[state,rectangle,minimum size=3mm] (Y) [right of=X] {$a_2$};
    \node[state,rectangle,minimum size=3mm] (Z) [right of=Y] {$a_3$};
    \node[state,rectangle,minimum size=3mm] (W) [right of=Z] {$a_4$};
    
    \node at (-2,-1) {Input word:};
    \node at (-2,0.5) {Positions:};
    \node at (-2,2) {Run:};
    
    \node (M) at (0,0.5) {$0$};
    \node (N) [right of=M] {$1$};
    \node (O) [right of=N] {$2$};
    \node (P) [right of=O] {$3$};
    \node (Q) [right of=P] {$4$};

  \node (AA) [above=0cm of A] {$(0,0)$} ;
  \node (AB) [right of=AA] {$(1,0)$};
  \node (AC) [right of=AB] {$(2,0)$};
  \node (AD) [above=0cm of D] {$(2,1)$};
  \node (AE) [left of=AD] {$(1,1)$};
  \node (AF) [above=0cm of F] {$(1,2)$};
  \node (AG) [right of=AF] {$(2,2)$};
  \node (AH) [right of=AG] {$(3,0)$};
  \node (AI) [right of=AH] {$(4,0)$};
\end{tikzpicture}
}
\caption{Graphical presentation of a run by means of crossing sequences.}\label{fig:run}
\end{figure}
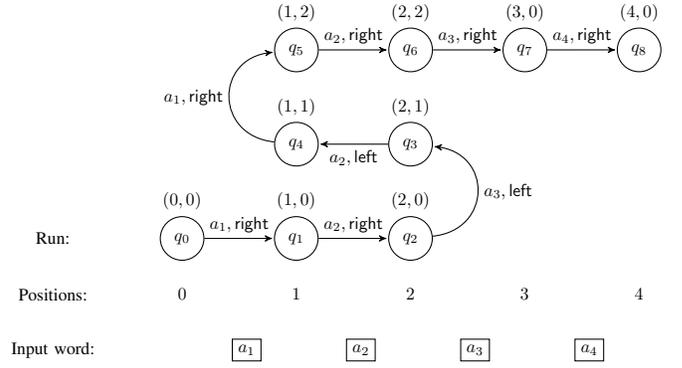

\paragraph{Crossing sequences.}
The first basic notion is that of crossing sequence. 
We follow the convenient presentation from \cite{HU79}, which appeals to a
graphical representation of runs of a two-way transducer where each configuration
is seen as point (location) in a two-dimensional space. 
Let $u=a_1\cdots a_n$ be an input word (recall that $a_1=\vdash$ and $a_n=\dashv$) 
and let $\rho$ be a run of a two-way automaton (or transducer) $\Tt$ on $u$.
The \emph{positions} of $\rho$ are the numbers from $0$ to $n$, corresponding
to ``cuts'' between two consecutive letters of the input. For example, 
position $0$ is just before the first letter $a_1$,
position $n$ is just after the last letter $a_n$,
and any other position $x$, with $1\le x<n$, is between 
the letters $a_x$ and $a_{x+1}$.
 
We say that a transition $u\,q\,v \trans{a,d} u'\,q'\,v'$ of $\rho$ \emph{crosses position $x$} 
if either $d=\Right$ and $|u|=x$, or $d=\Left$ and $|u'|=x$.
A \emph{location} of $\rho$ is any pair $(x,y)$ for which there are at least 
$y+1$ transitions in $\rho$ crossing position $x$; the component $y$ of a location
is called \emph{level}.
Each location is associated a state. Formally, we say that
$q$ is the \emph{state at location $\ell=(x,y)$} in $\rho$, and we denote this 
by writing $\rho(\ell)=q$, if the $(y+1)$-th transition that crosses $x$ ends up
in state $q$.
The \emph{crossing sequence} at position $x$ of $\rho$ is the tuple $\rho|x=(q_0,\dots,q_h)$,
where the $q_y$'s are all the states at locations of the form $(x,y)$, for $y=0,\dots,h$.

As suggested by Fig.~\ref{fig:run}, any run can be represented 
as an annotated path between locations.
For example, if a location $(x,y)$ is reached by a rightward transition,
then the head of the automaton has read the symbol $a_x$; 
if it is reached by a leftward transition, then the head has read the symbol $a_{x+1}$.  
Note that in a successful run $\rho$ every crossing sequence has 
odd length and every rightward (resp.~leftward) transition reaches 
a location with even (resp.~odd) level.
We can identify four types of transitions between locations, 
depending on the parities of the levels (the reader may refer 
again to Fig.~\ref{fig:run}):
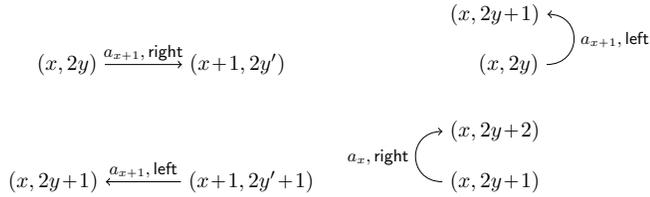
\begin{center}
\vspace{-0.5mm}
\scalebox{.81}{%
\begin{tikzpicture}[baseline=0,yscale=1.1]
  \draw (-0.25,1.75) node (node1) {$\phantom{\!+\!1}~(x,2y)$};
  \draw (3,1.75) node (node2) {$(x\!+\!1,2y')~\phantom{\!+\!1}$};
  \draw (-0.25,0) node (node3) {$(x,2y\!+\!1)$};
  \draw (3,0) node (node4) {$(x\!+\!1,2y'\!+\!1)$};
  \draw (7,1.75) node (node5) {$\phantom{\!+\!1}~(x,2y)$};
  \draw (7,2.5) node (node6) {$(x,2y\!+\!1)$};
  \draw (7,0) node (node7) {$(x,2y\!+\!1)$};
  \draw (7,0.75) node (node8) {$(x,2y\!+\!2)$};
  \draw (node1) edge [->] node [above=-0.05, scale=0.9] {\small $a_{x+1},\Right$} (node2);
  \draw (node4) edge [->] node [above=-0.05, scale=0.9] {\small $a_{x+1},\Left$} (node3);
  \draw (node5.east) edge [->, out=0, in=0, looseness=2] 
        node [right=0, scale=0.9] {\small $a_{x+1},\Left$} (node6.east);
  \draw (node7.west) edge [->, out=180, in=180, looseness=2] 
        node [left=0, scale=0.9] {\small $a_x,\Right$} (node8.west);
\end{tikzpicture}
}
\end{center}
Hereafter, we will identify runs with the corresponding annotated paths between locations. 
It is also convenient to define a total order $\leqtime$ on the locations of a run $\rho$ 
by letting $\ell_1 \leqtime \ell_2$ if $\ell_2$ is reachable from $\ell_1$ by following the 
path described by $\rho$ --- the order $\leqtime$ on locations is called 
\emph{run order}. 
Given two locations $\ell_1 \leqtime \ell_2$ of a run $\rho$, we write $\rho[\ell_1,\ell_2]$ 
for the factor of the run that starts in $\ell_1$ and ends in $\ell_2$. Note that the latter 
is also a run and hence the notation $\outb{\rho[\ell_1,\ell_2]}$ is permitted.
Two runs $\rho_1,\rho_2$ can be concatenated, provided that $\rho_1$ ends in location $(x,y)$, 
$\rho_2$ starts in location $(x,y')$, such that $y'=y \pmod{2}$ and $(x,y)$, $(x,y')$ are
labelled by the same state. We denote by $\r_1 \r_2$ the run resulting from concatenating
$\r_1$ with $\r_2$.
Clearly, we have $\rho[\ell_1,\ell_2] ~ \rho[\ell_2,\ell_3] = \rho[\ell_1,\ell_3]$
for all locations $\ell_1 \leqtime \ell_2 \leqtime \ell_3$.

\paragraph{Normalization.}
Without loss of generality, we will assume that successful runs of
functional transducers are \emph{normalized}, meaning that they never 
visit two locations with the same position, the same state, and both 
either at even or at odd level. Indeed, if this were not
the case, say if a successful run $\rho$ visited two locations $\ell_1=(x,y)$
and $\ell_2=(x,y')$ such that $\rho(\ell_1)=\rho(\ell_2)$ and
$y,y'$ are both even or both odd, then the output produced by $\rho$ 
between $\ell_1$ and $\ell_2$ should be empty, as otherwise
by repeating the factor $\rho[\ell_1,\ell_2]$ of $\rho$ we could 
obtain successful runs that produces different outputs on the same 
input, thus contradicting the assumption that the transducer is
functional. Now that we know that the output of $\rho$ produced 
between $\ell_1$ and $\ell_2$ is empty, we could drop the
factor $\rho[\ell_1,\ell_2]$, thus obtaining a successful run
with the same output.
It is easy to see that, in every normalized successful run, 
the crossing sequences have length at most $2|Q|-1$.

We define $\hmax = 2|Q|-1$. Moreover, by $\cmax$ we denote the
\emph{capacity} of the transducer, which is the maximal length of the
output of a transition.

\section{Two-way transducers vs one-way transducers}\label{sec:mainres}

In this section we state our main result, which is the existence of an
elementary algorithm for checking whether a two-way transducer is
equivalent to some one-way transducer. We call such transducers
\emph{one-way definable}. Before stating our result, we
give a few examples.

\begin{example}\label{ex:one-way-definability}
We consider two-way transducers that accept any input $u$ from a given
regular language $R$  
and output the word $u\,u$. We will argue how, depending on  $R$, 
these transducers may or may not be one-way definable.
\begin{enumerate}
\item If $R=(a+b)^*$ there is no equivalent one-way transducer, 
      as the output language is not regular. 
      If $R$ is finite, then the transduction mapping
        $u\in R$ to $u\,u$ can be implemented by a one-way transducer that guesses
        $u$ (this requires as many states as the size of $R$), checks the
        input, and outputs two copies of the guessed word.
  \item A special case of transduction with finite domain is given by
        $R_n = \{ a_0 \, w_0 \, \cdots a_{2^n-1} \, w_{2^n-1} \::\: 
                  a_0,\dots,a_{2^n-1}\in\{a,b\} \}$, 
        where $n\in\bbN$ and each $w_i$ is the binary encoding of the counter $i=0,\dots,2^n-1$.
        It is easy to see (cf.~Proposition 15 \cite{bgmp15}) that
        the transduction mapping
        $u\in R_n$ to $u\,u$ can be implemented by a two-way transducer with quadratically many
        states w.r.t.~$n$, while every equivalent one-way transducer has at least 
        $2^{2^n}$ states, since it needs to guess a word of length $2^n$.
  \item Consider now the periodic language $R=(abc)^*$. 
        The function that maps $u\in R$ to $u\,u$ can be easily implemented by a 
        one-way transducer: it suffices to output two letters
        (i.e., $ab$, $ca$, $bc$, in turn)
        for each input letter, while checking that the input is in $R$.
\end{enumerate}
\end{example}

\begin{example}\label{ex:running}
We consider a slightly more complicated transduction
that is defined on input words of the form $u_1 \:\#\: \dots \:\#\: u_n$,
where each factor $u_i$ is over the alphabet $\Sigma=\{a,b,c\}$. 
The output of the transduction is of the form
$w_1 \:\#\: \dots \:\#\: w_n$, where each $w_i$
is either $u_i \: u_i$ or just $u_i$, depending on whether or not
$u_i\in (abc)^*$ and $u_{i+1}$ has even length, with $u_{n+1}=\emptystr$.

The obvious way to implement the transduction is by means
of a two-way transducer that performs multiple passes 
on the factors of the input:
a first left-to-right pass is performed on 
$u_i \,\#\, u_{i+1}$ to produce the first copy of $u_i$ 
and to check whether $u_i\in (abc)^*$ and $|u_{i+1}|$ is even; if so,
a second pass on $u_i$ is performed to produce 
another copy of $u_i$.

The transduction can also be implemented by a one-way
transducer: when entering a factor $u_i$, the transducer
guesses whether or not $u_i\in (abc)^*$ and $|u_{i+1}|$ is even; 
depending on this it outputs either $(abc\,abc)^{\frac{|u_i|}{3}}$ 
or $u_i$, and checks that the guess is correct.
\end{example}

\noindent
Our main result is:
\begin{theorem}\label{th:main}
There is an algorithm that from a functional two-way
transducer $\cT$ constructs in triple exponential time 
a \emph{one-way} transducer $\cT'$ with the following properties:
\begin{itemize}
  \item $\cT'\subseteq\cT$,
  \item $\dom(\cT)=\dom(\cT')$ iff $\Tt$ is one-way definable.
\end{itemize}
Moreover, the second property above can be checked in double exponential space
w.r.t.~$|\cT|$. 
\end{theorem}

We remark that a similar characterization for a much more
restricted class of transducers (sweeping transducers) 
appeared in \cite{bgmp15}. The proof of Theorem~\ref{th:main},
however, is more technical, as it requires a better understanding
of the structure of the runs of two-way transducers and a 
non-trivial generalization of the combinatorial arguments 
from \cite{bgmp15}.

The proof of the theorem spans along the next three sections.
In Section \ref{sec:effects}, we present the basic concepts
for reasoning on runs of two-way automata. This includes the
definition of a finite semigroup for describing the shapes of 
two-way runs, as well as Ramsey-type arguments that are
used to bound the length of the outputs produced by
pieces of runs without loops.
In Section \ref{sec:inversions} we provide the main combinatorial
arguments for characterizing one-way definability. The crucial
notion will be that of \emph{inversion}, that captures behaviours
of the two-way transducer that are problematic for one-way definability.
Finally, in Section \ref{sec:decomposition} we exploit the combinatorial
results and the Ramsey-type arguments to derive the existence of 
suitable decompositions of runs that lead to the construction of equivalent
one-way transducers.

\section{Untangling runs of two-way transducers}\label{sec:effects}

This section is devoted to untangling the structure of 
runs of two-way transducers. 
Whereas the classical transformation of two-way automata into one-way
automata based on crossing sequences is rather simple, we will need
a much deeper understanding of runs of two-way transducers, because of
the additional outputs. In a nutshell, being 
one-way definable is related to periodicities (with bounded periods) 
in the output, and these periodicities are generated by loops in the run. 
We will actually work with so called idempotent loops, 
that generate periodicities in the output in a ``nice'' way.
We will derive the existence of idempotent loops with bounded 
outputs using Ramsey-based arguments.

We fix throughout the paper a functional two-way transducer $\cT$, 
an input word $u$, and a successful run $\rho$ of $\cT$ on $u$. We 
assume that $\r$ is normalized, i.e., every state occurs at most once
in each crossing sequence of $\rho$ at levels of a given parity. 

For simplicity, we denote by $\omega$ the length of the input word $u$. 
We will consider \emph{intervals of positions} of the form 
$I=[x_1,x_2]$, with $0 \le x_1<x_2 \le \omega$. The
\emph{containment relation} $\subseteq$ on intervals is defined by
$[x_3,x_4] \subseteq [x_1,x_2]$ if $x_1\le x_3 < x_4\le x_2$.

\paragraph{Factors, flows, and effects.}
A factor of a run $\rho$ is a contiguous 
subsequence of $\r$. A factor \emph{intercepted} by an interval 
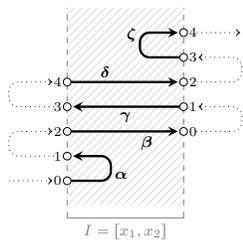
\begin{wrapfigure}{r}{3.2cm}
\vspace{-1.5mm}\hspace{-1mm}
\centering
\scalebox{0.73}{%
\begin{tikzpicture}[baseline=0, inner sep=0, outer sep=0, minimum size=0pt, xscale=0.53, yscale=0.45]
  \tikzstyle{dot} = [draw, circle, fill=white, minimum size=4pt]
  \tikzstyle{fulldot} = [draw, circle, fill=black, minimum size=4pt]
  \tikzstyle{factor} = [->, shorten >=1pt, dotted, rounded corners=5]
  \tikzstyle{fullfactor} = [->, >=stealth, shorten >=1pt, very thick, rounded corners=5]

  \fill [pattern=north east lines, pattern color=gray!25]
        (4,-1) rectangle (8,7);
  \draw [dashed, thin, gray] (4,-1) -- (4,7);
  \draw [dashed, thin, gray] (8,-1) -- (8,7);
  \draw [gray] (4,-1.25) -- (4,-1.5) -- (8,-1.5) -- (8,-1.25);
  \draw [gray] (6,-1.75) node [below] {\footnotesize $I=[x_1,x_2]$};

  \draw (2,0) node (node0) {};
  \draw (4,0) node [dot] (node1) {};
  \draw (node1) node [left=0.1] (node1') {\scriptsize $0$};
  \draw (5.5,0) node (node2) {};
  \draw (5.5,1) node (node3) {};
  \draw (4,1) node [dot] (node4) {};
  \draw (node4) node [left=0.1] (node4') {\scriptsize $1$};
  \draw (2,1) node (node5) {};
  \draw (2,2) node (node6) {};
  \draw (4,2) node [dot] (node7) {};
  \draw (node7) node [left=0.1] (node7') {\scriptsize $2$};
  \draw (8,2) node [dot] (node8) {};
  \draw (node8) node [right=0.1] (node8') {\scriptsize $0$};
  \draw (10,2) node (node9) {};
  \draw (10,3) node (node10) {};
  \draw (8,3) node [dot] (node11) {};
  \draw (node11) node [right=0.1] (node11') {\scriptsize $1$};
  \draw (4,3) node [dot] (node12) {};
  \draw (node12) node [left=0.1] (node12') {\scriptsize $3$};
  \draw (2,3) node (node13) {};
  \draw (2,4) node (node14) {};
  \draw (4,4) node [dot] (node15) {};
  \draw (node15) node [left=0.1] (node15') {\scriptsize $4$};
  \draw (8,4) node [dot] (node16) {};
  \draw (node16) node [right=0.1] (node16') {\scriptsize $2$};
  \draw (10,4) node (node17) {};
  \draw (10,5) node (node18) {};
  \draw (8,5) node [dot] (node19) {};
  \draw (node19) node [right=0.1] (node19') {\scriptsize $3$};
  \draw (6.5,5) node (node20) {};
  \draw (6.5,6) node (node21) {};
  \draw (8,6) node [dot] (node22) {};
  \draw (node22) node [right=0.1] (node22') {\scriptsize $4$};
  \draw (10,6) node (node23) {};

  \draw [factor] (node0) -- (node1');
  \draw [fullfactor] (node1) -- (node2.center) -- node [below right=1mm] {\footnotesize $\pmb{\alpha}$}
                     (node3.center) -- (node4); 
  \draw [factor] (node4') -- (node5.center) -- (node6.center) -- (node7'); 
  \draw [fullfactor] (node7) -- node [below right=1.25mm] {\footnotesize $~~\pmb{\beta}$} (node8);
  \draw [factor] (node8') -- (node9.center) -- (node10.center) -- (node11');
  \draw [fullfactor] (node11) -- node [below=1mm] {\footnotesize $\pmb{\gamma}$} (node12);
  \draw [factor] (node12') -- (node13.center) -- (node14.center) -- (node15');
  \draw [fullfactor] (node15) -- node [above left=1.25mm] {\footnotesize $\pmb{\delta}~~$} (node16);
  \draw [factor] (node16') -- (node17.center) -- (node18.center) -- (node19');
  \draw [fullfactor] (node19) -- (node20.center) -- node [above left=1mm] {\footnotesize $\pmb{\zeta}$}
                     (node21.center) -- (node22);
  \draw [factor] (node22') -- (node23);
\end{tikzpicture}
}
\caption{Intercepted factors.}\label{fig:intercepted-factors}
\hspace{-3mm}\vspace{-3mm}
\end{wrapfigure} 
$I=[x_1,x_2]$ is a maximal factor of $\rho$ that visits only 
positions $x\in I$, and never uses a left transition from 
position $x_1$ or a right transition from position $x_2$.  

Fig.~\ref{fig:intercepted-factors} on the right gives 
an example of an interval $I$ that intercepts the 
factors $\alpha,\beta,\gamma,\delta,\zeta$.
The numbers that annotate the endpoints of the factors 
represent their levels.

\smallskip
Every factor $\alpha$ intercepted by an interval $I=[x_1,x_2]$
is of one of the four types below, depending on its first
location $(x,y)$ and its last location $(x',y')$: 
\begin{itemize}
\item $\alpha$ is an $\LL$-factor if $x=x'=x_1$,
\item $\alpha$ is an $\RR$-factor if $x=x'=x_2$,
\item $\alpha$ is an $\LR$-factor if $x=x_1$ and $x'=x_2$, 
\item $\alpha$ is an $\RL$-factor if $x=x_2 $ and $x'=x_1$.
\end{itemize}
In Fig.~\ref{fig:intercepted-factors} we see that $\a$ is an
$\LL$-factor, $\b,\delta$ are $\LR$-factors, $\z$ is an $\RR$-factor,
and $\g$ is an $\RL$-factor.

\begin{definition}\label{def:flow}
Let $\r$ be a run and $I = [x_1, x_2]$ an interval of $\r$. 
Let $h_i$ be the length of the crossing sequence $\r|x_i$ for both $i=1$ and $i=2$.

The \emph{flow} $F_I$ of $I$ is a directed graph with set of nodes
$\set{0,\dots,\max(h_1,h_2)-1}$ and set of edges consisting of all
$(y,y')$ such that there exists a factor of $\r$ intercepted by $I$ 
that starts at location $(x_i,y)$ and ends at location $(x_j,y')$,
for $i,j\in\{1,2\}$.

The \emph{effect} $E_I$ of $I$ is the triple $(F_I,c_1,c_2)$, 
where $c_i=\r|x_i$ is the crossing sequence at $x_i$.
\end{definition}

\noindent
For example, the interval $I$ of Fig.~\ref{fig:intercepted-factors} 
has the flow graph $0\mapsto 1\mapsto 3\mapsto 4\mapsto 2\mapsto 0$.
It is easy to see that every node of a flow $F_I$ has at most one
incoming and at most one outgoing edge. More precisely, if $y<h_1$ is
even, then it has one outgoing edge (corresponding to an $\LR$- or $\LL$-factor
intercepted by $I$), and if it is odd it has one incoming edge (corresponding
to an $\RL$- or $\LL$-factor intercepted by $I$). Similarly, if $y<h_2$ is 
even, then it has one incoming edge (corresponding to an $\LR$- or $\RR$-factor), 
and if it is odd it has one outgoing edge (corresponding to an $\RL$- or
$\RR$-factor). 

In the following we consider generic effects that are not necessarily
associated with intervals of specific runs. The definition of such effects
should be clear: these are triples consisting of a graph (called flow)
and two crossing sequences of lengths $h_1,h_2 \le \hmax$, with sets of 
nodes of the form $\{0,\ldots,\max(h_1,h_2)-1\}$, 
that satisfy the in/out-degree properties stated above. 

It is convenient to distinguish the edges in a flow based on the
parity of the source and target nodes. Formally, we partition any 
flow $F$ into the following subgraphs:
\begin{itemize}
  \item $F_\LR$ consists of all edges of $F$ between pairs of even nodes,
  \item $F_\RL$ consists of all edges of $F$ between pairs of odd nodes,
  \item $F_\LL$ consists of all edges of $F$ from an even node to an odd node,
  \item $F_\RR$ consists of all edges of $F$ from an odd node to an even node.
\end{itemize}

We denote by $\cF$ (resp.~$\cE$) the set of all flows (resp.~effects)
augmented with a dummy element $\bot$. We equip both sets $\cF$ and $\cE$ with 
a semigroup structure, where the corresponding products $\circ$ and $\odot$ are 
defined below (similar definitions appear in \cite{Birget1990}). We need
this semigroup structure in order to identify \emph{idempotent loops},
that play a crucial role in our characterization of one-way
definability. 

\begin{definition}\label{def:product}
For two graphs $G,G'$, we denote by $G\cdot G'$ the graph with edges of 
the form $(y,y'')$ such that $(y,y')$ is an edge of $G$ and $(y',y'')$ is an
edge of $G'$, for some node $y'$ that belongs
to both $G$ and $G'$. 
Similarly, we denote by $G^*$ the graph with edges $(y,y')$ 
such that there exists a (possibly empty) path in $G$ from $y$ to $y'$.

The product of two flows $F,F'$ is the unique flow $F\circ F'$ (if it exists) such that:
\begin{itemize}
\item $(F\circ F')_\LR = F_\LR \cdot (F'_\LL \cdot F_\RR)^* \cdot F'_\LR$,
\item $(F\circ F')_\RL = F'_\RL \cdot (F_\RR \cdot F'_\LL)^* \cdot F_\RL$,
\item $(F\circ F')_\LL = F_\LL ~\cup~  F_\LR \cdot (F'_\LL \cdot F_\RR)^* \cdot F'_\LL \cdot F_\RL$,
\item $(F\circ F')_\RR = F'_\RR ~\cup~  F'_\RL \cdot (F_\RR \cdot F'_\LL)^* \cdot F_\RR \cdot F'_\LR$.
\end{itemize}
If no flow $F\circ F'$ exists with the above properties, 
then we let $F\circ F'=\bot$.

The product of two effects $E=(F,c_1,c_2)$ and $E'=(F',c'_1,c'_2)$ is either the effect
$E\odot E' = (F\circ F',c_1,c'_2)$ or the dummy element $\bot$, depending on whether 
$F\circ F'\neq \bot$ and $c_2=c'_1$.
\end{definition}

For example, let $F$ be the flow of interval $I$ 
in Fig.~\ref{fig:intercepted-factors}. Then 
$(F \circ F)_\LL=\set{(0,1),(2,3)}$, 
$(F \circ F)_\RR=\set{(1,2),(3,4)}$, and
$(F \circ F)_\LR=\set{(4,0)}$ 
--- one can quickly verify this with the 
help of Fig.~\ref{fig:pumping}.

It is also easy to see that $(\cF,\circ)$ and $(\cE,\odot)$ are finite semigroups, and that
for every run $\r$ and every pair of consecutive intervals $I=[x_1,x_2]$ and $J=[x_2,x_3]$ of $\r$,
$F_{I\cup J} = F_I \circ F_J$ and $E_{I\cup J} = E_I \odot E_J$.
In particular, the function $E$ that associates each interval $I$ of $\rho$ with the 
corresponding effect $E_I$ can be seen as a semigroup homomorphism. 

Note that, in a normalized successful run, there are at most 
$|Q|^{\hmax}$ distinct crossing sequences and at most $4^{\hmax}$ 
distinct flows, since there are at most $\hmax$ edges in a flow, 
and each one has one of the 4 possible types $\LL,\ldots,\RR$. 
Hence there are at most $(2|Q|)^{2\hmax}$ distinct effects.

\paragraph{Loops and components.}
Loops of a two-way run are the basic building blocks for characterizing 
one-way definability. We will consider special types of 
loops, called idempotent loops, when showing that outputs generated in 
non left-to-right manner are essentially periodic.

\begin{definition}\label{def:idempotent}
A \emph{loop} of $\r$ is an interval $L=[x_1,x_2]$ whose endpoints have 
the same crossing sequences, i.e.~$\r|x_1=\r|x_2$. 
It is said to be \emph{idempotent} if $E_L = E_L \odot E_L$
and $E_L\neq\bot$.
\end{definition}

\noindent
For example, the interval $I$ of Fig.~\ref{fig:intercepted-factors} is a loop,
if one assumes that the crossing sequences at the borders of $I$ are the same. 
However, by comparing with Fig.~\ref{fig:pumping}, it is easy to see that $I$ 
is not idempotent. On the other hand, the loop consisting of 2 copies of $I$ 
is idempotent. 

\input{figures/pumping}

Given a loop $L=[x_1,x_2]$ and a number $m\in\bbN$, we can introduce
$m$ new copies of $L$ and connect the intercepted factors in the 
obvious way.
Fig.~\ref{fig:pumping} shows how to do this for $m=1$ and $m=2$. 
The operation that we just described is called \emph{pumping}, 
and results in a new run of the transducer $\cT$ on the word 
\[
\pump_L^{m+1}(u)\; := \; u[0,x_1] \cdot \big( u[x_1+1,x_2] \big)^{m+1} \cdot u[x_2+1,n] \ .
\]
We denote by  $\pump_L^{m+1}(\rho)$ the pumped\footnote{Using similar constructions, one could remove a loop $L$ 
          from a run $\rho$, resulting in the run $\pump_L^0(\rho)$. 
          As we do not need this, 
          the operation $\pump_L$ will always be parametrized
          by a positive number $m+1$.} run on $\pump_L^{m+1}(u)$.%

The goal in this section is to describe the shape of the pumped run
$\pump_L^{m+1}(\rho)$ (and the produced output as well) when $L$ is an
{\sl idempotent} loop. We will focus on idempotent loops because
pumping non-idempotent loops may induce permutations of factors that 
are difficult to handle.  For example, if we consider again the 
non-idempotent loop $I$ to the left of Fig.~\ref{fig:pumping}, the factor 
of the run between $\beta$ and $\gamma$ (to the right of $I$, highlighted in red) 
precedes the factor between $\gamma$ and $\delta$ (to the left of $I$, 
again in red), but this ordering is reversed when a new copy of $I$ 
is added.

When pumping a loop $L$, subsets of factors intercepted by $L$ are glued
together to form longer factors intercepted by the unioned copies of $L$. 
The concept of component that we introduce below aims at identifying the
groups of factors that are glued together.

\begin{definition}\label{def:component}
A \emph{component} of a loop $L$ is any strongly 
connected component of its flow $F_L$ 
(note that this is also a cycle, since 
every node in it has in/out-degree $1$). 
Given a component $C$, we denote by 
$\min(C)$ (resp.~$\max(C)$) the minimum (resp.~maximum) 
node in $C$.
We say that $C$ is \emph{left-to-right} (resp.~\emph{right-to-left}) 
if $\min(C)$ is even (resp., odd).

\noindent
An \emph{$(L,C)$-factor} is a factor of the run that is 
intercepted by $L$ and corresponds to an edge of $C$.
\end{definition}

\noindent
For example, the loop $I$ of Fig.~\ref{fig:pumping} contains
a single component $C=\{0\mapsto 1\mapsto 3\mapsto 4\mapsto 2\mapsto 0\}$ 
which is left-to-right.
Another example is given in Fig.~\ref{fig:loop-trace}, where the 
loop $L$ has three components $C_1,C_2,C_3$ (ordered from bottom to top):
$\a_1,\a_2,\a_3$ are the $(L,C_1)$-factors, $\b_1,\b_2,\b_3$ are
the $(L,C_2)$-factors, and $\g_1$ is the unique $(L,C_3)$-factor.

We will usually list the $(L,C)$-factors based on their order of occurrence in the run.

\input{figures/pumpingcomp}

The following lemma (proved in the appendix) 
describes the precise shape and order of such 
factors when the loop $L$ is idempotent. 
It can be used to reason on the shape of 
runs obtained by pumping idempotent loops.

\begin{lemma}\label{lem:component2}
If $C$ is a left-to-right (resp.~right-to-left) component 
of an {\sl idempotent} loop $L$, then the $(L,C)$-factors are in the following order: 
$k$ $\LL$-factors (resp.~$\RR$-factors), followed by one $\LR$-factor (resp.~$\RL$-factor), 
followed by $k$ $\RR$-factors (resp.~$\LL$-factors), for some $k \ge 0$. 
\end{lemma}

We also need to introduce the notions of 
\emph{anchor} (Def.~\ref{def:anchor}) 
and \emph{trace} (Def.~\ref{def:trace}).

\begin{definition}\label{def:anchor}
Let $C$ be a component of an idempotent loop $L = [x_1,x_2]$.
The \emph{anchor} of $C$ inside $L$, denoted%
\footnote{In denoting the anchor --- and similarly the trace --- of a component $C$ 
          inside a loop $L$, we omit the annotation specifying $L$, since this is 
          often understood from the context.}
$\an{C}$, is either the location $\big(x_1,\max(C)\big)$ or the location 
$\big(x_2,\max(C)\big)$, depending on whether $C$ is left-to-right or right-to-left.
\end{definition}

\noindent
Intuitively, the anchor $\an{C}$ of a component $C$ of $L$ is the source location 
of the unique $\LR$- or $\RL$-factor intercepted by $L$ that corresponds to an edge of $C$ 
(recall Lemma \ref{lem:component2}).

\begin{definition}\label{def:trace}
Let $C$ be a component of some idempotent loop $L$ and let
$(i_0,i_1),(i_1,i_2),\dots,(i_{k-1},i_k),(i_k,i_{k+1})$ be 
a cycle of $C$, where $i_0=i_{k+1}=\max(C)$. 
For every $j=0,\dots,k$, let $\beta_j$ be the factor intercepted by 
$L$ that corresponds to the edge $(i_j,i_{j+1})$ of $C$.
The \emph{trace} of $C$ inside $L$ is the run $\tr{C} = \beta_0 ~ \beta_1 ~ \cdots ~ \beta_k$
(note that this is not necessarily a factor of the original run $\rho$).
\end{definition}

\noindent
Intuitively, the trace $\tr{C}$ is obtained by concatenating the $(L,C)$-factors together, 
where the first factor is the (unique) $\LR$-/$\RL$-factor that starts at the anchor $\an{C}$ 
and the remaining ones are the $\LL$-factors interleaved with the $\RR$-factors.

For example, by referring again to the components $C_1,C_2,C_3$
of Fig.~\ref{fig:loop-trace}, we have the following traces:
$\tr{C_1}=\alpha_2\:\alpha_1\:\alpha_3$,
$\tr{C_2}=\beta_2\:\beta_1\:\beta_3$, and $\tr{C_3}=\gamma_1$. 

As shown by the following proposition (proved in the appendix), 
iterations of idempotent loops translate to iterations of traces $\tr{C}$ of components.

\begin{proposition}\label{prop:pumping}
Let $L$ be an idempotent loop of $\rho$ with components $C_1,\dots,C_k$, 
listed according to the order of their anchors:
$\an{C_1}\lesstime\cdots\lesstime\an{C_k}$. 
For all $m\in\bbN$, we have 
\[
  \pump_L^{m+1}(\rho) ~=~ 
  \rho_0 ~ \tr{C_1}^m ~ \rho_1 ~ \cdots ~ \rho_{k-1} ~ \tr{C_k}^m ~ \rho_k
\]
where 
\begin{itemize}
  \item $\rho_0$ is the prefix of $\rho$ that ends at  $\an{C_1}$, 
  \item $\rho_i$ is the factor $\rho[\an{C_i},\an{C_{i+1}}]$, for all $1\le i<k$,
  \item $\rho_k$ is the suffix of $\rho$ that starts at  $\an{C_k}$.
\end{itemize}
\end{proposition}

For example, referring to the left hand-side of Fig.~\ref{fig:loop-trace}, 
the run $\r_0$ goes until the first location marked by a black dot. 
The run $\r_1$ and $\r_2$, resp., are between the first and the second 
black dot, and the second and third black dot. 
Finally, $\r_3$ is the suffix starting at the last black dot.
The pumped run $\pump_L^{m+1}(\rho)$ for $m=2$ is depicted to 
the right of Fig.~\ref{fig:loop-trace}.

\paragraph{Ramsey-type arguments.}
We conclude the section by describing a technique that 
can be used for bounding the length of the outputs produced 
by factors of the run $\rho$. 
This technique is based on Ramsey-type arguments 
and relies on Simon's ``factorization forest''
theorem~\cite{factorization_forests,factorization_forests_for_words_paper},
which we recall below.

Let $X$ be a set of positions of $\rho$.
A \emph{factorization forest} for $X$ is an unranked tree, where the nodes are 
intervals $I$ with endpoints in $X$, labelled with the corresponding effect $E_I$, 
the ancestor relation is given by the containment order on intervals, the leaves are the 
minimal intervals $[x_1,x_2]$, with $x_2$ successor of $x_1$ in $X$, and for every 
internal node $I$ with children $J_1,\dots,J_k$, we have:
\begin{itemize}
  \item $I=J_1\cup\dots\cup J_k$, 
  \item $E_I = E_{J_1}\odot\dots\odot E_{J_k}$, 
  \item if $k>2$, then $E_I = E_{J_1} = \dots = E_{J_k}$ 
        is an idempotent of the semigroup $(\cE,\odot)$.
\end{itemize}

We will make use of the following three constants defined from the transducer $\cT$:
the maximum number $\cmax$ of letters output by a single transition, 
the maximal length $\hmax = 2|Q|-1$ of a crossing sequence, and the
maximal size $\emax = (2|Q|)^{2\hmax}$ of the effect semigroup
$(\cE,\odot)$. 
By ${\boldsymbol{B = \cmax \cdot \hmax \cdot (2^{3\emax}+4)}}$ 
we will denote the main constant appearing in all subsequent sections.

\begin{theorem}%
[Factorization forest theorem 
 \cite{factorization_forests,factorization_forests_for_words_paper}]%
\label{th:simon}
For every set $X$ of positions of $\rho$, there is a factorization forest for $X$
of height at most $3\emax$.
\end{theorem}

It is easy to use the above theorem to show that every run 
that produces an output longer than $\bound$ contains an idempotent 
loop with non-empty output.  Below, we present a result 
in the same spirit, but refined in a way that it can be used to 
find anchors of components of loops inside specific intervals.

In order to state it formally, we need to consider subsequences 
of $\rho$ induced by sets of locations that are not necessarily intervals.
Recall that $\rho[\ell_1,\ell_2]$ denotes the factor of $\rho$ delimited 
by two locations $\ell_1\leqtime\ell_2$. Similarly, given any set $Z$ of 
(possibly non-consecutive) locations, we denote by $\rho\mid Z$ the
subsequence of $\rho$ induced by $Z$. 
\begin{wrapfigure}{r}{3.3cm}
\vspace{-1mm}\hspace{-2mm}
\scalebox{0.85}{%
\begin{tikzpicture}[baseline=0, inner sep=0, outer sep=0, minimum size=0pt, xscale=0.35, yscale=0.4]
  \tikzstyle{dot} = [draw, circle, fill=white, minimum size=4pt]
  \tikzstyle{fulldot} = [draw, circle, fill=black, minimum size=4pt]
  \tikzstyle{factor} = [->, shorten >=1pt, dotted, rounded corners=5]
  \tikzstyle{fullfactor} = [->, >=stealth, shorten >=1pt, very thick, rounded corners=5]

  \fill [pattern=north east lines, pattern color=gray!25]
        (6,-0.75) rectangle (12,6.75);
  \draw [dashed, thin, gray] (6,-0.75) -- (6,6.75);
  \draw [dashed, thin, gray] (12,-0.75) -- (12,6.75);
  \draw [gray] (6,-1) -- (6,-1.25) -- (12,-1.25) -- (12,-1);
  \draw [gray] (9,-1.5) node [below] {\footnotesize $I=[x_1,x_2]$};

  \draw (3,0) node (node0) {};
  \draw (4.5,0) node [fulldot] (node0') {};
  \draw (6,0) node [dot] (node1) {};
  \draw (8,0) node (node2) {};
  \draw (8,1) node (node3) {};
  \draw (6,1) node [dot] (node4) {};
  \draw (4,1) node (node5) {};
  \draw (4,2) node (node6) {};
  \draw (6,2) node [dot] (node7) {};
  \draw (12,2) node [dot] (node8) {};
  \draw (13,2) node (node9) {};
  \draw (13,3) node (node10) {};
  \draw (12,3) node [dot] (node11) {};
  \draw (10,3) node (node12) {};
  \draw (10,4) node (node13) {};
  \draw (12,4) node [dot](node14) {};
  \draw (14,4) node (node15) {};
  \draw (14,5) node (node16) {};
  \draw (12,5) node [dot] (node17) {};
  \draw (6,5) node [dot] (node18) {};
  \draw (3,5) node (node19) {};
  \draw (3,6) node (node20) {};
  \draw (4,6) node [fulldot] (node20') {};
  \draw (6,6) node [dot] (node21) {};
  \draw (12,6) node [dot] (node22) {};
  \draw (14,6) node (node23) {};

  \draw [factor] (node0) -- (node0');
  \draw [factor] (node0') -- (node1);
  \draw [fullfactor] (node1) -- (node2.center) -- (node3.center) -- (node4); 
  \draw [factor] (node4) -- (node5.center) -- (node6.center) -- (node7); 
  \draw [fullfactor] (node7) -- (node8);
  \draw [factor] (node8) -- (node9.center) -- (node10.center) -- (node11); 
  \draw [fullfactor] (node11) -- (node12.center) -- (node13.center) -- (node14); 
  \draw [factor] (node14) -- (node15.center) -- (node16.center) -- (node17); 
  \draw [fullfactor] (node17) -- (node18);
  \draw [factor] (node18) -- (node19.center) -- (node20.center) -- (node20'); 
  \draw [factor] (node20') -- (node21); 
  \draw [factor] (node21) -- (node22);
  \draw [factor] (node22) -- (node23);
  
  \draw (node0') node [below = 1.2mm] {$\ell_1~~$};
  \draw (node20') node [above right = 1mm] {$\ell_2$};
\end{tikzpicture}
}
\hspace{-3mm}\vspace{-3mm}
\end{wrapfigure} 
A transition of $\rho\mid Z$ is a transition from 
some $\ell$ to $\ell'$, where
both $\ell,\ell'$ belong to $Z$. The
output $\out{\rho\mid Z}$ is the concatenation of the outputs of the
transitions of $\rho\mid Z$ (in the order given by $\r$).
An example of subrun $\rho\mid Z$ is represented 
by the thick arrows in the figure to the right, where
$Z=[\ell_1,\ell_2]\cap (I\times\bbN)$.

\begin{theorem}\label{thm:simon2}
Let $I=[x_1,x_2]$ be an interval of positions, $K=[\ell_1,\ell_2]$ 
an interval of locations, and $Z = K \:\cap\: (I\times\bbN)$.
If $\big|\out{\rho\mid Z}\big| > \bound$, 
then there exist an idempotent loop $L$ and a component $C$ of $L$ such that
\begin{itemize}
  \item $x_1 < \min(L) < \max(L) < x_2$ (in particular, $L\subsetneq I$),
  \item $\ell_1\lesstime\an{C}\lesstime\ell_2$ (in particular, $\an{C} \in K$),
  \item $\out{\tr{C}} \neq \emptystr$.
\end{itemize}
\end{theorem}

\section{Inversions and periods}\label{sec:inversions}

As suggested by Examples~\ref{ex:one-way-definability} and~\ref{ex:running}, 
a typical phenomenon that may prevent a transducer from being one-way definable
is that of an \emph{inversion}. An inversion essentially corresponds to
a long output produced from right to left.
The main result in this section is
Proposition~\ref{prop:inversion}, that shows that the output produced
between the locations delimiting an inversion must be periodic, with 
bounded period.

\begin{definition}\label{def:inversion}
An \emph{inversion} of $\r$ is a tuple $(L_1,C_1,L_2,C_2)$ such that
\begin{itemize}
  \item $L_i$ is an idempotent loop, for both $i=1,2$,
  \item $C_i$ is a component of $L_i$, for both $i=1,2$,
  \item $\an{C_1} \leqtime \an{C_2}$,
  \item $\an{C_i}=(x_i,y_i)$, for both $i=1,2$, and $x_1\ge x_2$,
  \item both $\out{\tr{C_1}}$ and $\out{\tr{C_2}}$ are non-empty.
\end{itemize}
\end{definition}

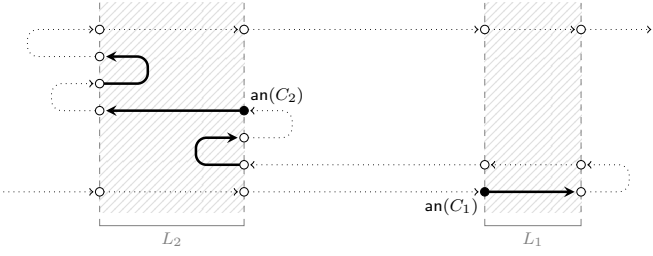
\begin{figure}[!t]
\centering
\scalebox{0.8}{%
\begin{tikzpicture}[baseline=0, inner sep=0, outer sep=0, minimum size=0pt, xscale=0.4, yscale=0.45]
  \tikzstyle{dot} = [draw, circle, fill=white, minimum size=4pt]
  \tikzstyle{fulldot} = [draw, circle, fill=black, minimum size=4pt]
  \tikzstyle{factor} = [->, shorten >=1pt, dotted, rounded corners=5]
  \tikzstyle{fullfactor} = [->, >=stealth, shorten >=1pt, very thick, rounded corners=5]

  \fill [pattern=north east lines, pattern color=gray!25]
        (4,-0.75) rectangle (10,7);
  \draw [dashed, thin, gray] (4,-0.75) -- (4,7);
  \draw [dashed, thin, gray] (10,-0.75) -- (10,7);
  \draw [gray] (4,-1) -- (4,-1.25) -- (10,-1.25) -- (10,-1);
  \draw [gray] (7,-1.5) node [below] {\footnotesize $L_2$};

  \fill [pattern=north east lines, pattern color=gray!25]
        (20,-0.75) rectangle (24,7);
  \draw [dashed, thin, gray] (20,-0.75) -- (20,7);
  \draw [dashed, thin, gray] (24,-0.75) -- (24,7);
  \draw [gray] (20,-1) -- (20,-1.25) -- (24,-1.25) -- (24,-1);
  \draw [gray] (22,-1.5) node [below] {\footnotesize $L_1$};

  \draw (0,0) node (node0) {};
  \draw (4,0) node [dot] (node1) {};
  \draw (10,0) node [dot] (node2) {};
  \draw (20,0) node [fulldot] (node3) {};

  \draw (24,0) node [dot] (node4) {};
  \draw (26,0) node (node5) {};
  \draw (26,1) node (node6) {};
  \draw (24,1) node [dot] (node7) {};
  \draw (20,1) node [dot] (node8) {};

  \draw (10,1) node [dot] (node9) {};
  \draw (8,1) node (node10) {};
  \draw (8,2) node (node11) {};
  \draw (10,2) node [dot] (node12) {};
  \draw (12,2) node (node13) {};
  \draw (12,3) node (node14) {};
  \draw (10,3) node [fulldot] (node15) {};
  \draw (4,3) node [dot] (node16) {};
  \draw (2,3) node (node17) {};
  \draw (2,4) node (node18) {};
  \draw (4,4) node [dot] (node19) {};
  \draw (6,4) node (node20) {};
  \draw (6,5) node (node21) {};
  \draw (4,5) node [dot] (node22) {};

  \draw (1,5) node (node23) {};
  \draw (1,6) node (node24) {};
  \draw (4,6) node [dot] (node25) {};
  \draw (10,6) node [dot] (node26) {};
  \draw (20,6) node [dot] (node27) {};
  \draw (24,6) node [dot] (node28) {};
  \draw (27,6) node (node29) {};

  \draw [factor] (node0) -- (node1);
  \draw [factor] (node1) -- (node2);
  \draw [factor] (node2) -- (node3);
  \draw [fullfactor] (node3) -- (node4);

  \draw [factor] (node4) -- (node5.center) -- (node6.center) -- (node7);
  \draw [factor] (node7) -- (node8);
  \draw [factor] (node8) -- (node9);
  \draw [fullfactor] (node9) -- (node10.center) -- (node11.center) -- (node12);
  \draw [factor] (node12) -- (node13.center) -- (node14.center) -- (node15);
  \draw [fullfactor] (node15) -- (node16);
  \draw [factor] (node16) -- (node17.center) -- (node18.center) -- (node19);
  \draw [fullfactor] (node19) -- (node20.center) -- (node21.center) -- (node22);
  
  \draw [factor] (node22) -- (node23.center) -- (node24.center) -- (node25);
  \draw [factor] (node25) -- (node26);
  \draw [factor] (node26) -- (node27);
  \draw [factor] (node27) -- (node28);
  \draw [factor] (node28) -- (node29);
  
  \draw (node3) node [below left=1.5mm] {\footnotesize $\an{C_1}$};
  \draw (node15) node [above right=1.5mm] {\footnotesize $\an{C_2}$};
\end{tikzpicture}
}
\caption{An inversion with components intercepting the highlighted factors.}\label{fig:inversion}
\end{figure}

\noindent
Fig.~\ref{fig:inversion} gives an example of an inversion involving
the loop $L_1$ with its first component and the loop $L_2$ with its
second component (we highlighted the anchors and the factors
corresponding to these components).

\begin{definition}\label{def:period}
A word $w=a_1 \cdots a_n$ has \emph{period} $p$ if $a_i = a_{i+p}$ 
for all pairs of positions $i,i+p$ of $w$. 
\end{definition}

\noindent
For example, $w = abc\, abc\, ab$ has period $3$. 

One-way definability of functional two-way transducers essentially 
amounts to showing that the output produced by every inversion has bounded
period.
The proposition below shows a slightly stronger periodicity property, 
which refers to the output produced inside the inversion extended on 
both sides by the trace outputs.
We will need this stronger property later, when dealing with 
overlapping portions of the run delimited by different inversions.  

\begin{proposition}\label{prop:inversion}
If $\cT$ is one-way definable, then for every inversion $(L_1,C_1,L_2,C_2)$ 
of a successful run $\rho$ of $\cT$, the word
\[
  \outb{\tr{C_1}} ~ \outb{\rho[\an{C_1},\an{C_2}]} ~ \outb{\tr{C_2}}
\]
has period $p$ that divides both $|\out{\tr{C_1}}|$ and $|\out{\tr{C_2}}|$.
Moreover, $p \le \bound$.
\end{proposition}

The basic combinatorial argument for proving
Proposition~\ref{prop:inversion} is a classical result in word
combinatorics called Fine and Wilf's theorem~\cite{fine-wilf}.
Essentially, the theorem says that, whenever two periodic words
$w_1,w_2$ share a sufficiently long factor, then they have as period
the greatest common divisor of the two original periods.  Below, we
state a slightly stronger variant of Fine-Wilf's theorem, which
contains an additional claim showing how to align a common factor
of the words $w_1,w_2$ so as to form a third word $w_3$ that contains
a prefix of $w_1$ and a suffix of $w_2$.  The additional claim will be
fully exploited in the proof of
Proposition~\ref{prop:block-periodicity}.

\begin{lemma}[Fine-Wilf's theorem]\label{lemma:fine-wilf}
If $w_1 = w'_1\,w\:w''_1$ has period $p_1$, 
$w_2 = w'_2\,w\,w''_2$ has  period $p_2$, and 
the common factor $w$ has length at least $p_1+p_2-\gcd(p_1,p_2)$, 
then $w_1$, $w_2$, and $w_3 = w'_1\,w\,w''_2$ have period $\gcd(p_1,p_2)$.
\end{lemma}

Two further combinatorial results are heavily used in the proof
of Proposition \ref{prop:inversion}. 
The first one is a result of Kortelainen~\cite{kortelainen98}, 
which was later improved and 
simplified by Saarela \cite{saarela15}. It is related to word 
equations with iterated factors, like those that arise from 
considering outputs of pumped versions of a run.
To improve readability, we highlight the important 
iterations of factors inside the considered equations.

\begin{theorem}[Theorem 4.3 in \cite{saarela15}]\label{thm:saarela}
Consider a word equation
\[
  v_0 \: \pmb{v_1^m} \: v_2 \: ... \: v_{k-1} \: \pmb{v_k^m} \: v_{k+1} 
  ~=~
  w_0 \: \pmb{w_1^m} \: w_2 \: ... \: w_{k'-1} \: \pmb{w_{k'}^m} \: w_{k'+1}
\]
where $m$ is the unknown and $v_i,w_j$ are words.
Then the set of solutions of the equation is either finite or $\bbN$.
\end{theorem}

The second combinatorial result considers a word equation with
iterated factors parametrized by two unknowns $m_1,m_2$ that
occur in opposite order in the left, respectively right 
hand-side of the equation. This type of equation 
arises when we compare the output associated with an inversion of 
$\cT$ and the output produced by an equivalent one-way 
transducer $\cT'$.

\begin{lemma}\label{lemma:one-way-vs-two-way}
Consider a word equation of the form
\[
  v_0^{(m_1,m_2)} \: \pmb{v_1^{m_1}} \: v_2^{(m_1,m_2)} \: \pmb{v_3^{m_2}} \: v_4^{(m_1,m_2)}
  ~=~
  w_0 \: \pmb{w_1^{m_2}} \: w_2 \: \pmb{w_3^{m_1}} \: w_4
\]
where $m_1,m_2$ are the unknowns, $v_1,v_3$ are non-empty words,
and $v_0^{(m_1,m_2)},v_2^{(m_1,m_2)},v_4^{(m_1,m_2)}$ are words
that may contain factors of the form $v^{m_1}$ or $v^{m_2}$,
for a generic word $v$.
If the above equation holds for all $m_1,m_2\in\bbN$, 
then the words
$\pmb{v_1} ~ \pmb{v_1^{m_1}} ~ v_2^{(m_1,m_2)} ~ \pmb{v_3^{m_2}} ~ \pmb{v_3}$
are periodic with period $\gcd(|v_1|,|v_3|)$, for all $m_1,m_2\in\bbN$.
\end{lemma}

The last ingredient used in the proof of Proposition \ref{prop:inversion}
is a bound on the period of the output produced 
by an inversion. 
For this, we introduce a suitable notion of minimality of loops 
and loop components:

\begin{definition}\label{def:output-minimal}
Consider pairs $(L,C)$ consisting of an idempotent loop $L$
and a component $C$ of $L$.
\begin{itemize}
\item On such pairs, we define the
relation 
$\lesspair$ by  $(L',C') \lesspair (L,C)$ if
$L'\subsetneq L$ and at least one $(L',C')$-factor
is contained in some $(L,C)$-factor.
\item A pair $(L,C)$ is \emph{output-minimal} if for all 
pairs $(L',C') \lesspair (L,C)$, we have $\out{\tr{C'}}=\emptystr$.
\end{itemize}
\end{definition}

\noindent
Note that the relation $\lesspair$ is not a partial order in general
(it is however antisymmetric).
Lemma~\ref{lem:output-minimal} below shows that the length of the
output trace of $C$ inside $L$ is bounded whenever $(L,C)$ is
output-minimal.

\begin{lemma}\label{lem:output-minimal}
For every output-minimal pair $(L,C)$, $|\out{\tr{C}}| \le \bound$.
\end{lemma}

\begin{proof}[Proof sketch]
We use a Ramsey-type argument here: if $|\out{\tr{C}}| > \bound$,
then Theorem~\ref{thm:simon2} can be applied to exhibit an idempotent loop
strictly inside $L$ and a component $C$ of it with non-empty trace output.
This would contradict the output-minimality of $(L,C)$.
\end{proof}

We remark that the above lemma cannot be used directly to bound the 
period of the output produced by an inversion.
The reason is that we cannot assume that inversions are built up 
from output-minimal pairs.
\begin{wrapfigure}{r}{4.2cm}
\vspace{-1mm}\hspace{-2.25mm}
\scalebox{0.8}{%
\begin{tikzpicture}[baseline=0, inner sep=0, outer sep=0, minimum size=0pt, scale=0.3, xscale=0.68,yscale=1.2]
  \tikzstyle{dot} = [draw, circle, fill=white, minimum size=4pt]
  \tikzstyle{fulldot} = [draw, circle, fill=black, minimum size=4pt]
  \tikzstyle{factor} = [->, shorten >=1pt, dotted, rounded corners=5]
  \tikzstyle{fullfactor} = [->, >=stealth, shorten >=1pt, very thick, rounded corners=5]

  \fill [pattern=north east lines, pattern color=gray!25]
        (4,-0.75) rectangle (10,9);
  \draw [dashed, thin, gray] (4,-0.75) -- (4,9);
  \draw [dashed, thin, gray] (10,-0.75) -- (10,9);
  \draw [gray] (4,-1) -- (4,-1.25) -- (10,-1.25) -- (10,-1);
  \draw [gray] (7,-1.5) node [below] {\footnotesize $L_2$};

  \fill [pattern=north east lines, pattern color=gray!25]
        (15,-0.75) rectangle (21,9);
  \draw [dashed, thin, gray] (15,-0.75) -- (15,9);
  \draw [dashed, thin, gray] (21,-0.75) -- (21,9);
  \draw [gray] (15,-1) -- (15,-1.25) -- (21,-1.25) -- (21,-1);
  \draw [gray] (18,-1.5) node [below] {\footnotesize $L_1$};

  \draw (1,0) node (node0) {};
  \draw (4,0) node [dot] (node1) {};
  \draw (10,0) node [dot] (node2) {};
  \draw (15,0) node [dot] (node3) {};
  \draw (21,0) node [dot] (node4) {};
  \draw (23,0) node (node5) {};
  \draw (23,1) node (node6) {};

  \draw (21,1) node [dot] (node7) {};
  \draw (18,1) node (node8) {};
  \draw (18,2) node (node9) {};
  \draw (21,2) node (node10) {};
  \draw (24,2) node (node11) {};
  \draw (24,3) node (node12) {};
  \draw (21,3) node [fulldot] (node13) {};
  \draw (15,3) node [dot] (node14) {};

  \draw (10,3) node [dot] (node15) {};
  \draw (4,3) node [dot] (node16) {};
  \draw (2,3) node (node17) {};
  \draw (2,4) node (node18) {};

  \draw (4,4) node [dot] (node19) {};
  \draw (7,4) node (node20) {};
  \draw (7,5) node (node21) {};
  \draw (4,5) node [dot] (node22) {};
  \draw (1,5) node (node23) {};
  \draw (1,6) node (node24) {};
  \draw (4,6) node [fulldot] (node25) {};
  \draw (10,6) node [dot] (node26) {};

  \draw (15,6) node [dot] (node27) {};
  \draw (19,6) node (node28) {};
  \draw (19,7) node (node29) {};
  \draw (15,7) node [dot] (node30) {};

  \draw (10,7) node [dot] (node31) {};
  \draw (6,7) node (node32) {};
  \draw (6,8) node (node33) {};
  \draw (10,8) node [dot] (node34) {};

  \draw (15,8) node [dot] (node35) {};
  \draw (21,8) node [dot] (node36) {};
  \draw (24,8) node (node37) {};

  \draw [factor] (node0) -- (node1);
  \draw [factor] (node1) -- (node2);
  \draw [factor] (node2) -- (node3);
  \draw [factor] (node3) -- (node4);
  \draw [factor] (node4) -- (node5.center) -- (node6.center) -- (node7);
  
  \draw [fullfactor] (node7) -- (node8.center) -- (node9.center) -- (node10);
  \draw [factor] (node10) -- (node11.center) -- (node12.center) -- (node13);
  \draw [fullfactor] (node13) -- (node14);

  \draw [factor] (node14) -- (node15);
  \draw [factor] (node15) -- (node16);
  \draw [factor] (node16) -- (node17.center) -- (node18.center) -- (node19);

  \draw [fullfactor] (node19) -- (node20.center) -- (node21.center) -- (node22);
  \draw [factor] (node22) -- (node23.center) -- (node24.center) -- (node25);
  \draw [fullfactor] (node25) -- (node26);
  
  \draw [factor] (node26) -- (node27);
  \draw [fullfactor, red] (node27) -- (node28.center) -- (node29.center) -- (node30);
  \draw [factor] (node30) -- (node31);
  \draw [fullfactor, red] (node31) -- (node32.center) -- (node33.center) -- (node34);

  \draw [factor] (node34) -- (node35);
  \draw [factor] (node35) -- (node36);
  \draw [factor] (node36) -- (node37);
  
  \draw (node13) node [above right=1mm] {\footnotesize $\an{C_1}$};
  \draw (node25) node [above left=1mm] {\footnotesize $\an{C_2}$};
\end{tikzpicture}
}
\hspace{-3mm}\vspace{-2mm}
\end{wrapfigure} 
A counter-example is given in 
the figure to the right,
which shows a run where 
the only inversion $(L_1,C_1,L_2,C_2)$ contains pairs that are 
not output-minimal:
the factors that produce long outputs are those in red,
but they occur outside $\rho[\an{C_1},\an{C_2}]$.

\medskip
We are now ready to prove Proposition~\ref{prop:inversion}.
Here we only present the key ideas, and refer the reader 
to the appendix for more details.

\begin{proof}[Proof sketch of Proposition~\ref{prop:inversion}]
In the first half of the proof we pump the two loops $L_1$ and $L_2$
so that we obtain also loops in the assumed equivalent one-way
transducer $\cT'$. 
We then consider the outputs of the pumped runs
of $\cT$ and $\cT'$, which contain iterated factors
parametrized by two natural numbers $m_1,m_2$.
As those outputs must agree due to the equivalence of $\Tt,\Tt'$, we get
an equation as in Lemma~\ref{lemma:one-way-vs-two-way},
where the word $v_1$
belongs to $\out{\tr{C_1}}^+$ 
and the word $v_3$ 
belongs to $\out{\tr{C_2}}^{+}$.
Lemma~\ref{lemma:one-way-vs-two-way} shows that the word described by
the equation has  period $p$ dividing $\gcd(|v_1|,|v_3|)$, 
and Lemma~\ref{lemma:fine-wilf} shows that 
$p$ even divides $|\out{\tr{C_1}}|$ and $|\out{\tr{C_2}}|$. 
Finally, we use Theorem~\ref{thm:saarela} to transfer the periodicity
property from the word of the equation
to the word $w \:=\: \out{\tr{C_1}} \: \out{\rho[\an{C_1},\an{C_2}]} \: \out{\tr{C_2}}$
produced by the original run of $\cT$.
This is possible because the word of the equation is obtained
by iterating factors of $w$. In particular, by reasoning separately 
on the parameters that define those iterations, and by stating the 
periodicity property as an equation in the form required by 
Theorem~\ref{thm:saarela}, one can prove that the periodicity 
equation holds on all parameters, and thus in particular on $w$.

In the second half of the proof we show that the period $p$ is bounded
by $\bound$. This requires a refinement of the previous arguments and involves
pumping the run of $\Tt$ simultaneously on three different loops. The
idea is that by pumping we manage to find inversions with
some output-minimal pair $(L_0,C_0)$. In this way we show that the
period $p$ also divides $\out{\tr{C_0}}$, which is bounded by $\bound$
according to Lemma~\ref{lem:output-minimal}.
\end{proof}

\section{One-way definability}\label{sec:decomposition}

Proposition~\ref{prop:inversion} is the main combinatorial argument for characterizing
two-way transducers that are one-way definable. In this section we provide the remaining
arguments. Roughly, the idea is to decompose every successful run $\rho$ into factors 
that produce long outputs either in a left-to-right manner (``diagonals''), 
or based on an almost periodic pattern (``blocks''). 

We say that a word $w$ is \emph{almost periodic with bound $p$} if $w=w_0~w_1~w_2$ 
for some words $w_0,w_2$ of length at most $p$ and some word $w_1$ 
of period at most $p$. 

We illustrate the following definition in Fig.~\ref{fig:factors}.

\begin{definition}\label{def:factors}
Consider a factor $\rho[\ell_1,\ell_2]$ of the run, where
$\ell_1=(x_1,y_1)$, $\ell_2=(x_2,y_2)$, and $x_1\le x_2$.
We call $\rho[\ell_1,\ell_2]$
\begin{itemize}
  \item a \emph{diagonal} 
        if for all $x\in[x_1,x_2]$, there is a location $\ell_x$ at position $x$
        such that $\ell_1\leqtime\ell_x\leqtime\ell_2$ and the words
        $\out{\rho\mid Z_{\ell_x}^\upperleft}$ and 
        $\out{\rho\mid Z_{\ell_x}^\lowerright}$ 
        have length at most $\bound$,
        where $Z_{\ell_x}^\upperleft = [{\ell_x},\ell_2] \:\cap\: \big([0,x]\times\bbN\big)$
        and $Z_{\ell_x}^\lowerright = [\ell_1,{\ell_x}] \:\cap\: \big([x,\omega]\times\bbN\big)$;
  \item a \emph{block} if the word
        $\out{\rho[\ell_1,\ell_2]}$ is almost periodic with bound $\bound$, 
        and 
        $\out{\rho\mid Z^\leftshort}$ and 
        $\out{\rho\mid Z^\rightshort}$ have length at most $\bound$,
        where $Z^\leftshort = [\ell_1,\ell_2] \:\cap\: \big([0,x_1]\times \bbN\big)$
        and $Z^\rightshort = [\ell_1,\ell_2] \:\cap\: \big([x_2,\omega]\times \bbN\big)$.
\end{itemize}
\end{definition}

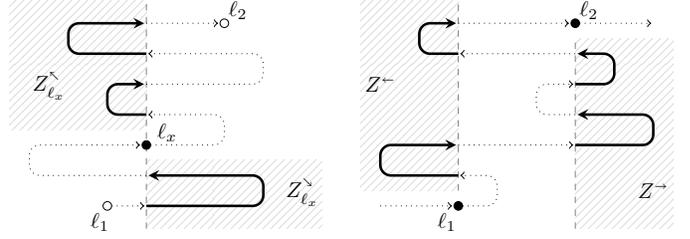
\begin{figure}[!t]
\centering
\scalebox{0.81}{%
\begin{tikzpicture}[baseline=0, inner sep=0, outer sep=0, minimum size=0pt, xscale=0.32, yscale=0.5]
\begin{scope}
  \tikzstyle{dot} = [draw, circle, fill=white, minimum size=4pt]
  \tikzstyle{fulldot} = [draw, circle, fill=black, minimum size=4pt]
  \tikzstyle{factor} = [->, shorten >=1pt, dotted, rounded corners=5]
  \tikzstyle{fullfactor} = [->, >=stealth, shorten >=1pt, very thick, rounded corners=5]

  \fill [pattern=north east lines, pattern color=gray!25]
        (1,6.75) rectangle (8,2.5);
  \fill [pattern=north east lines, pattern color=gray!25]
        (8,1.5) rectangle (17,-0.75);
  \draw [dashed, thin, gray] (8,-0.75) -- (8,6.75);
 
  \draw (6,0) node [dot] (node0) {};
  \draw (8,0) node (node1) {};
  \draw (14,0) node (node2) {};
  \draw (14,1) node (node3) {};
  \draw (8,1) node (node4) {};
  \draw (2,1) node (node5) {};
  \draw (2,2) node (node6) {};
  \draw (8,2) node [fulldot] (node7) {};
  \draw (12,2) node (node8) {};
  \draw (12,3) node (node9) {};
  \draw (8,3) node (node10) {};
  \draw (6,3) node (node11) {};
  \draw (6,4) node (node12) {};
  \draw (8,4) node (node13) {};
  \draw (14,4) node (node14) {};
  \draw (14,5) node (node15) {};
  \draw (8,5) node (node16) {};
  \draw (4,5) node (node17) {};
  \draw (4,6) node (node18) {};
  \draw (8,6) node (node19) {};
  \draw (12,6) node [dot] (node20) {};

  \draw [factor] (node0) -- (node1); 
  \draw [fullfactor] (node1) -- (node2.center) -- (node3.center) -- (node4); 
  \draw [factor] (node4) -- (node5.center) -- (node6.center) -- (node7); 
  \draw [factor] (node7) -- (node8.center) -- (node9.center) -- (node10);
  \draw [fullfactor] (node10) -- (node11.center) -- (node12.center) -- (node13); 
  \draw [factor] (node13) -- (node14.center) -- (node15.center) -- (node16); 
  \draw [fullfactor] (node16) -- (node17.center) -- (node18.center) -- (node19); 
  \draw [factor] (node19) -- (node20);
  
  \draw (node0) node [below = 1.2mm] {$\ell_1~~$};
  \draw (node7) node [above right = 0.8mm] {$~\ell_x$};
  \draw (node20) node [above right = 1mm] {$\ell_2$};

  \draw (16,0.5) node {$Z_{\ell_x}^\lowerright$};
  \draw (3,4) node {$Z_{\ell_x}^\upperleft$};
\end{scope}
\begin{scope}[xshift=18cm]
  \tikzstyle{dot} = [draw, circle, fill=white, minimum size=4pt]
  \tikzstyle{fulldot} = [draw, circle, fill=black, minimum size=4pt]
  \tikzstyle{factor} = [->, shorten >=1pt, dotted, rounded corners=5]
  \tikzstyle{fullfactor} = [->, >=stealth, shorten >=1pt, very thick, rounded corners=5]

  \fill [pattern=north east lines, pattern color=gray!25]
        (1,0.5) rectangle (6,6.75);
  \fill [pattern=north east lines, pattern color=gray!25]
        (12,-0.75) rectangle (17,5.5);
  \draw [dashed, thin, gray] (6,0.5) -- (6,6.75);
  \draw [dashed, thin, gray] (12,-0.75) -- (12,5.5);
 
  \draw (2,0) node (node0) {} ;
  \draw (6,0) node [fulldot] (node1) {};
  \draw (8,0) node (node2) {};
  \draw (8,1) node (node3) {};
  \draw (6,1) node (node4) {};
  \draw (2,1) node (node5) {};
  \draw (2,2) node (node6) {};
  \draw (6,2) node (node7) {};
  \draw (12,2) node (node8) {};
  \draw (16,2) node (node9) {};
  \draw (16,3) node (node10) {};
  \draw (12,3) node (node11) {};
  \draw (10,3) node (node12) {};
  \draw (10,4) node (node13) {};
  \draw (12,4) node (node14) {};
  \draw (14,4) node (node15) {};
  \draw (14,5) node (node16) {};
  \draw (12,5) node (node17) {};
  \draw (6,5) node (node18) {};
  \draw (4,5) node (node19) {};
  \draw (4,6) node (node20) {};
  \draw (6,6) node (node21) {};
  \draw (12,6) node [fulldot] (node22) {};
  \draw (16,6) node (node23) {} ;

	\draw [factor] (node0) -- (node1) ;
  \draw [factor] (node1) -- (node2.center) -- (node3.center) -- (node4); 
  \draw [fullfactor] (node4) -- (node5.center) -- (node6.center) -- (node7); 
  \draw [factor] (node7) -- (node8);
  \draw [fullfactor] (node8) -- (node9.center) -- (node10.center) -- (node11); 
  \draw [factor] (node11) -- (node12.center) -- (node13.center) -- (node14); 
  \draw [fullfactor] (node14) -- (node15.center) -- (node16.center) -- (node17); 
  \draw [factor] (node17) -- (node18);
  \draw [fullfactor] (node18) -- (node19.center) -- (node20.center) -- (node21); 
  \draw [factor] (node21) -- (node22);
  
  \draw (node1) node [below = 1.2mm] {$\ell_1~~~$};
  \draw (node22) node [above right = 1mm] {$\ell_2$};

  \draw (2,4) node {$Z^\leftshort$};
  \draw (16,0.5) node {$Z^\rightshort$};
  \draw [factor] (node22) -- (node23) ;
  
\end{scope}
\end{tikzpicture}
}
\caption{Outputs that need to be bounded in a diagonal and in a block.}\label{fig:factors}
\end{figure}

\noindent
Intuitively, the output of a diagonal $\rho[\ell_1,\ell_2]$ 
can be 
simulated 
while scanning the input interval $[x_1,x_2]$ from left to right, 
since the outputs of $\rho\mid Z_{\ell_x}^\upperleft$ and 
$\rho\mid Z_{\ell_x}^\lowerright$ are bounded. 
A similar argument applies to a block $\rho[\ell_1,\ell_2]$, 
where in addition, one exploits the fact that the output is almost periodic. 
Roughly, the idea is that one can simulate the output of a block
by outputting symbols according to a periodic pattern, and in a 
number that is determined from the transitions on $u[x_1,x_2]$ and the
guessed (bounded) outputs on $Z^\leftshort$ and $Z^\rightshort$.

The general idea for turning a two-way transducer $\cT$ into an equivalent 
one-way transducer $\cT'$ is to guess (and check) a factorization of a
successful run of $\cT$ into factors that are either diagonals or blocks,
and properly arranged following the order of positions.

\begin{definition}\label{def:decomposition}
A \emph{decomposition} of $\rho$ is a factorization $\prod_i\,\rho[\ell_i,\ell_{i+1}]$
of $\rho$ into diagonals and blocks, where $\ell_i=(x_i,y_i)$ and $x_i < x_{i+1}$ for all $i$.
\end{definition}

The one-way transducer $\cT'$ whose existence is stated by Theorem \ref{th:main}
simulates $\cT$ precisely on those inputs $u$ that have some successful
run admitting a decomposition.
To provide further intuition on the notion of decomposition,
we consider again the transduction of Example~\ref{ex:running} and the 
two-way transducer $\cT$ that implements it in the most natural way.
Fig.~\ref{fig:decomposition} shows an example of a run of $\cT$ on
an input of the form $u_1 \:\#\: u_2 \:\#\: u_3 \:\#\: u_4$, where
$u_2, u_4 \in (abc)^*$, $u_1\,u_3\nin (abc)^*$, and $u_3$ has even length. 
The factors of the run that produce long outputs are highlighted 
by the bold arrows. The first and third factors of the decomposition, 
i.e.~$\rho[\ell_1,\ell_2]$ and $\rho[\ell_3,\ell_4]$, are diagonals
(represented by the blue hatched areas); the second and fourth factors 
$\rho[\ell_2,\ell_3]$ and $\rho[\ell_4,\ell_5]$ are blocks
(represented by the red hatched areas). 

\begin{figure}[!t]
\centering
\scalebox{1}{
\begin{tikzpicture}[baseline=0, inner sep=0, outer sep=0, minimum size=0pt, scale=0.32]
  \tikzstyle{dot} = [draw, circle, fill=white, minimum size=4pt]
  \tikzstyle{fulldot} = [draw, circle, fill=black, minimum size=4pt]
  \tikzstyle{factor} = [->, shorten >=1pt, dotted, rounded corners=5]
  \tikzstyle{fullfactor} = [->, >=stealth, shorten >=1pt, very thick, rounded corners=5]

  \fill [pattern=north east lines, pattern color=blue!18]
        (0.5,-.5) rectangle (6.5,2.5);
  \fill [pattern=north east lines, pattern color=red!20]
        (6.5,1.5) rectangle (12.5,4.5);
  \fill [pattern=north east lines, pattern color=blue!18]
        (12.5,3.5) rectangle (18.5,6.5);
  \fill [pattern=north east lines, pattern color=red!20]
        (18.5,5.5) rectangle (24.5,8.5);
  \draw [dashed, thin, gray] (0.5,-1) -- (0.5,9);
  \draw [dashed, thin, gray] (6.5,-1) -- (6.5,9);
  \draw [dashed, thin, gray] (12.5,-1) -- (12.5,9);
  \draw [dashed, thin, gray] (18.5,-1) -- (18.5,9);
  \draw [dashed, thin, gray] (24.5,-1) -- (24.5,9);
  \draw [gray] (1,-1.25) -- (1,-1.5) -- (6,-1.5) -- (6,-1.25);
  \draw [gray] (3.5,-2) node [below] {\small $u_1$};
  \draw [gray] (6.5,-1.75) node [below] {\footnotesize $\#$};
  \draw [gray] (7,-1.25) -- (7,-1.5) -- (12,-1.5) -- (12,-1.25);
  \draw [gray] (9.5,-2) node [below] {\small $u_2$};
  \draw [gray] (12.5,-1.75) node [below] {\footnotesize $\#$};
  \draw [gray] (13,-1.25) -- (13,-1.5) -- (18,-1.5) -- (18,-1.25);
  \draw [gray] (15.5,-2) node [below] {\small $u_3$};
  \draw [gray] (18.5,-1.75) node [below] {\footnotesize $\#$};
  \draw [gray] (19,-1.25) -- (19,-1.5) -- (24,-1.5) -- (24,-1.25);
  \draw [gray] (21.5,-2) node [below] {\small $u_4$};

  \draw (0.5,0) node [dot] (node0) {};  

  \draw (2,0) node (node1) {};
  \draw (5.5,0) node (node2) {};
  \draw (12,0) node (node3) {};
  \draw (12,1) node (node4) {};
  \draw (5,1) node (node5) {};
  \draw (5,2) node (node6) {};

  \draw (6.5,2) node [dot] (node10') {};

  \draw (8,2) node (node11) {};
  \draw (11.5,2) node (node12) {};
  \draw (18,2) node (node13) {};
  \draw (18,3) node (node14) {};
  \draw (17,3) node (node15) {};
  \draw (8,3) node (node16) {};
  \draw (7,3) node (node17) {};
  \draw (7,4) node (node18) {};
  \draw (8,4) node (node19) {};
  \draw (11.5,4) node (node20) {};

  \draw (12.5,4) node [dot] (node20') {};

  \draw (14,4) node (node21) {};
  \draw (17.5,4) node (node22) {};
  \draw (24,4) node (node23) {};
  \draw (24,5) node (node24) {};
  \draw (17,5) node (node25) {};
  \draw (17,6) node (node26) {};

  \draw (18.5,6) node [dot] (node30') {};

  \draw (20,6) node (node31) {};
  \draw (23.5,6) node (node32) {};
  \draw (24,6) node (node33) {};
  \draw (24,7) node (node34) {};
  \draw (23,7) node (node35) {};
  \draw (20,7) node (node36) {};
  \draw (19,7) node (node37) {};
  \draw (19,8) node (node38) {};
  \draw (20,8) node (node39) {};
  \draw (23.5,8) node (node40) {};

  \draw (24.5,8) node [dot] (node40') {};

  \draw [factor] (node0) -- (node1);
  \draw [fullfactor] (node1) -- (node2);
  \draw [factor] (node2) -- (node3.center) -- (node4.center) 
                         -- (node5.center) -- (node6.center) -- (node10');

  \draw [factor] (node10') -- (node11);
  \draw [fullfactor] (node11) -- (node12);
  \draw [factor] (node12) -- (node13.center) -- (node14.center) -- (node15) -- 
                 (node15) -- (node16) --
                 (node16) -- (node17.center) -- (node18.center) -- (node19); 
  \draw [fullfactor] (node19) -- (node20); 
  \draw [factor] (node20) -- (node20');

  \draw [factor] (node20') -- (node21);
  \draw [fullfactor] (node21) -- (node22);
  \draw [factor] (node22) -- (node23.center) -- (node24.center) 
                          -- (node25.center) -- (node26.center) -- (node30');

  \draw [factor] (node30') -- (node31);
  \draw [fullfactor] (node31) -- (node32);
  \draw [factor] (node32) -- (node33.center) -- (node34.center) -- (node35) --
                 (node35) -- (node36) -- 
                 (node36) -- (node37.center) -- (node38.center) -- (node39);
  \draw [fullfactor] (node39) -- (node40);
  \draw [factor] (node40) -- (node40');

  \draw (node0) node [left = 1.2mm] {$\ell_1$};
  \draw (node10') node [above left = 1mm] {$\ell_2$};
  \draw (node20') node [above right = 1mm] {$\ell_3$};
  \draw (node30') node [above left = 1mm] {$\ell_4$};
  \draw (node40') node [right = 1.2mm] {$\ell_5$};
\end{tikzpicture}
}
\caption{A decomposition of a run of a two-way transducer.}\label{fig:decomposition}
\end{figure}
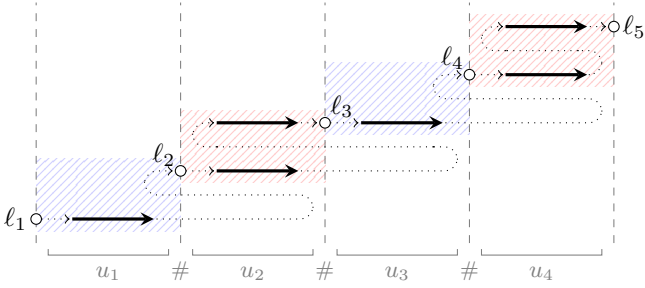

\begin{theorem}\label{thm:main2}
Let $\cT$ be a functional two-way transducer. 
The following are equivalent:
\begin{itemize}
  \item[\PR1)] $\cT$ is one-way definable.
  \item[\PR2)] For all inversions $(L_1,C_1,L_2,C_2)$ of all 
               successful runs of $\cT$, the word
               \[
                 \outb{\tr{C_1}} ~ \outb{\rho[\an{C_1},\an{C_2}]} ~ \outb{\tr{C_2}}
               \]
               has period $p \le \bound$ dividing $|\out{\tr{C_1}}|$, $|\out{\tr{C_2}}|$.
  \item[\PR3)] Every successful run of $\cT$ admits a decomposition. 
\end{itemize}
\end{theorem}

The implication from \PR1 to \PR2 was already shown 
in Proposition~\ref{prop:inversion}.  
The rest of this section is devoted to prove the 
implications from \PR2 to \PR3 and from \PR3 to \PR1.
The issues related to the complexity of the characterization 
will be discussed further below.

\paragraph{From periodicity to existence of decompositions (\PR2$\then$\PR3).}
As usual, we fix a successful run $\rho$ of $\cT$.
We will prove a slightly stronger result than the
implication from \PR2 to \PR3, namely: if every inversion of 
$\rho$ satisfies the periodicity property stated in \PR2,
then $\rho$ admits a decomposition
(note that this is independent of whether other runs satisfy or not \PR2). 
To identify the blocks of a possible decomposition 
of $\rho$ we consider a suitable equivalence relation 
between locations:

\begin{definition}\label{def:equivalence}
A location $\ell$ is \emph{covered} by an inversion $(L_1,C_1,L_2,C_2)$ 
if $\an{C_1} \leqtime \ell \leqtime \an{C_2}$.
We define the relation $\crossrel$ by letting $\ell\crossrel\ell'$ 
if $\ell,\ell'$ are covered by the {\sl same} inversion. 
We define the equivalence relation $\simeq$ as the reflexive and transitive closure of $\crossrel$.
\end{definition}

\noindent
Locations covered by the same inversion $(L_1,C_1,L_2,C_2)$ 
yield an interval w.r.t.~the run ordering $\leqtime$. 
Thus every non-singleton $\simeq$-class can be seen as 
a union of such intervals, say $K_1,\dots,K_m$, 
that are two-by-two overlapping, namely, $K_i \cap K_{i+1} \neq \emptyset$ for all $i<m$. 
In particular, a non-singleton $\simeq$-class is an interval of locations 
witnessed by a series of inversions $(L_{2i},C_{2i},L_{2i+1},C_{2i+1})$ such that
$\an{C_{2i}} \leqtime \an{C_{2i+2}} \leqtime \an{C_{2i+1}} \leqtime \an{C_{2i+3}}$.

The next result exploits the shape of a non-singleton $\simeq$-class, 
the assumption that $\rho$ satisfies the periodicity property stated in \PR2, 
and Lemma~\ref{lemma:fine-wilf}, to show that the output produced inside 
an $\simeq$-class has bounded period.

\begin{proposition}\label{prop:block-periodicity}
If $\rho$ satisfies the periodicity property stated in \PR2
and $\ell \leqtime \ell'$ are two locations in the 
same $\simeq$-class, then $\outb{\rho[\ell,\ell']}$ 
has period at most $\bound$.
\end{proposition}

The $\simeq$-classes considered so far cannot be directly used as blocks
for the desired decomposition of $\rho$, since the $x$-coordinates of their 
endpoints might not be in the appropriate order. The next definition
takes care of this, by enlarging the $\simeq$-classes according to
$x$-coordinates of anchors.

\begin{definition}\label{def:bounding-box}
Let $K=[\ell,\ell']$ be a non-singleton $\simeq$-class, 
let $\an{K}$ be the restriction of $K$ to the locations 
that are anchors of components of inversions,
and let $X_{\an{K}} = \{x \::\: \exists y\: (x,y)\in \an{K}\}$ 
be the projection of $\an{K}$ on positions.
\par\noindent
We define $\block{K}=[\ell_1,\ell_2]$, where 
\begin{itemize}
  \item $\ell_1$ is the latest location $(x,y) \leqtime \ell$ 
        such that $x = \min\big(X_{\an{K}}\big)$, 
  \item $\ell_2$ is the earliest location $(x,y) \geqtime \ell'$ 
        such that $x = \max\big(X_{\an{K}}\big)$ 
\end{itemize}  
(note that the location $\ell_1$ exists since
$\ell$ is the anchor of the first component of an inversion,
and $\ell_2$ exists for similar reasons).
\end{definition}

\begin{lemma}\label{lemma:bounding-box}
If $K=[\ell,\ell']$ is a non-singleton $\simeq$-class, then $\rho[\ell_1,\ell_2]$ is a block,
where $[\ell_1,\ell_2]=\block{K}$.
\end{lemma}

\begin{proof}[Proof sketch]
  The periodicity of $\out{\rho[\ell,\ell']}$ is obtained by applying
  Proposition~\ref{prop:block-periodicity}.
  Then Theorem~\ref{thm:simon2} is applied twice:
  first to bound $\out{\rho[\ell_1,\ell]}$ and $\out{\rho[\ell',\ell_2]}$
  (hence proving that $\out{\rho[\ell_1,\ell_2]}$ is almost periodic
  with bound $\bound$),
  and second, to bound $\out{\rho\mid Z^\leftshort}$ and
  $\out{\rho\mid Z^\rightshort}$, as introduced in Definition~\ref{def:factors}.
\end{proof}

The next lemma shows that blocks do not overlap along the input axis: 

\begin{lemma}\label{lem:consecutive-blocks}
Suppose that $K_1$ and $K_2$ are two different non-singleton $\simeq$-classes
such that $\ell \lesstime \ell'$ for all $\ell \in K_1$ and $\ell' \in
K_2$.
Let $\block{K_1}=[\ell_1,\ell_2]$ and $\block{K_2}=[\ell_3,\ell_4]$,
with $\ell_2=(x_2,y_2)$ and $\ell_3=(x_3,y_3)$. 
Then $x_2 < x_3$.
\end{lemma}

\begin{proof}[Proof sketch]
  If $x_2\geq x_3$, one can exhibit an inversion between a component of a loop
  in $K_1$ and another one in $K_2$, and deduce that $K_1=K_2$.
\end{proof}

For the sake of brevity, we call \emph{$\simeq$-block} any 
factor of the form $\rho\mid\block{K}$ that is obtained by applying 
Definition~\ref{def:bounding-box} to a non-singleton $\simeq$-class $K$.
The results obtained so far imply that every location 
covered by an inversion is also covered by an $\simeq$-block
(Lemma \ref{lemma:bounding-box}), and that the order of occurrence 
of $\simeq$-blocks is the same as the order of positions 
(Lemma \ref{lem:consecutive-blocks}).
So the $\simeq$-blocks can be used as factors for the decomposition
of $\rho$ we are looking for. Below, we show that the remaining 
factors of $\rho$, which do not overlap the $\simeq$-blocks, are diagonals. 
This will complete the construction of a decomposition of $\rho$.

Formally, we say that a factor $\rho[\ell_1,\ell_2]$ \emph{overlaps}
another factor $\rho[\ell_3,\ell_4]$ if 
$[\ell_1,\ell_2] \:\cap\: [\ell_3,\ell_4] \neq \emptyset$, 
$\ell_2 \neq \ell_3$, and $\ell_1 \neq \ell_4$.

\begin{lemma}\label{lem:diagonal}
Let $\rho[\ell_1,\ell_2]$ be a
factor of $\rho$ that does not overlap any $\simeq$-block,
with $\ell_1=(x_1,y_1)$, $\ell_2=(x_2,y_2)$, and $x_1<x_2$.
Then $\rho[\ell_1,\ell_2]$ is a diagonal.
\end{lemma}

\begin{proof}[Proof sketch]
  If $\rho[\ell_1,\ell_2]$ is not a diagonal,
  we can find a location $\ell_1\leqtime \ell \leqtime \ell_2$
  for which $|\out{\rho\mid Z_\ell^\upperleft}| > \bound$ and
  $|\out{\rho\mid Z_\ell^\lowerright}| > \bound$ (recall
  Definition~\ref{def:factors}).
  By applying again  Theorem~\ref{thm:simon2}, 
  we derive the existence of an inversion between $\ell_1$ and
  $\ell_2$, and thus of an $\simeq$-block overlapping
  $\rho[\ell_1,\ell_2]$. 
\end{proof}

\paragraph{From decompositions to one-way definability (\PR3$\then$\PR1).}
Hereafter, we denote by $U$ the language of words $u\in\dom(\cT)$ such that 
{\sl all} successful runs of $\cT$ on $u$ admit a decomposition.

So far, we know that if $\Tt$ is one-way definable (\PR1), 
then $U=\dom(\Tt)$ (\PR3). 
This reduces the one-way definability problem for $\cT$ 
to the containment problem $\dom(\cT) \subseteq U$. 
We will see later how the latter problem can be decided 
in double exponential space
by further reducing it to checking the emptiness of the 
intersection of the languages $\dom(\cT)$ and $U^\complement$, 
where $U^\complement$ is the complement of $U$.

Below, we show how to construct a one-way transducer 
$\Tt'$ of triple exponential size such that 
\[
  \Tt' \:\subseteq\: \Tt
  \qquad\text{and}\qquad
  \dom(\Tt') \:\supseteq\: U.
\]
In particular, the existence of such a transducer $\cT'$ 
proves the implication from \PR3 to \PR1 of Theorem~\ref{thm:main2}.
It also proves the second item of Theorem~\ref{th:main}, because
when $\cT$ is one-way definable, $U=\dom(\cT)$, and hence $\cT$ and 
$\cT'$ are equivalent.

Intuitively, given an input $u$, the one-way transducer $\cT'$ 
will guess a successful run $\rho$ of $\cT$ on $u$ and a 
decomposition of $\rho$, and then use the decomposition to 
simulate the output produced by $\rho$.
Note that $\cT'$ accepts at least all the words of $U$, 
possibly more. As a matter of fact, it would be difficult
to construct a transducer whose domain coincides with $U$,
since checking membership in $U$ involves a universal quantification.
The proof of the following result is in the appendix.

\begin{proposition}\label{prop:sufficiency}
Given a functional two-way transducer $\cT$, 
one can construct in $\threeexptime$ 
a one-way transducer $\cT'$ such that 
$\cT' \subseteq \cT$ and $\dom(\cT') \supseteq U$.
\end{proposition}

\paragraph{Deciding one-way definability.}
Recall that $\cT$ is one-way definable iff $\dom(\cT) \subseteq U$, 
so iff $\dom(\cT) \cap U^\complement =\emptyset$. 
The lemma below exploits the characterization of Theorem~\ref{thm:main2} 
to show that the language $U^\complement$ can be recognized by an
 NFA $\cU^\complement$ of triple exponential size. The lemma
actually shows that the NFA recognizing $U^\complement$ can be constructed
using double exponential {\sl workspace}. 

\begin{lemma}\label{lem:U-complement}
Given a functional two-way transducer $\cT$, 
one can construct in $\twoexpspace$ an NFA recognizing $U^\complement$.
\end{lemma}

\begin{proof}
Consider an input word $u$. By Theorem~\ref{thm:main2} we know that 
$u\in U^\complement$ iff there exist a successful run $\rho$ of $\cT$ 
on $u$ and an inversion $\cI=(L_1,C_1,L_2,C_2)$ of $\rho$ such that
no positive number $p \le \bound$ is a period of the word
\[
  w_{\rho,\cI} ~=~ 
  \outb{\tr{C_1}} ~ \outb{\rho[\an{C_1},\an{C_2}]} ~ \outb{\tr{C_2}}.
\]
The latter condition on $w_{\rho,\cI}$ can be rephrased
as follows: there is a function $f:\{1,\dots,\bound\} \rightarrow \{1,\dots,|w_{\rho,\cI}|\}$
such that $w_{\rho,\cI}[f(p)] \neq w_{\rho,\cI}[f(p)+p]$ 
for all positive numbers $p\le\bound$.
Recall that $\bound = \cmax \cdot \hmax \cdot (2^{3\emax}+4)$,
where $\hmax = 2|Q|-1$, $\emax = (2|Q|)^{2\hmax}$, and
$Q$ is the state space of the two-way transducer $\cT$.
This means that the run $\rho$, the inversion $\cI$, and the 
function $f$ described above can all be guessed within double exponential space,
namely, using a number of states that is at most a triple exponential 
w.r.t.~$|\cT|$. In particular, we can construct in $\twoexpspace$ 
an NFA recognizing $U^\complement$.
\end{proof}

As a consequence of the previous lemma and of Theorem~\ref{thm:main2}, we have that the emptiness 
of the language $\dom(\cT) \cap U^\complement$, and hence the one-way 
definability of $\cT$, can be decided in $\twoexpspace$:

\begin{corollary}\label{cor:complexity}
The problem of deciding whether a functional two-way transducer is one-way definable is in $\twoexpspace$.
\end{corollary}

\section{Definability by sweeping transducers}\label{sec:sweepingness}

A two-way transducer is called \emph{sweeping} 
if every successful run of it performs reversals only 
at the extremities of the input word, i.e.~when reading 
the symbols $\vdash$ or $\dashv$. Similarly, we call
it \emph{$k$-pass sweeping} if it is sweeping and every 
successful run performs at most $k-1$ reversals. 
Clearly, a $1$-pass sweeping transducer is the same 
as a one-way transducer. 

In this section we are considering the following question: given a
functional two-way transducer, is it equivalent to some $k$-pass
sweeping transducer? We call such transducers \emph{$k$-pass sweeping
definable}. If the parameter $k$ is not given \emph{a priori}, then we
denote them as \emph{sweeping definable transducers}. 

In \cite{bgmp16} we built up on the characterization of 
one-way definability for (the restricted class of) sweeping 
transducers~\cite{bgmp15} in order to determine the minimal number
of passes required by sweeping transductions.
Essentially, the idea was to consider a 
generalization of the notion of inversion, called 
\emph{$k$-inversion}, and proving that $k$-pass sweeping 
definability is equivalent to asking that every $k$-inversion 
generates a periodic output.

We show that we can follow the same approach for two-way transducers. 
More precisely, we first define a \emph{co-inversion} 
in a way similar to Definition \ref{def:inversion}, namely,
as a tuple $(L_1,C_1,L_2,C_2)$ consisting of two idempotent 
loops $L_1,L_2$, a component $C_1$ of $L_1$, and a component 
$C_2$ of $L_2$ such that 
\begin{itemize}
  \item $\an{C_1} \leqtime \an{C_2}$,
  \item $\out{\tr{C_1}},\out{\tr{C_2}}\neq\emptystr$, and
  \item $\an{C_i}=(x_i,y_i)$ for $i=1,2$, then $x_1 \le x_2$.
\end{itemize}
The only difference compared to inversions is the ordering 
of the positions of the anchors, which is now reversed.

Alternating inversions and co-inversions leads to:

\begin{definition}\label{def:k-inversion}
A \emph{$k$-inversion} is a tuple 
$\bar\cI = (\cI_0,\dots,\cI_{k-1})$, where $\cI_i = (L_i,C_i,L'_i,C'_i)$
is either an inversion or a co-inversion depending on whether $i$ is 
even or odd, and $\an{C'_i} \leqtime \an{C_{i+1}}$ for all $i<k-1$.
\par\noindent
A $k$-inversion $\bar\cI$ is \emph{safe} if 
for {\sl some} $0\le i<k$, the word 
\[
  \outb{\tr{C_i}} ~ \outb{\rho[\an{C_i},\an{C'_i}]} ~ \outb{\tr{C'_i}}
\] 
has period $p\le\bound$ dividing 
$|\out{\tr{C_i}}|$ and $|\out{\tr{C'_i}}|$. 
\end{definition}

Similar to the characterization of $k$-pass 
sweeping definability in~\cite{bgmp16}, 
we show now the following characterization for 2-way
transducers, using Theorem~\ref{thm:main2} as a black-box:

\begin{theorem}\label{thm:k-pass-sweeping}
Let $\cT$ be a functional two-way transducer and $k>0$. 
The following are equivalent:
\begin{enumerate}
  \item $\cT$ is $k$-pass sweeping definable.
  \item All $k$-inversions of all successful runs of $\cT$ are safe.
\end{enumerate}
The problem of deciding whether the above conditions hold
is in $\twoexpspace$; more precisely, it can be decided in double
exponential space w.r.t.~$|\cT|$ and in polynomial space w.r.t.~$k$.
\end{theorem}

\begin{proof}[Proof sketch]
A proof of this result (modulo the necessary changes in complexity 
due to the new characterization) can be found in \cite{bgmp16}.
Here we present in an informal way the main steps of the proof.

Proving the implication from 2) to 1) boils down to factorize a
successful run $\rho$ of $\cT$ into factors $\rho_1,\dots,\rho_k$ 
in such a way that, for every odd (resp.~even) index $i$, 
$\rho_i$ contains only inversions (resp.~co-inversions) that are
safe, namely, that yield periodic outputs. We use the 
constructions presented in Section \ref{sec:decomposition} to
simulate the output of each factor $\rho_i$ with a one-way transducer,
which scans the input either from left to right or from right to left, 
depending on whether $i$ is odd or even. 

The implication from 1) to 2) amounts at showing that every 
$k$-inversion is safe under the assumption that $\cT$ is $k$-pass
sweeping definable. The proof builds upon the characterization of
one-way definability. More precisely, we consider a successful run 
of $\cT$ and the corresponding run of an equivalent $k$-pass sweeping 
transducer $\cT'$ that produces the same output. We then pump those
runs simultaneously on all loops $L_1,\dots,L_{2k}$ that form the 
$k$-inversion. By reasoning as in the proof of 
Proposition~\ref{prop:inversion}, we derive a periodicity
property that shows that the $k$-inversion is safe.

Finally, the $\twoexpspace$ complexity of the decision problem
follows from reducing $k$-pass sweeping definability to the 
emptiness of the language $\dom(\cT) \:\cap\: U^\complement$,
where $U$ is now the language of words $u\in\dom(\cT)$
such that all $k$-inversions of all successful runs on $u$ are safe.
As usual the latter problem is solved
by constructing an NFA that recognizes $U^\complement$ by
guessing a successful run $\rho$ of $\cT$ and an unsafe 
$k$-inversion of $\rho$.
\end{proof}

A similar problem, called \emph{sweeping definability}, concerns 
the characterization of those transductions that are definable by 
sweeping transducers, but this time without enforcing any bound 
on the number of passes (or reversals). 
Of course the latter problem 
is interesting only when the transductions are presented by means of 
two-way transducers. 
Below we show that the sweeping definability problem reduces to 
the $k$-pass sweeping definability problem, when we set $k$ large enough.

\begin{theorem}\label{thm:sweeping}
A functional two-way transducer $\cT$ is sweeping definable iff
it is $k$-pass sweeping definable, for $k=2\hmax\cdot (2^{3\emax}+1)$.
\end{theorem}

\begin{proof}[Proof sketch]
The right-to-left implication is trivial.
The proof of the converse direction is in the appendix;
here we only provide a rough idea.
Suppose that $\cT$ is not $k$-pass sweeping definable,
for $k=2\hmax\cdot (2^{3\emax}+1)$.
By Theorem \ref{thm:k-pass-sweeping}, there exists
a successful run $\rho$ of $\cT$ and an unsafe 
$k$-inversion $\bar\cI$ of $\rho$. 
One can exploit the fact that $k$ is large enough 
to find an idempotent loop $L$ and an intercepted 
factor of it that covers two consecutive (co-)inversions
of $\bar\cI$. 
Then, by pumping the loop $L$, one can 
introduce arbitrarily long alternations between 
inversions and co-inversions, thus showing that
there are successful runs with unsafe
$k'$-inversions for all $k'>0$.
By Theorem \ref{thm:k-pass-sweeping}, this
proves that $\cT$ is not sweeping definable.
\end{proof}

\begin{corollary}\label{cor:sweeping}
The problem of deciding sweeping definability of a 
functional two-way transducer is in $\twoexpspace$.
\end{corollary}

Another consequence is that it is decidable in $\twoexpspace$
whether a functional two-way transducer is equivalent to some
two-way transducer performing a bounded number of reversals in every run.
Indeed, in \cite{bgmp16} we proved that a functional transducer is
$k$-pass sweeping definable iff it is $(k-1)$-reversal definable.

Other classes of transducers are amenable to characterizations
via similar techniques. For example, we may consider an even more
restricted variant of transducer, called \emph{rotating transducer}.
This is a sweeping transducer that emits output only when moving from left to right.
Such a transducer is called \emph{$k$-pass} if it performs at most $k$ passes from 
left to right.
To characterize those transductions that are definable by $k$-pass rotating transducers
it suffices to modify slightly the definition of $k$-inversion, by 
removing co-inversions. Formally, one defines a \emph{rotating $k$-inversion}
as a tuple $\bar\cI = (\cI_0,\dots,\cI_{k-1})$, where each
$\cI_i = (L_i,C_i,L'_i,C'_i)$ is an inversion and 
$\an{C'_i} \lesstime \an{C_{i+1}}$ for all $i<k-1$. 
The analogous of Theorems \ref{thm:k-pass-sweeping}
and \ref{thm:sweeping} would then carry over.

\section{Conclusions}\label{sec:conclusions}

It was shown recently \cite{fgrs13} that it is decidable whether 
a given two-way transducer can be implemented by some one-way
transducer, however the complexity of the algorithm is non-elementary.

The main contribution of our paper is a new algorithm that solves 
the above question with elementary complexity, precisely in $\twoexpspace$.
The algorithm is based on a characterization of those transductions, 
given as two-way transducers, that can be realized by one-way 
transducers. The flavor of our characterization is 
different from that of \cite{fgrs13}. The approach from
\cite{fgrs13} is based on a variant of Rabin and Scott's construction 
\cite{RS59} of one-way automata, and on local modifications of
the two-way run. Our approach relies instead on the global notion
of \emph{inversions} and on combinatorial arguments, and is inspired
by our previous result for sweeping
transducers~\cite{bgmp15}. The technical challenge in this paper
compared to~\cite{bgmp15} is however significant, and required several
involved proof ingredients, ranging from the type of loops we
consider, up to the decomposition of the runs.

Our characterization based on inversions yields not only an elementary
solution for the problem of one-way definability, but also for definability by sweeping
(resp.~rotating) transducers, with either known or unknown number of passes.
All characterizations above are effective, and can be decided in $\twoexpspace$.

\bibliographystyle{IEEEtran}
\bibliography{biblio}

\begin{thebibliography}{10}
\providecommand{\url}[1]{#1}
\csname url@samestyle\endcsname
\providecommand{\newblock}{\relax}
\providecommand{\bibinfo}[2]{#2}
\providecommand{\BIBentrySTDinterwordspacing}{\spaceskip=0pt\relax}
\providecommand{\BIBentryALTinterwordstretchfactor}{4}
\providecommand{\BIBentryALTinterwordspacing}{\spaceskip=\fontdimen2\font plus
\BIBentryALTinterwordstretchfactor\fontdimen3\font minus
  \fontdimen4\font\relax}
\providecommand{\BIBforeignlanguage}[2]{{%
\expandafter\ifx\csname l@#1\endcsname\relax
\typeout{** WARNING: IEEEtran.bst: No hyphenation pattern has been}%
\typeout{** loaded for the language `#1'. Using the pattern for}%
\typeout{** the default language instead.}%
\else
\language=\csname l@#1\endcsname
\fi
#2}}
\providecommand{\BIBdecl}{\relax}
\BIBdecl

\bibitem{sch61}
M.~Sch\"utzenberger, ``A remark on finite transducers,'' \emph{Information and
  Control}, vol.~4, no. 2-3, pp. 185--196, 1961.

\bibitem{ahu69}
A.~Aho, J.~Hopcroft, and J.~Ullman, ``A general theory of translation,''
  \emph{Math.~Syst.~Theory}, vol.~3, no.~3, pp. 193--221, 1969.

\bibitem{eil76}
S.~Eilenberg, \emph{Automata, Langages and Machines}.\hskip 1em plus 0.5em
  minus 0.4em\relax Academic Press, 1976.

\bibitem{RS59}
M.~Rabin and D.~Scott, ``Finite automata and their decision problems,''
  \emph{IBM J. Res. Dev.}, vol.~3, no.~2, pp. 114--125, 1959.

\bibitem{she59}
J.~Shepherdson, ``The reduction of two-way automata to one-way automata,''
  \emph{IBM J. Res. Dev.}, vol.~3, no.~2, pp. 198--200, 1959.

\bibitem{EH98}
J.~Engelfriet and H.~J. Hoogeboom, ``{MSO} definable string transductions and
  two-way finite-state transducers,'' \emph{ACM Trans. Comput. Logic}, vol.~2,
  no.~2, pp. 216--254, 2001.

\bibitem{AlurC10}
R.~Alur and P.~Cern{\'y}, ``Expressiveness of streaming string transducers.''
  in \emph{FSTTCS}, ser. LIPIcs, vol.~8.\hskip 1em plus 0.5em minus 0.4em\relax
  Schloss Dagstuhl - Leibniz-Zentrum fuer Informatik, 2010, pp. 1--12.

\bibitem{fgrs13}
E.~Filiot, O.~Gauwin, P.~Reynier, and F.~Servais, ``From two-way to one-way
  finite state transducers,'' in \emph{28th Annual {ACM/IEEE} Symposium on
  Logic in Computer Science, {LICS} 2013, New Orleans, LA, USA, June 25-28,
  2013}.\hskip 1em plus 0.5em minus 0.4em\relax {IEEE} Computer Society, 2013,
  pp. 468--477.

\bibitem{bgmp15}
F.~Baschenis, O.~Gauwin, A.~Muscholl, and G.~Puppis, ``One-way definability of
  sweeping transducer,'' in \emph{35th {IARCS} Annual Conference on Foundation
  of Software Technology and Theoretical Computer Science, {FSTTCS} 2015,
  December 16-18, 2015, Bangalore, India}, ser. LIPIcs, vol.~45.\hskip 1em plus
  0.5em minus 0.4em\relax Schloss Dagstuhl - Leibniz-Zentrum f{\"u}r
  Informatik, 2015, pp. 178--191.

\bibitem{bgmp16}
------, ``Minimizing resources of sweeping and streaming string transducers,''
  in \emph{43rd International Colloquium on Automata, Languages, and
  Programming, {ICALP} 2016, July 11-15, 2016, Rome, Italy}, ser. LIPIcs,
  vol.~55.\hskip 1em plus 0.5em minus 0.4em\relax Schloss Dagstuhl -
  Leibniz-Zentrum fuer Informatik, 2016, pp. 114:1--114:14, full version
  available at \url{https://hal.archives-ouvertes.fr/hal-01274992}.

\bibitem{FiliotKT14}
E.~Filiot, S.~N. Krishna, and A.~Trivedi, ``First-order definable string
  transformations,'' in \emph{34th International Conference on Foundation of
  Software Technology and Theoretical Computer Science, {FSTTCS} 2014, December
  15-17, 2014, New Delhi, India}, ser. LIPIcs, vol.~29.\hskip 1em plus 0.5em
  minus 0.4em\relax Schloss Dagstuhl - Leibniz-Zentrum f{\"u}r Informatik,
  2014, pp. 147--159.

\bibitem{CartonDartois15}
O.~Carton and L.~Dartois, ``Aperiodic two-way transducers and
  {FO}-transductions,'' in \emph{24th {EACSL} Annual Conference on Computer
  Science Logic (CSL)}, ser. LIPIcs, vol.~41.\hskip 1em plus 0.5em minus
  0.4em\relax Schloss Dagstuhl - Leibniz-Zentrum fuer Informatik, 2015, pp.
  160--174.

\bibitem{FGL16}
E.~Filiot, O.~Gauwin, and N.~Lhote, ``First-order definability of rational
  transductions: An algebraic approach,'' in \emph{Proceedings of the 31st
  Annual {ACM/IEEE} Symposium on Logic in Computer Science, {LICS}'16, New
  York, NY, USA, July 5-8, 2016}.\hskip 1em plus 0.5em minus 0.4em\relax {ACM},
  2016, pp. 387--396.

\bibitem{DRT16}
L.~Daviaud, P.~Reynier, and J.~Talbot, ``A generalised twinning property for
  minimisation of cost register automata,'' in \emph{Proceedings of the 31st
  Annual {ACM/IEEE} Symposium on Logic in Computer Science, {LICS}'16, New
  York, NY, USA, July 5-8, 2016}.\hskip 1em plus 0.5em minus 0.4em\relax {ACM},
  2016, pp. 857--866.

\bibitem{CG14mfcs}
C.~Choffrut and B.~Guillon, ``An algebraic characterization of unary two-way
  transducers,'' in \emph{Mathematical Foundations of Computer Science 2014 -
  39th International Symposium, {MFCS} 2014, Budapest, Hungary, August 25-29,
  2014. Proceedings, Part {I}}, ser. Lecture Notes in Computer Science, vol.
  8634.\hskip 1em plus 0.5em minus 0.4em\relax Springer, 2014, pp. 196--207.

\bibitem{Gui15}
B.~Guillon, ``Sweeping weakens two-way transducers even with a unary output
  alphabet,'' in \emph{Seventh Workshop on Non-Classical Models of Automata and
  Applications - {NCMA} 2015, Porto, Portugal, August 31 - September 1, 2015.
  Proceedings}, ser. books@ocg.at, vol. 318.\hskip 1em plus 0.5em minus
  0.4em\relax {\"{O}}sterreichische Computer Gesellschaft, 2015, pp. 91--108.

\bibitem{HU79}
J.~E. Hopcroft and J.~D. Ullman, \emph{Introduction to Automata Theory,
  Languages, and Computation}.\hskip 1em plus 0.5em minus 0.4em\relax
  Addison-Wesley, 1979.

\bibitem{Birget1990}
J.~Birget, ``Two-way automaton computations,'' \emph{RAIRO - Theoretical
  Informatics and Applications - Informatique Théorique et Applications},
  vol.~24, no.~1, pp. 47--66, 1990.

\bibitem{factorization_forests}
I.~Simon, ``Factorization forests of finite height,'' \emph{Theoretical
  Computer Science}, vol.~72, no.~1, pp. 65--94, 1990.

\bibitem{factorization_forests_for_words_paper}
T.~Colcombet, ``Factorisation forests for infinite words,'' in \emph{FCT}, ser.
  LNCS, vol. 4639.\hskip 1em plus 0.5em minus 0.4em\relax Springer, 2007, pp.
  226--237.

\bibitem{fine-wilf}
N.~Fine and H.~Wilf, ``Uniqueness theorems for periodic functions,''
  \emph{Proceedings of the American Mathematical Society}, vol.~16, pp.
  109--114, 1965.

\bibitem{kortelainen98}
J.~Kortelainen, ``On the system of word equations $x_0u_1^ix_1 u_2^i x_2 \dots
  u_m^i x_m = y_0 v_1^iy_1v_2^i y_2\dots v_m^i y_m$ ($i=0,1,2,\dots$) in a free
  monoid,'' \emph{Journal of Automata, Languages and Combinatorics}, vol.~3,
  no.~1, pp. 43--57, 1998.

\bibitem{saarela15}
A.~Saarela, ``Systems of word equations, polynomials and linear algebra: a new
  approach,'' \emph{European Journal of Combinatorics}, vol.~47, no.~5, pp.
  1--14, 2015.

\end{thebibliography}

\newpage
\onecolumn
\appendix

Before proving Lemma~\ref{lem:component2},
we show that in a loop, the levels of each component form an interval.

\begin{lemma}\label{lem:component}
Let $C$ be a component of a loop $L=[x_1,x_2]$, $y^-=\min(C)$, and $y^+=\max(C)$. 
The nodes of $C$ are precisely the levels in the interval $[y^-,y^+]$.
Moreover, if $C$ is left-to-right (resp.~right-to-left), 
then $y^+$ is the smallest level 
$\ge y^-$ such that between $(x_1,y^-)$ and $(x_2,y^+)$ (resp.~$(x_2,y^-)$ and $(x_1,y^+)$)
there are equally many $\LL$-factors and $\RR$-factors intercepted by $L$.
\end{lemma}

\begin{figure}[!t]
\centering
\begin{tikzpicture}[baseline=0, inner sep=0, outer sep=0, minimum size=0pt, scale=0.32]
  \tikzstyle{dot} = [draw, circle, fill=white, minimum size=4pt]
  \tikzstyle{fulldot} = [draw, circle, fill=black, minimum size=4pt]
  \tikzstyle{factor} = [->, shorten >=1pt, dotted, rounded corners=5]
  \tikzstyle{fullfactor} = [->, >=stealth, shorten >=1pt, very thick, rounded corners=5]

\begin{scope}[yscale=1.2]
  \fill [pattern=north east lines, pattern color=gray!25]
        (0,-2) rectangle (6,13);
  \draw [dashed, thin, gray] (0,-2) -- (0,13);
  \draw [dashed, thin, gray] (6,-2) -- (6,13);
  \draw [gray] (0,-2.25) -- (0,-2.5) -- (6,-2.5) -- (6,-2.25);
  \draw [gray] (3,-2.75) node [below] {\footnotesize $L$};
 
  \draw (3,-0.75) node () {$\vdots$};
  \draw (3,12.25) node () {$\vdots$};

  \draw (0,0) node [dot] (node1) {};
  \draw (2,0) node (node2) {};
  \draw (2,1) node (node3) {};
  \draw (0,1) node [dot] (node4) {};

  \draw (0,2) node [dot] (node5) {};
  \draw (2,2) node (node6) {};
  \draw (2,3) node (node7) {};
  \draw (0,3) node [dot] (node8) {};

  \draw (0,4) node [dot] (node9) {};
  \draw (2,4) node (node10) {};
  \draw (2,5) node (node11) {};
  \draw (0,5) node [dot] (node12) {};

  \draw (0,6) node [fulldot] (node13) {};
  \draw (6,6) node [fulldot] (node14) {};

  \draw (6,7) node [dot] (node15) {};
  \draw (4,7) node (node16) {};
  \draw (4,8) node (node17) {};
  \draw (6,8) node [dot] (node18) {};

  \draw (6,9) node [dot] (node19) {};
  \draw (4,9) node (node20) {};
  \draw (4,10) node (node21) {};
  \draw (6,10) node [dot] (node22) {};

  \draw (6,11) node [fulldot] (node23) {};
  \draw (0,11) node [fulldot] (node24) {};

  \draw (node13) node [left = 2mm, gray] {\footnotesize $\ort y_i$};
  \draw (node14) node [right = 2mm, gray] {\footnotesize $\olft y_{i-1}+1$};
  \draw (node23) node [right = 2mm, gray] {\footnotesize $\olft y_i$};
  \draw (node24) node [left = 2mm, gray] {\footnotesize $\ort y_i+1$};

  \draw [fullfactor] (node1) -- (node2.center) -- (node3.center) -- (node4); 
  \draw [fullfactor] (node5) -- (node6.center) -- (node7.center) -- (node8); 
  \draw [fullfactor] (node9) -- (node10.center) -- (node11.center) -- (node12); 

  \draw [fullfactor, blue] (node13) -- (node14);

  \draw [fullfactor] (node15) -- (node16.center) -- (node17.center) -- (node18); 
  \draw [fullfactor] (node19) -- (node20.center) -- (node21.center) -- (node22); 

  \draw [fullfactor, red] (node23) -- (node24);
\end{scope}

\begin{scope}[yscale=1.7,xshift=20cm]
  \draw [dashed, thin, gray] (0,-0.5) -- (0,8.5);
  \draw [gray] (0,-1.75) node [below] {\footnotesize $F_L$};
    
  \draw (0,-0.75) node () {$\vdots$};
  \draw (0,9.25) node () {$\vdots$};

  \draw (0,0) node [dot] (node1) {};
  \draw (2,0) node (node2) {};
  \draw (2,1) node (node3) {};
  
  \draw (0,1) node [dot] (node4) {};
  \draw (-2,1) node (node5) {};
  \draw (-2,2) node (node6) {};
  
  \draw (0,2) node [dot] (node7) {};
  \draw (2,2) node (node8) {};
  \draw (2,3) node (node9) {};

  \draw (0,3) node [dot] (node10) {};
  \draw (-2,3) node (node11) {};
  \draw (-2,4) node (node12) {};

  \draw (0,4) node [dot] (node13) {};
  \draw (2,4) node (node14) {};
  \draw (2,5) node (node15) {};

  \draw (0,5) node [fulldot] (node16) {};
  \draw (-3,5) node (node17) {};
  \draw (-3,8) node (node18) {};
  \draw (-2,8) node (node19) {};
  \draw (0,8) node [fulldot] (node20) {};

  \draw (0,7) node [fulldot] (node21) {};
  \draw (4,7) node (node22) {};
  \draw (4,-1) node (node23) {};
  \draw (2,-1) node (node24) {};
  \draw (0,0) node [fulldot] (node25) {};

  \draw (node1) node [left = 1mm, gray] {\footnotesize $\olft y_{i-1}+1$};
  \draw (node16) node [above right = 1mm, gray] {\footnotesize $\olft y_i$};
  \draw (node21) node [left = 1mm, gray] {\footnotesize $\ort y_i$};
  \draw (node20) node [right = 1mm, gray] {\footnotesize $\ort y_i+1$};

  \draw [fullfactor] (node1) -- (node2.center) -- (node3.center) -- (node4); 
  \draw [fullfactor] (node4) -- (node5.center) -- (node6.center) -- (node7); 
  \draw [fullfactor] (node7) -- (node8.center) -- (node9.center) -- (node10); 
  \draw [fullfactor] (node10) -- (node11.center) -- (node12.center) -- (node13); 
  \draw [fullfactor] (node13) -- (node14.center) -- (node15.center) -- (node16); 
  \draw [fullfactor, red] (node16) -- (node17.center) -- 
                          (node18.center) -- (node19.center) -- (node20);
  \draw [fullfactor, blue] (node21) -- (node22.center) -- 
                           (node23.center) -- (node24.center) -- (node25);  
\end{scope}
\end{tikzpicture}
\caption{Some factors intercepted by $L$ and the corresponding edges in the flow.}%
\label{fig:edges} 
\end{figure}

\begin{proof}
To ease the understanding the reader may refer to Fig.~\ref{fig:edges}, 
that shows some factors intercepted by $L$ and the corresponding edges in the flow.

We begin the proof by partitioning the set of levels of the flow into 
suitable intervals as follows.
We observe that every loop $L=[x_1,x_2]$ intercepts equally many
$\LL$-factors and $\RR$-factors. This is so because the crossing
sequences at $x_1,x_2$ have the same length $h$. 
We also observe that the sources of the factors intercepted by $L$
are either of the form $(x_1,y)$, with $y$ even, or $(x_2,y)$, with $y$ odd.
For any location $\ell\in\{x_1,x_2\}\times\bbN$ that is the source
of an intercepted factor, we define $d_{\ell}$ to be the difference 
between the number of $\LL$-factors and the number of $\RR$-factors 
intercepted by $L$ that {\sl end} at a location {\sl strictly before} $\ell$.
Intuitively, $d_{\ell}=0$ when the prefix of the run up to location $\ell$ 
has visited equally many times the position $x_1$ and the position $x_2$. 
For the sake of brevity, we let $d_y=d_{(x_1,y)}$ for an even level $y$,
and $d_y=d_{(x_2,y)}$ for an odd level $y$. 
Note that $d_0=0$. We also let $d_{h+1}=0$, by convention.

We now consider the numbers $z$'s, with $0\le z\le h+1$, such that 
$d_z=0$, that is: $0 = z_0 < z_1 < \dots < z_k = h+1$.  
Using a simple induction, we prove that for all $i \le k$, 
the parity of $z_i$ is the same as the parity of its index $i$. 
The base case $i = 0$ is trivial, since $z_0 = 0$. For the inductive
case, suppose that $z_i$ is even (the case of $z_i$ odd is similar). 
We prove that $z_{i+1}$ is odd by a case distinction based on the 
type of factor intercepted by $L$ that starts at level $z_i$. 
If this factor is an $\LR$-factor, then it ends at the same level $z_i$,
and hence $d_{z_i+1} = d_{z_i} = 0$, which implies that $z_{i+1} = z_i +1$ is odd. 
Otherwise, if the factor is an $\LL$-factor, then for all levels $z$ strictly
between $z_i$ and $z_{i+1}$, we have $d_z > 0$, and since $d_{z_{i+1}} =0$, 
the last factor before $z_{i+1}$ must decrease $d_z$, that is, must be an $\RR$-factor. 
This implies that $(x_2,z_{i+1})$ is the source of an intercepted factor, 
and thus $z_{i+1}$ is odd.

The levels $0 = z_0 < z_1 < \dots < z_k = h+1$ induce a partition of
the set of nodes of the flow into intervals of the form $Z_i=[z_i,z_{i+1}-1]$.
To prove the lemma, it is suffices to show that the subgraph of the flow 
induced by each interval $Z_i$ is connected. Indeed, because the union of 
the previous intervals covers all the nodes of the flow, and because each 
node has one incoming and one outgoing edge, this will imply that the 
intervals coincide with the components of the flow.

\medskip
Now, let us fix an interval of the partition, which we denote by $Z$ 
to avoid clumsy notation. Hereafter, we will focus on the edges of subgraph
of the flow induced by $Z$ (we call it \emph{subgraph of $Z$} for short). 
We prove a few basic properties of these edges.
For the sake of brevity, we call $\LL$-edges the edges of the subgraph of $Z$
that correspond to the $\LL$-factors intercepted by $L$, and similarly for 
the $\RR$-edges, $\LR$-edges, and $\RL$-edges.

We make a series of assumption to simplify our reasoning.
First, we assume that the edges are ordered based on the occurrences of the 
corresponding factors in the run. For instance, we may say the first, 
second, etc.~$\LR$-edge (of the subgraph of $Z$) --- from now on, we 
tacitly assume that the edges are inside the subgraph of $Z$.
Second, we assume that the first edge of the subgraph of $Z$ starts 
at an even node, namely, it is an $\LL$-edge or an $\LR$-edge 
(if this were not the case, one could apply symmetric arguments to prove the lemma).
From this it follows that the subgraph contains $n$ $\LR$-edges interleaved by
$n-1$ $\RL$-edges, for some $n>0$. 
Third, we assume that $\min(Z)=0$, in order to avoid clumsy notations 
(otherwise, we need to add $\min(Z)$ to all the levels considered hereafter).

Now, we observe that, by definition of $Z$, there are equally 
many $\LL$-edges and $\RR$-edges: 
indeed, the difference between the number of $\LL$-edges and the number
of $\RR$-edges at the beginning and at the end of $Z$ is the same, namely,
$d_z = 0$ for both $z=\min(Z)$ and $z=\max(Z)$.
It is also easy to see that the $\LL$-edges and the $\RR$-edges 
are all of the form $y \rightarrow y+1$, for some level $y$.
We call these edges \emph{incremental edges}.

For the other edges, we denote by $\ort y_i$ (resp.~$\olft y_i$) 
the source level of the $i$-th $\LR$-edge (resp.~the $i$-th $\RL$-edge). 
Clearly, each $\ort y_i$ is even, and each $\olft y_i$ is odd,
and $i\le j$ implies $\ort y_i < \ort y_j$ and $\olft y_i < \olft y_j$.
Consider the location $(x_1,\ort y_i)$, which is the source of the 
$i$-th $\LR$-edge (e.g.~the edge in blue in the figure).  
The latest location at position $x_2$ that 
precedes $(x_1,\ort y_i)$ must be of the form 
$(x_2,\olft y_{i-1})$, provided that $i>1$.
This implies that, for all $1<i\le n$, the $i$-th $\LR$-edge 
is of the form $\ort y_i \rightarrow \olft y_{i-1} + 1$.
For $i=1$, we recall that $\min(Z)=0$ and observe that the 
first location at position $x_2$ that occurs after the location 
$(x_1,0)$ is $(x_2,0)$, and thus the
first $\LR$-edge has a similar form: $\ort y_1 \rightarrow \olft y_0 + 1$,
where $\olft y_0 = -1$ by convention.

Using symmetric arguments, we see that the $i$-th $\RL$-edge 
(e.g.~the one in red in the figure) is of the form 
$\olft y_i \rightarrow \ort y_i + 1$. In particular,
the last $\LR$-edge starts at the level $\ort y_n = \max(Z)$.

Summing up, we have just seen that the edges of the subgraph of $Z$ are of the following forms:
\begin{itemize}
  \item \rightward{$y \rightarrow y+1$}                      
        \hspace{25mm} (incremental edges),
  \item \rightward{$\ort y_i \rightarrow \olft y_{i-1} + 1$} 
        \hspace{25mm} ($i$-th $\LR$-edge, for $i=1,\dots,n$),
  \item \rightward{$\olft y_i \rightarrow \ort y_i + 1$}     
        \hspace{25mm} ($i$-th $\RL$-edge, for $i=1,\dots,n-1$).
\end{itemize}
In addition, we have $\ort y_i +1 = \olft y_i + 2d_{\olft y_i}$.
Since $d_z > 0$ for all $\min(Z) < z < \max(Z)$, this implies that 
$\ort y_i > \olft y_i$.

\medskip
The goal is to prove that the subgraph of $Z$ is strongly connected, 
namely, it contains a cycle that visits all its nodes. As a matter of fact, 
because components are also strongly connected subgraphs, and because every 
node in the flow has in-/out-degree $1$, this will imply that the considered 
subgraph coincides with a component $C$, thus implying that the nodes in $C$ 
form an interval.
Towards this goal, we will prove a series of claims that aim at
identifying suitable sets of nodes that are covered by paths in 
the subgraph of $Z$.
Formally, we say that a path \emph{covers} a set $Y$ 
if it visits all the nodes in $Y$, and possibly other nodes.
As usual, when we talk of edges or paths, we tacitly 
understand that they occur inside the subgraph of $Z$.
On the other hand, we do not need to assume $Y\subseteq Z$,
since this would follow from the fact that $Y$ is covered by
a path inside $Z$.
For example, the right hand-side of Fig.~\ref{fig:edges} shows 
a path from $\ort y_i$ to $\ort y_i+1$ that covers the set 
$Y=\{\ort y_i,\ort y_i+1\} \cup [\olft y_{i-1}+1,\olft y_i]$.

The covered sets will be intervals of the form
\[
  Y_i ~=~ [\olft y_{i-1}+1,\olft y_i].
\]
Note that the sets $Y_i$ are well-defined for all $i=1,\dots,n-1$, but not 
for $i=n$ since $\olft y_n$ is not defined either (the subgraph of $Z$
contains only $n-1$ $\RL$-edges). 

\begin{claim}
For all $i=1,\dots,n-1$, there is a path from $\ort y_i$ to $\ort y_i+1$ 
that covers $Y_i$ (for short, we call it an \emph{incremental path}).
\end{claim}

\begin{claimproof}
We prove the claim by induction on $i$. The base case $i=1$ is rather easy. 
Indeed, we recall the convention that $\olft y_0+1 = \min(Z)=0$. 
In particular, the node $\olft y_0+1$ is the target of the first 
$\LR$-edge of the subgraph of $Z$.
Before this edge, according to the order induced by the run, 
we can only have $\LL$-edges of the form $y \rightarrow y+1$, 
with $y=0,2,\dots,\ort y_1-2$. Similarly, after the $\LR$-edge we have 
$\RR$-edges of the form $y \rightarrow y+1$, with $y=1,3,\dots,\olft y_1-2$.
Those incremental edges can be connected to form the path 
$\olft y_{i-1}+1 \rightarrow^* \olft y_1$ that covers 
the interval $[\olft y_0+1,\olft y_1]$ . 
By prepending to this path the $\LR$-edge $\ort y_1 \rightarrow \olft y_0 + 1$,
and by appending the $\RL$-edge $\olft y_1 \rightarrow \ort y_1 + 1$,
we get a path from $\ort y_1$ to $\ort y_1+1$ that covers 
the interval $[\olft y_0+1,\olft y_1]$. The latter interval is precisely 
the set $Y_1$.

For the inductive step, we fix $1<i<n$ and we construct the desired
path from $\ort y_i$ to $\ort y_i+1$. The initial edge of this path
is defined to be the $\LR$-edge $\ort y_i \rightarrow \olft y_{i-1} + 1$.
Similarly, the final edge of the path will be the $\RL$-edge 
$\olft y_i \rightarrow \ort y_i + 1$, which exists since $i<n$.
It remains to connect $\olft y_{i-1} + 1$ to $\olft y_i$.
For this, we consider the edges that depart from nodes strictly 
between $\olft y_{i-1}$ and $\olft y_i$.

Let $y$ be an arbitrary node in $[\olft y_{i-1}+1,\olft y_i-1]$. 
Clearly, $y$ cannot be of the form $\olft y_j$, for some $j$, 
because it is strictly between $\olft y_{i-1}$ and $\olft y_i$.
So $y$ cannot be the source of an $\RL$-edge.
Moreover, recall that the $\LL$-edges and the $\RR$-edges are the
of the form $y \rightarrow y+1$. As these incremental edges 
do not pose particular problems for the construction of the path, 
we focus mainly on the $\LR$-edges that depart from nodes inside 
$[\olft y_{i-1}+1,\olft y_i-1]$.

Let $\ort y_j \rightarrow \olft y_{j-1}+1$ be such an $\LR$-edge, 
for some $j$ such that $\ort y_j \in [\olft y_{i-1}+1,\olft y_i-1]$.
If we had $j\ge i$, then we would have 
$\ort y_j \ge \ort y_i > \olft y_i$, but this would contradict
the assumption that $\ort y_j \in [\olft y_{i-1}+1,\olft y_i-1]$.
So we know that $j<i$. 
This enables the use of the inductive hypothesis, which implies
the existence of an incremental path from $\ort y_j$ to $\ort y_j+1$ 
that covers the interval $Y_j$. 

Finally, by connecting the above paths using the 
incremental edges, and by adding the initial and 
final edges $\ort y_i \rightarrow \olft y_{i-1} + 1$ and
$\olft y_i \rightarrow \ort y_i + 1$, we obtain a path 
from $\ort y_i$ to $\ort y_i+1$. It is easy to see that this
path covers the interval $Y_i$.
\end{claimproof}

Next, we define 
\[
  Y ~=~ [\olft y_{n-1}+1,\ort y_n] ~\cup \bigcup_{1\le i<n} Y_i.
\]
We prove a claim similar to the previous one, 
but now aiming to cover $Y$ with a {\sl cycle}.
Towards the end of the proof we will argue that the set 
$Y$ coincides with the full interval $Z$, 
thus showing that there is a component $C$ 
whose set of notes is precisely $Z$.

\begin{claim}
There is a cycle that covers $Y$.
\end{claim}

\begin{claimproof}
It is convenient to construct our cycle starting from the last 
$\LR$-edge, that is, $\ort y_n \rightarrow \olft y_{n-1} + 1$, 
since this will cover the upper node $\ort y_n=\max(Z)$.
From there we continue to add edges and incremental paths, 
following an approach similar to the proof of the previous 
claim, until we reach the node $\ort y_n$ again.
More precisely, we consider the edges that depart from nodes strictly 
between $\olft y_{n-1}$ and $\ort y_n$. As there are only $n-1$ $\RL$-edges,
we know that every node in the interval $[\olft y_{n-1}+1,\ort y_n-1]$
must be source of an $\LL$-edge, an $\RR$-edge, or an $\LR$-edge.
As usual, incremental edges do not pose particular problems for 
the construction of the cycle, so we focus on the $\LR$-edges.
Let $\ort y_i \rightarrow \olft y_{i-1}+1$ be such an $\LR$-edge,
with $\ort y_i \in [\olft y_{n-1}+1,\ort y_n-1]$. Since $i < n$, we know 
from the previous claim that there is a path from $\ort y_i$ to $\ort y_i+1$
that covers $Y_i$.
We can thus build a cycle $\pi$ by connecting the above  
paths using the incremental edges and the $\LR$-edge 
$\ort y_n \rightarrow \olft y_{n-1} + 1$.

By construction, the cycle $\pi$ covers 
the interval $[\olft y_{n-1}+1,\ort y_n]$, and for every
$i<n$, if $\pi$ visits $\ort y_i$, then $\pi$ covers $Y_i$.
So to complete the proof --- namely, to show that $\pi$ covers the entire set $Y$ ---
it suffices to prove that $\pi$ visits each node $\ort y_i$, with $i<n$.

Suppose, by way of contradiction, that $\ort y_i$ is the node 
with the highest index $i<n$ that is not visited by $\pi$.
Recall that $\ort y_i > \olft y_i$. This shows that
\[
  \ort y_i ~\in~ [\olft y_i+1,\ort y_n] ~= \bigcup_{i\le j<n-1}[\olft y_j+1,\olft y_{j+1}] ~\cup~ [\olft y_{n-1}+1,\ort y_n].
\]
As we already proved that $\pi$ covers the interval $[\olft y_{n-1}+1,\ort y_n]$,
we know that $\ort y_i \in [\olft y_j+1,\olft y_{j+1}]$ for some $j$ with $i\le j<n-1$.
Now recall that $\ort y_i$ is the highest node that is not visited by $\pi$. This means
that $\ort y_{j+1}$ is visited by $\pi$. Moreover, since $j+1<n$, we know 
that $\pi$ uses the incremental path from $\ort y_{j+1}$ to $\ort y_{j+1}+1$, which
covers $Y_{j+1} = [\olft y_j+1,\olft y_{j+1}]$. But this contradicts the fact that $\ort y_i$ 
is not visited by $\pi$, since $\ort y_i \in [\olft y_j+1,\olft y_{j+1}]$.
\end{claimproof}

We know that the set $Y$ is covered by a cycle of the subgraph of $Z$, and that
$Z$ is an interval whose endpoints are consecutive levels $z<z'$, with $d_z=d_{z'}=0$.
For the homestretch, we prove that $Y=Z$. This will imply that the nodes of 
the cycle are precisely the nodes of the interval $Z$. 
Moreover, because the cycle must coincide with a component $C$ of the flow 
(recall that all the nodes have in-/out-degree $1$), this will show that 
the nodes of $C$ are precisely those of $Z$. 

To prove $Y=Z$ it suffices to recall its definition as the union of 
the interval $[\olft y_{n-1}+1,\ort y_n]$ with the sets $Y_i$, for all $i=1,\dots,n-1$.
Clearly, we have that $Y\subseteq Z$.
For the converse inclusion, we also recall that 
$\olft y_0+1 = 0 = \min(Z)$ and $\ort y_n=\max(Z)$. 
Consider an arbitrary level $z\in Z$. Clearly, we have either 
$z\le \olft y_i$, for some $1\le i<n$, or $z > \olft y_n$. 
In the former case, by choosing the smallest index $i$ such that $z \le \olft y_i$, 
we get $z\in [\olft y_{i-1}+1,\olft y_i]$, whence $z\in Y_i \subseteq Y$.
In the latter case, we immediately have $z\in Y$, by construction.
\end{proof}

\begin{replemma}{lem:component2}
If $C$ is a left-to-right (resp.~right-to-left) component 
of an {\sl idempotent} loop $L$, then the $(L,C)$-factors are in the following order: 
$k$ $\LL$-factors (resp.~$\RR$-factors), followed by one $\LR$-factor (resp.~$\RL$-factor), 
followed by $k$ $\RR$-factors (resp.~$\LL$-factors), for some $k \ge 0$. 
\end{replemma}

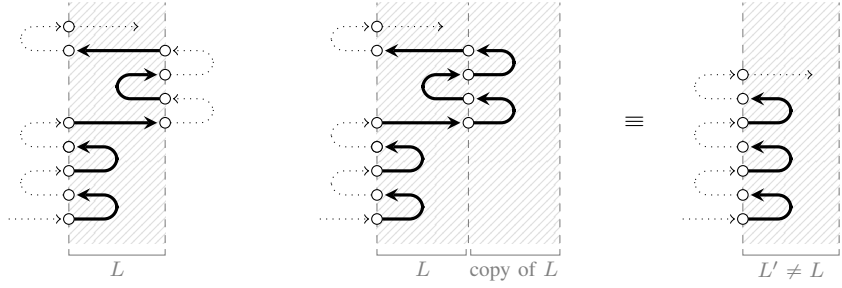
\begin{figure}[!t]
\centering
\begin{tikzpicture}[baseline=0, inner sep=0, outer sep=0, minimum size=0pt, scale=0.32]
  \tikzstyle{dot} = [draw, circle, fill=white, minimum size=4pt]
  \tikzstyle{fulldot} = [draw, circle, fill=black, minimum size=4pt]
  \tikzstyle{factor} = [->, shorten >=1pt, dotted, rounded corners=5]
  \tikzstyle{fullfactor} = [->, >=stealth, shorten >=1pt, very thick, rounded corners=5]

\begin{scope}[xshift=-10cm]
  \fill [pattern=north east lines, pattern color=gray!25]
        (4,-1) rectangle (8,9);
  \draw [dashed, thin, gray] (4,-1) -- (4,9);
  \draw [dashed, thin, gray] (8,-1) -- (8,9);
  \draw [gray] (4,-1.25) -- (4,-1.5) -- (8,-1.5) -- (8,-1.25);
  \draw [gray] (6,-1.75) node [below] {\footnotesize $L$};

  \draw (1.5,0) node (node0) {};
  \draw (4,0) node [dot] (node1) {};
  \draw (6,0) node (node2) {};
  \draw (6,1) node (node3) {};
  \draw (4,1) node [dot] (node4) {};
  \draw (2,1) node (node5) {};
  \draw (2,2) node (node6) {};
  \draw (4,2) node [dot] (node7) {};
  \draw (4,3) node [dot] (node8) {};
  \draw (10,2) node (node9) {};
  \draw (10,3) node (node10) {};
  \draw (4,3) node [dot] (node12) {};
  \draw (2,3) node (node13) {};
  \draw (2,4) node (node14) {};
  \draw (4,4) node [dot] (node15) {};
  \draw (8,4) node [dot] (node16) {};
  \draw (10,4) node (node17) {};
  \draw (10,5) node (node18) {};
  \draw (8,5) node [dot] (node19) {};
  \draw (6,5) node (node20) {};
  \draw (6,6) node (node21) {};
  \draw (8,6) node [dot] (node22) {};
  \draw (8,7) node  [dot] (node23) {};
  \draw (10,6) node (node24) {};
  \draw (10,7) node (node25) {} ;
  \draw (4,7) node [dot] (node26) {};
  \draw (4,8) node [dot] (node27) {};
  \draw (2,7) node (node28) {};
  \draw (2,8) node (node29) {};
  \draw (7,8) node (node30) {};
  \draw (6,2) node (node77) {};
  \draw (6,3) node (node88) {};

  \draw [factor] 
  (node0) -- (node1);
  \draw [fullfactor] (node1) -- (node2.center) -- 
                     (node3.center) -- (node4); 
  \draw [fullfactor] (node7) -- (node77.center) -- 
  (node88.center) -- (node8);
  
  \draw [factor] (node4) -- (node5.center) -- (node6.center) -- (node7); 
  
  \draw [factor] (node12) -- (node13.center) -- (node14.center) -- (node15);
  \draw [fullfactor] (node15) --  (node16);
  \draw [factor] (node16) -- (node17.center) -- (node18.center) -- (node19);
  \draw [fullfactor] (node19) -- (node20.center) -- 
                     (node21.center) -- (node22);
  \draw [factor] (node22) -- (node24.center) -- (node25.center) -- (node23);
  \draw [fullfactor] (node23) -- (node26) ;
   \draw [factor] (node26) -- (node28.center) -- (node29.center) -- (node27);
   \draw [factor] (node27) -- (node30) ;
\end{scope}

\begin{scope}[xshift=18cm]
  \fill [pattern=north east lines, pattern color=gray!25]
        (4,-1) rectangle (8,9);
  \draw [dashed, thin, gray] (4,-1) -- (4,9);
  \draw [dashed, thin, gray] (8,-1) -- (8,9);
  \draw [gray] (4,-1.25) -- (4,-1.5) -- (8,-1.5) -- (8,-1.25);
  \draw [gray] (6,-1.75) node [below] {\footnotesize $L' \neq L$};

  \draw (1.5,0) node (node0) {};
  \draw (4,0) node [dot] (node1) {};
  \draw (6,0) node (node2) {};
  \draw (6,1) node (node3) {};
  \draw (4,1) node [dot] (node4) {};
  \draw (2,1) node (node5) {};
  \draw (2,2) node (node6) {};
  \draw (4,2) node [dot] (node7) {};
  \draw (4,3) node [dot] (node8) {};
  \draw (4,4) node [dot] (node9) {};
  \draw (4,5) node [dot] (node10) {};
  \draw (6,4) node  (node11) {};
  \draw (6,5) node  (node12) {};
  \draw (2,3) node (node13) {};
  \draw (2,4) node (node14) {};
  \draw (4,6) node [dot] (node15) {};
  \draw (2,5) node (node17) {};
  \draw (2,6) node (node18) {};
  \draw (6,5) node (node20) {};
  \draw (6,6) node (node21) {};
  \draw (10.5,6) node (node23) {};
  \draw (6,2) node (node77) {};
  \draw (6,3) node (node88) {};

\draw (-0.5,4) node { $\equiv$}; 
  \draw [factor] 
  (node0) -- (node1);
  \draw [fullfactor] (node1) -- (node2.center) -- 
                     (node3.center) -- (node4); 
  \draw [factor] (node4) -- (node5.center) -- (node6.center) -- (node7); 
  \draw [fullfactor] (node7) -- (node77.center) --
  (node88.center) -- (node8);
  \draw [factor] (node8) -- (node13.center) -- (node14.center) -- (node9);
  \draw [fullfactor] (node9) -- (node11.center) -- (node12.center) -- (node10);
    \draw [factor] (node10) -- (node17.center) -- (node18.center) -- (node15);
    \draw [factor] (node15) -- (7,6);
\end{scope}

\begin{scope}[xshift=3cm, xscale=0.95]
  \fill [pattern=north east lines, pattern color=gray!25]
        (4,-1) rectangle (8,9);
  \fill [pattern=north east lines, pattern color=gray!25]
        (8,-1) rectangle (12,9);
  \draw [dashed, thin, gray] (4,-1) -- (4,9);
  \draw [dashed, thin, gray] (8,-1) -- (8,9);
  \draw [dashed, thin, gray] (12,-1) -- (12,9);
  \draw [gray] (4,-1.25) -- (4,-1.5) -- (7.9,-1.5) -- (7.9,-1.25);
  \draw [gray] (8.1,-1.25) -- (8.1,-1.5) -- (12,-1.5) -- (12,-1.25);
  \draw [gray] (6,-1.75) node [below] {\footnotesize $L$};
  \draw [gray] (10,-1.75) node [below] {\footnotesize copy of $L$};

  \draw (1.5,0) node (node0) {};
  \draw (4,0) node [dot] (node1) {};
  \draw (6,0) node (node2) {};
  \draw (6,1) node (node3) {};
  \draw (4,1) node [dot] (node4) {};
  \draw (2,1) node (node5) {};
  \draw (2,2) node (node6) {};
  \draw (4,2) node [dot] (node7) {};
  \draw (4,3) node [dot] (node8) {};
  \draw (10,2) node (node9) {};
  \draw (10,3) node (node10) {};
  \draw (4,3) node [dot] (node12) {};
  \draw (2,3) node (node13) {};
  \draw (2,4) node (node14) {};
  \draw (4,4) node [dot] (node15) {};
  \draw (8,4) node [dot] (node16) {};
  \draw (10,4) node (node17) {};
  \draw (10,5) node (node18) {};
  \draw (8,5) node [dot] (node19) {};
  \draw (6,5) node (node20) {};
  \draw (6,6) node (node21) {};
  \draw (8,6) node [dot] (node22) {};
  \draw (8,7) node  [dot] (node23) {};
  \draw (10,6) node (node24) {};
  \draw (10,7) node (node25) {} ;
  \draw (4,7) node [dot] (node26) {};
  \draw (4,8) node [dot] (node27) {};
  \draw (2,7) node (node28) {};
  \draw (2,8) node (node29) {};
  \draw (7,8) node (node30) {};
  \draw (6,2) node (node77) {};
  \draw (6,3) node (node88) {};

  \draw [factor] 
  (node0) -- (node1);
  \draw [fullfactor] (node1) -- (node2.center) -- 
                     (node3.center) -- (node4); 
  \draw [fullfactor] (node7) -- (node77.center) -- 
  (node88.center) -- (node8);
  
  \draw [factor] (node4) -- (node5.center) -- (node6.center) -- (node7); 
  
  \draw [factor] (node12) -- (node13.center) -- (node14.center) -- (node15);
  \draw [fullfactor] (node15) --  (node16);
  \draw [fullfactor] (node16) -- (node17.center) -- (node18.center) -- (node19);
  \draw [fullfactor] (node19) -- (node20.center) -- 
                     (node21.center) -- (node22);
  \draw [fullfactor] (node22) -- (node24.center) -- (node25.center) -- (node23);
  \draw [fullfactor] (node23) -- (node26) ;
   \draw [factor] (node26) -- (node28.center) -- (node29.center) -- (node27);
   \draw [factor] (node27) -- (node30) ;
\end{scope}

\end{tikzpicture}
\caption{Pumping a loop $L$ with a wrong shape and showing it is not
  idempotent.}\label{fig:notidem} 
\end{figure}

\begin{proof}
Suppose that $C$ is a left-to-right component of $L$.
We show by way of contradiction that $C$ has only one $\LR$-factor
and no $\RL$-factor. By Lemma~\ref{lem:component} this will yield to
the claimed shape. Fig.~\ref{fig:notidem} can be used as a reference
example for the arguments that follow.

We begin by listing the $(L,C)$-factors.
As usual, we order them based on their occurrences in the run $\rho$.
Let $\gamma$ be the first $(L,C)$-factor that is not an $\LL$-factor, 
and let $\beta_1,\dots,\beta_k$ be the $(L,C)$-factors that precede $\gamma$ 
(these are all $\LL$-factors).
Because $\gamma$ starts at an even level, it must be an $\LR$-factor.
Suppose that there is another $(L,C)$-factor, say $\zeta$, that comes 
after $\gamma$ and it is neither an $\RR$-factor nor an $\LL$-factor. 
Because $\zeta$ starts at an odd level, it must be an $\RL$-factor. 
Further let $\delta_1,\dots,\delta_{k'}$ be the intercepted $\RR$-factors 
that occur between $\gamma$ and $\zeta$.
We claim that $k'<k$, namely, that the number of $\RR$-factors between 
$\gamma$ and $\zeta$ is strictly less than the number of $\LL$-factors 
before $\gamma$. Indeed, if this were not the case, then, by Lemma 
\ref{lem:component}, the level where $\zeta$ starts would not belong to
the component $C$.

Now, consider the pumped run $\rho'=\pump_L^2(\rho)$, obtained by adding a new
copy of $L$. Let $L'$ be the loop of $\rho'$ obtained from the union
of $L$ and its copy. Since $L$ is idempotent, the components of $L$ are isomorphic
to the components of $L'$. In particular, we can denote by $C'$ the component
of $L'$ that is isomorphic to $C$. 
Let us consider the $(L',C')$-factors of $\rho'$. The first $k$ such factor 
are isomorphic to the $k$ $\LL$-factors $\beta_1,\dots,\beta_k$ from $\rho$.
However, the $(k+1)$-th element has a different shape: it is isomorphic to
$\gamma~\beta_1~\delta_1~\beta_2~\dots~\delta_{k'}~\beta_{k'+1}~\zeta$,
and in particular it is an $\LL$-factor.
This implies that the $(k+1)$-th edge of $C'$ is of the form $(y,y+1)$,
while the $(k+1)$-th edge of $C$ is of the form $(y,y-2k)$. 
This contradiction comes from having assumed the existence of 
the $\RL$-factor $\zeta$, and is illustrated in Fig.~\ref{fig:notidem}.
\end{proof}

The following lemma will be used to prove Theorem~\ref{thm:simon2}.

\begin{lemma}\label{lem:adjacent-factors}
If $L_1=[x_1,x_2]$ and $L_2=[x_2,x_3]$ are consecutive idempotent loops 
with the same effect and $\alpha,\beta$ are two factors intercepted by 
$L_1,L_2$ that are adjacent in the run (namely, they share the endpoint at position $x_2$),
then $\alpha$ and $\beta$ correspond to edges of the same component of $L_1$
(or, equally, $L_2$).
\end{lemma}

\begin{proof}
Let $C$ be the component of $L_1$ and $(y,y')$ the edge of $C$ that corresponds to the factor $\alpha$
intercepted by $L_1$. Similarly, let $C'$ be the component of $L_2$ and $(y'',y''')$ the edge of $C'$
that corresponds to the factor $\beta$ intercepted by $L_2$. Since $\alpha$ and $\beta$ share the
endpoint at position $x_2$, we know that $y'=y''$. This shows that $C \cap C' \neq \emptyset$, and 
hence $C=C'$.
\end{proof}

\begin{repproposition}{prop:pumping}
Let $L$ be an idempotent loop of $\rho$ with components $C_1,\dots,C_k$, 
listed according to the order of their anchors:
$\an{C_1}\lesstime\cdots\lesstime\an{C_k}$. 
For all $m\in\bbN$, we have 
\[
  \pump_L^{m+1}(\rho) ~=~ 
  \rho_0 ~ \tr{C_1}^m ~ \rho_1 ~ \cdots ~ \rho_{k-1} ~ \tr{C_k}^m ~ \rho_k
\]
where 
\begin{itemize}
  \item $\rho_0$ is the prefix of $\rho$ that ends at  $\an{C_1}$, 
  \item $\rho_i$ is the factor of $\rho$ between  $\an{C_i}$ and $\an{C_{i+1}}$, 
        for all $i=1,\dots,k-1$,
  \item $\rho_k$ is the suffix of $\rho$ that starts at  $\an{C_k}$.
\end{itemize}
\end{repproposition}

\begin{proof}
Along the proof we sometimes refer to Fig.~\ref{fig:loop-trace} to ease
the intuition of some definitions and arguments.
Let $L=[x_1,x_2]$ be an idempotent loop and, for all $i=0,\dots,m$, let 
$L'_i=[x'_i,x'_{i+1}]$ be the $i$-th copy of the loop $L$ in the pumped run
$\rho'=\pump_L^{m+1}(\rho)$, where $x'_i = x_1 + i\cdot (x_2-x_1)$
(the ``$0$-th copy of $L$'' is the loop $L$ itself). 
Further let $L'=L'_0\cup\dots\cup L'_m = [x'_0,x'_{m+1}]$, that is, 
$L'$ is the loop of $\rho'$ that spans across the $m+1$ occurrences of $L$.
As $L$ is idempotent, the loops $L'_0,\dots,L'_m$ and $L'$ have all the 
same effect as $L$.
In particular, the components of $L'_0,\dots,L'_m$, and $L'$ are isomorphic 
to and in same order as those of $L$.
We denote these components by $C_1,\dots,C_k$.

We let $\ell_j=\an{C_j}$ be the anchor of each component $C_j$ inside 
the loop $L$ of $\rho$
(these locations are marked by black dots in the left hand-side 
of Fig.~\ref{fig:loop-trace}).
Similarly, we let $\ell'_{i,j}$ (resp.~$\ell'_j$) be the anchor of 
$C_j$ inside the loop $L'_i$ (resp.~$L'$).
From Definition \ref{def:anchor}, we have that either $\ell'_j=\ell'_{1,j}$ or $\ell'_j=\ell'_{m,j}$, 
depending on whether $C_j$ is left-to-right or right-to-left (or, equally, on whether $j$ is odd or even).

Now, let us consider the factorization of the pumped run $\rho'$ 
induced by the locations $\ell'_{i,j}$, for all $i=0,\dots,m$ and for $j=1,\dots,k$ 
(these locations are marked by black dots in the right hand-side of the figure).
By construction, the prefix of $\rho'$ that ends at location $\ell'_{0,1}$ 
coincides with the prefix of $\rho$ that ends at $\ell_1$, 
i.e.~$\rho_0$ in the statement of the proposition.
Similarly, the suffix of $\rho'$ that starts at location $\ell'_{m,k}$ is isomorphic to the
suffix of $\rho$ that starts at $\ell_k$, i.e. $\rho_k$ in the statement.
By construction, we also know that, for all odd (resp.~even) indices $j$, the factor
$\rho'[\ell'_{m,j},\ell'_{m,j+1}]$ (resp.~$\rho'[\ell_{0,j},\ell_{0,j+1}]$) is isomorphic 
to $\rho[\ell_j,\ell_{j+1}]$, i.e.~the $\rho_j$ of the statement. 

The remaining factors of $\rho'$ are those delimited by the pairs of locations 
$\ell'_{i,j}$ and $\ell'_{i+1,j}$, for all $i=0,\dots,m-1$ and all $j=1,\dots,k$. 
Consider one such factor $\rho'[\ell'_{i,j},\ell'_{i+1,j}]$, 
and assume that the index $j$ is odd (the case of an even $j$ is similar). 
This factor can be seen as a concatenation of factors intercepted by $L$ 
that correspond to edges of $C_j$ inside $L'_i$. 
More precisely, $\rho'[\ell'_{i,j},\ell'_{i+1,j}]$ is obtained by concatenating 
the unique $\LR$-factor of $C_j$ --- recall that by Lemma \ref{lem:component2} 
there is exactly one such factor --- with an interleaving of the $\LL$-factors 
and the $\RR$-factors of $C_j$. 
As the components are the same for all $L'_i$'s, this corresponds precisely to 
the trace $\tr{C_j}$ (cf.~Definition \ref{def:trace}).
Now that we know that $\rho'[\ell'_{i,j},\ell'_{i+1,j}]$ is isomorphic to $\tr{C_j}$, 
we can conclude that
$\rho'[\ell'_{0,j},\ell'_{m,j}] \:=\: 
 \rho'[\ell'_{0,j},\ell'_{1,j}] ~ \dots ~ \rho'[\ell'_{m-1,j},\ell'_{m,j}]$
is isomorphic to $\tr{C_j}^m$. 
\end{proof}

\begin{reptheorem}{thm:simon2}
Let $I=[x_1,x_2]$ be an interval of positions, $K=[\ell_1,\ell_2]$ 
an interval of locations, and $Z = K \:\cap\: (I\times\bbN)$.
If $\big|\out{\rho\mid Z}\big| > \bound$, 
then there exist an idempotent loop $L$ and a component $C$ of $L$ such that
\begin{itemize}
  \item $x_1 < \min(L) < \max(L) < x_2$ (in particular, $L\subsetneq I$),
  \item $\ell_1\lesstime\an{C}\lesstime\ell_2$ (in particular, $\an{C} \in K$),
  \item $\out{\tr{C}} \neq \emptystr$.
\end{itemize}
\end{reptheorem}

\begin{proof}
Let $I$, $K$, $Z$ be as in the statement, 
and suppose that $\big|\out{\rho\mid Z}\big| > \bound$.
We define 
$Z' = Z ~\setminus~ \{\ell_1,\ell_2\} ~\setminus~ \big(\{x_1,x_2\})\times\bbN\big)$
and we observe that there are at most $2\hmax$ locations in $Z$ 
that are missing from $Z'$. This means that $\rho\mid Z'$ contains 
all but $4\hmax$ transitions of $\rho\mid Z$, and because each 
transition outputs at most $\cmax$ letters, we have
$\big|\out{\rho\mid Z'}\big| > \bound - 4\cmax\cdot\hmax = \cmax\cdot\hmax\cdot 2^{3\emax}$.

For every level $y$, let $X_y$ be the set of positions $x$ such that
$(x,y)$ is the source location of a transition of $\rho\mid Z'$ 
that produces non-empty output.
For example, if we refer to Fig.~\ref{fig:ramsey-proof},
the vertical dashed lines represent the positions 
of $X_y$ for a particular level $y$; accordingly, the circles 
in the figure represent the locations of the form $(x,y)$, for 
$x\in X_y$.
Since each transition outputs at most $\cmax$ letters, 
we have $\sum_y |X_y| > \hmax\cdot 2^{3\emax}$.
Moreover, since there are at most $\hmax$ levels, 
there is a level $y$ (which we fix hereafter) such that $|X_y| > 2^{3\emax}$.

\begin{claim}
There are two consecutive loops $L_1=[x,x']$ and $L_2=[x',x'']$ 
such that $E_{L_1}=E_{L_2}=E_{L_1\cup L_2}$
and with endpoints $x,x',x''\in X_y$.
\end{claim}

\begin{claimproof}
By Theorem \ref{th:simon}, 
there is a factorization forest for $X$ of height at most $3\emax$.
Since $|X_y| > 2^{3\emax}$, we know that this factorization forest contains 
an internal node $L'=[x'_1,x'_{k+1}]$ with $k > 2$ children, say 
$L_1=[x'_1,x'_2]$, \dots $L_k=[x'_k,x'_{k+1}]$. 
By definition of factorization forest, the effects 
$E_{L'}$, $E_{L_1}$, \dots, $E_{L_k}$ are all equal and idempotent. 
Moreover, as $\rho$ is a valid run, the dummy element $\bot$ of the 
effect semigroup does not appear in the factorization forest. 
In particular, the effect $E_{L'}=E_{L_1}=\dots=E_{L_k}$ 
is a triple of the form $(F_{L'},c_1,c_2)$,
where $c_i=\rho|x_i$ is the crossing sequence at $x'_i$. 
Finally, since $E_{L'}$ is idempotent, we have that $c_1 = c_2$ and 
this is equal to the crossing sequences of $\rho$ at the positions 
$x'_1,\dots,x'_{k+1}$. This shows that $L_1,L_2$ are idempotent loops.
\end{claimproof}

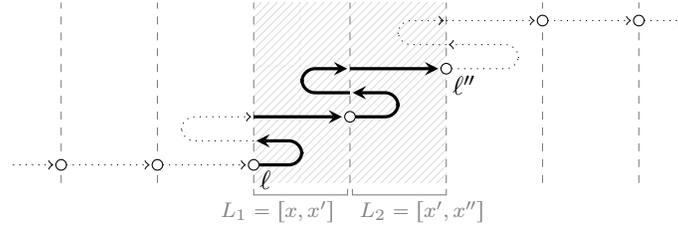
\begin{figure}[!t]
\centering
\begin{tikzpicture}[baseline=0, inner sep=0, outer sep=0, minimum size=0pt, scale=0.32]
  \tikzstyle{dot} = [draw, circle, fill=white, minimum size=4pt]
  \tikzstyle{fulldot} = [draw, circle, fill=black, minimum size=4pt]
  \tikzstyle{factor} = [->, shorten >=1pt, dotted, rounded corners=5]
  \tikzstyle{fullfactor} = [->, >=stealth, shorten >=1pt, very thick, rounded corners=5]

  \fill [pattern=north east lines, pattern color=gray!25]
        (10,-0.75) rectangle (18,6.75);
  \draw [dashed, thin, gray] (2,-0.75) -- (2,6.75);
  \draw [dashed, thin, gray] (6,-0.75) -- (6,6.75);
  \draw [dashed, thin, gray] (10,-0.75) -- (10,6.75);
  \draw [dashed, thin, gray] (14,-0.75) -- (14,6.75);
  \draw [dashed, thin, gray] (18,-0.75) -- (18,6.75);
  \draw [dashed, thin, gray] (22,-0.75) -- (22,6.75);
  \draw [dashed, thin, gray] (26,-0.75) -- (26,6.75);
  \draw [gray] (10,-1) -- (10,-1.25) -- (13.9,-1.25) -- (13.9,-1);
  \draw [gray] (14.1,-1) -- (14.1,-1.25) -- (18,-1.25) -- (18,-1);
  \draw [gray] (11,-1.5) node [below] {\footnotesize $L_1=[x,x']$};
  \draw [gray] (17,-1.5) node [below] {\footnotesize $L_2=[x',x'']$};

  \draw (0,0) node (node0) {};
  \draw (2,0) node [dot] (node1) {};
  \draw (6,0) node [dot] (node2) {};
  \draw (10,0) node [dot] (node3) {};
  \draw (12,0) node (node4) {};
  \draw (12,1) node (node5) {};
  \draw (10,1) node (node6) {};
  \draw (7,1) node (node7) {};
  \draw (7,2) node (node8) {};
  \draw (10,2) node (node9) {};
  \draw (14,2) node [dot] (node10) {};
  \draw (16,2) node (node11) {};
  \draw (16,3) node (node12) {};
  \draw (14,3) node (node13) {};
  \draw (12,3) node (node14) {};
  \draw (12,4) node (node15) {};
  \draw (14,4) node (node16) {};
  \draw (18,4) node [dot] (node17) {};
  \draw (21,4) node (node18) {};
  \draw (21,5) node (node19) {};
  \draw (18,5) node (node20) {};
  \draw (16,5) node (node21) {};
  \draw (16,6) node (node22) {};
  \draw (18,6) node (node23) {};
  \draw (22,6) node [dot] (node24) {};
  \draw (26,6) node [dot] (node25) {};
  \draw (28,6) node (node26) {};

  \draw [factor] (node0) -- (node1);
  \draw [factor] (node1) -- (node2);
  \draw [factor] (node2) -- (node3);
  \draw [fullfactor] (node3) -- (node4.center) -- (node5.center) -- (node6);
  \draw [factor] (node6) -- (node7.center) -- (node8.center) -- (node9);
  \draw [fullfactor] (node9) -- (node10);
  \draw [fullfactor] (node10) -- (node11.center) -- (node12.center) -- (node13);
  \draw [fullfactor] (node13) -- (node14.center) -- (node15.center) -- (node16);
  \draw [fullfactor] (node16) -- (node17);
  \draw [factor] (node17) -- (node18.center) -- (node19.center) -- (node20);
  \draw [factor] (node20) -- (node21.center) -- (node22.center) -- (node23);
  \draw [factor] (node23) -- (node24);
  \draw [factor] (node24) -- (node25);
  \draw [factor] (node25) -- (node26);
  
  \draw (node3) node [below right = 1.2mm] {$\ell$};
  \draw (node17) node [below right = 1.2mm] {$\ell''$};
\end{tikzpicture}
\caption{Two consecutive idempotent loops with the same effect.}%
\label{fig:ramsey-proof}
\end{figure}

Turning back to the proof of the theorem, we know that there are two 
consecutive idempotent loops $L_1=[x,x']$ and $L_2=[x',x'']$ with the 
same effect and with endpoints $x,x',x''\in X_y \subseteq I \:\setminus\: \{x_1,x_2\}$
(see again Fig.~\ref{fig:ramsey-proof}).

Let $\ell=(x,y)$ and $\ell''=(x'',y)$, and observe that both locations
belong to $Z'$. In particular, $\ell$ and $\ell''$ are strictly between $\ell_1$ and $\ell_2$.
Suppose by symmetry that $\ell\leqtime\ell''$.
Further let $C$ be the component of $L_1\cup L_2$ (or, equally, 
of $L_1$ or $L_2$) that contains the node $y$.
Below, we focus on the factors of $\rho[\ell,\ell'']$ that are intercepted by 
$L_1\cup L_2$: these are represented in Fig.~\ref{fig:ramsey-proof} by 
the thick arrows. 
By Lemma~\ref{lem:component2} all these factors correspond to edges
of the same component $C$, namely, they are $(L_1 \cup L_2,C)$-factors.

Consider any factor $\alpha$ of $\rho[\ell,\ell'']$ intercepted by
$L_1\cup L_2$, and assume that $\a=\b_1 \cdots \b_k$, where
$\beta_1,\dots,\beta_k$ are the factors intercepted by $L_1$
or $L_2$. By Lemma \ref{lem:adjacent-factors}, any two adjacent
factors $\beta_i,\beta_{i+1}$ correspond to edges in the same
component of $L_1$ and $L_2$, respectively. Thus, by transitivity, 
all factors $\beta_1,\dots,\beta_k$ correspond to edges in the same
component, say $C'$. 
We claim that $C'=C$. Indeed, if $\b_1$ is intercepted by $L_1$, 
then $C'=C$ because $\alpha$ and $\b_1$ start from the same location
and hence they correspond to edges of the flow that depart from the 
same node. The other case is where $\b_1$ is intercepted by $L_2$,
for which a symmetric argument can be applied.

So far we have shown that every factor of $\rho[\ell,\ell']$ intercepted by $L_1\cup L_2$
can be factorized into some $(L_1,C)$-factors and some $(L_2,C)$-factors.
We conclude the proof with the following observations:
\begin{itemize}
  \item By construction, both loops $L_1,L_2$ are contained in the interval of positions $I=[x_1,x_2]$,
        and have endpoints different from $x_1,x_2$.
  \item Both anchors of $C$ inside $L_1$ and $L_2$ belong to the interval of locations 
        $K\:\setminus\:\{\ell_1,\ell_2\}$. 
        This holds because
        $\rho[\ell,\ell']$ contains a factor $\alpha$ that is intercepted by $L_1\cup L_2$
        and spans across all the positions from $x$ to $x''$, namely,
        an $\LR$-factor. 
        This factor starts at the anchor of $C$ inside $L_1$ 
        and visits the anchor of $C$ inside $L_2$. 
        Moreover, by construction, $\alpha$ is also a factor of the subsequence $\rho\mid Z'$.
        This shows that the anchors of $C$ inside $L_1$ and $L_2$ belong to $Z'$, 
        and in particular to $K\:\setminus\:\{\ell_1,\ell_2\}$.
  \item The first factor of $\rho[\ell,\ell']$ that is intercepted by $L_1\cup L_2$
        starts at $\ell=(x,y)$, which by construction is the source location of some
        transition producing non-empty output.
        By the previous arguments, this factor is a concatenation of $(L_1,C)$-factors 
        and $(L_2,C)$-factors. This implies that the trace of $C$ inside $L_1$ or the 
        trace of $C$ inside $L_2$ produces non-empty output.
\qedhere
\end{itemize}        
\end{proof}

\begin{repproposition}{prop:inversion}
If $\cT$ is one-way definable, then for every inversion $(L_1,C_1,L_2,C_2)$ 
of a successful run $\rho$ of $\cT$, the word
\[
  \outb{\tr{C_1}} ~ \outb{\rho[\an{C_1},\an{C_2}]} ~ \outb{\tr{C_2}}
\]
has period $p$ that divides both $|\out{\tr{C_1}}|$ and $|\out{\tr{C_2}}|$.
Moreover, $p \le \bound$.
\end{repproposition}

\begin{proof}[Proof of Proposition \ref{prop:inversion}]
The proof of the first claim of the proposition 
is similar to the proof of Proposition 7 in~\cite{bgmp15} for
sweeping transducers. The main difficulty in the present proof is to
get a bound on the period  of the output of the
inversion. 

Let $(L_1,C_1,L_2,C_2)$ be an inversion of a successful run $\rho$
on input $u$.
Note that the two loops $L_1$ and $L_2$ might not be disjoint. 
In fact, two cases arise: 
either $\max(L_2) < \min(L_1)$ (that is, $L_1$ and $L_2$ are 
disjoint and $L_2$ is strictly to the left of $L_1$), or 
$\min(L_1) \le \min(L_2) \le \max(L_1) \le \max(L_2)$ 
(the fact that $\min(L_2) \le \max(L_1)$ follows from the fact that
the anchor $\an{C_2}$ is to the left of the anchor $\an{C_1}$).
For the sake of simplicity, we only deal with the case where 
$L_1$ and $L_2$ are disjoint, as shown in Fig.~\ref{fig:inversion}
--- the other case can be treated in a similar way by considering 
the rightmost copy of $L_1$ in the pumped run $\pump_{L_1}^3(\rho)$, 
which is clearly disjoint from the leftmost copy of $L_2$.

We begin by pumping the run $\rho$, together with the underlying input $u$, 
on the loops $L_1$ and $L_2$. Formally, for all numbers $m_1,m_2\in\bbN$, 
we define 
\[
\begin{array}{rcl}
  u^{(m_1,m_2)}    &=& \pump_{L_1}^{m_1+1}(\pump_{L_2}^{m_2}(u))  \\[1ex]
  \rho^{(m_1,m_2)} &=& \pump_{L_1}^{m_1+1}(\pump_{L_2}^{m_2}(\rho)).
\end{array}
\]
We identify the positions that mark the endpoints of the occurrences of 
$L_1$ and $L_2$ in the pumped run $\rho^{(m_1,m_2)}$. 
Formally, if $L_1=[x_1,x_2]$ and $L_2=[x_3,x_4]$, then the sets of positions 
are defined as follows:
\[
\begin{array}{rcl}
   X_2^{(m_1,m_2)} &=& \big\{ x_3 + i\cdot (x_4-x_3) ~:~ i=0,\dots,m_2+1 \big\} \\[1ex]
   X_1^{(m_1,m_2)} &=& \big\{ x_1 + i\cdot (x_2-x_1) + m_2\cdot(x_4-x_3) ~:~ i=0,\dots,m_1+1 \big\}. 
\end{array}
\]
Let $\Tt'$ be a  one-way  transducer equivalent to $\Tt$, and
consider a successful run $\lambda^{(m_1,m_2)}$ of $\cT'$ on the input $u^{(m_1,m_2)}$.
Since $\cT'$ has finitely many states, we can find a large enough number 
$m$ and two positions $x'_1<x'_2$ both in $X_1^{(m,m)}$, such that $L'_1=[x'_1,x'_2]$
is a loop of $\lambda^{(m,m)}$. Similarly, we can find two positions 
$x'_3<x'_4$ both in $X_2^{(m,m)}$, such that $L'_2=[x'_3,x'_4]$ is a loop of 
$\lambda^{(m,m)}$.
Clearly, $L'_1$ and $L'_2$ are also loops of $\rho^{(m,m)}$: 
indeed, $L'_1$ (resp.~$L'_2$) consists of
$k_1\le m$ (resp.~$k_2\le m$) copies of $L_1$ (resp.~$L_2$) in $\rho^{(m,m)}$.
In particular, for all $m_1,m_2\in\bbN$ we have:
\[
\begin{array}{rcl}
  \pump_{L'_1}^{m_1+1}(\pump_{L'_2}^{m_2+1}(u^{(m,m)})) 
  &=& u^{(f(m_1),g(m_2))} \\[1ex]
  \pump_{L'_1}^{m_1+1}(\pump_{L'_2}^{m_2+1}(\rho^{(m,m)})) 
  &=& \rho^{(f(m_1),g(m_2))} \\[1ex]
  \pump_{L'_1}^{m_1+1}(\pump_{L'_2}^{m_2+1}(\lambda^{(m,m)})) 
  &=& \lambda^{(f(m_1),g(m_2))}.
\end{array}
\]
where $f(m_1)=k_1\cdot m_1+m$, $g(m_2)=k_2\cdot m_2+m$.

Now we observe that the run $\lambda^{(f(m_1),g(m_2))}$ 
of $\cT'$ produces the same output as the run 
$\rho^{(f(m_1),g(m_2))}$ of $\cT$
--- this holds thanks to the fact that the transducers are functional, 
otherwise it may happen that the pumped runs $\lambda^{(f(m_1),g(m_2))}$ 
and $\rho^{(f(m_1),g(m_2))}$ produce different outputs.
Let us denote this output by $w^{(f(m_1),g(m_2))}$. 
Below, we show two possible factorizations of $w^{(f(m_1),g(m_2))}$
based on the shapes of the pumped runs $\lambda^{(f(m_1),g(m_2))}$ and
$\rho^{(f(m_1),g(m_2))}$.
For the first factorization, we recall that $L'_2$ precedes $L'_1$, 
according to the ordering of positions, 
and that the run $\lambda^{(f(m_1),g(m_2))}$ 
is left-to-right. We thus obtain
\begin{equation}\label{eq:one-way-eq}
  w^{(f(m_1),g(m_2))} ~=~ w_0 ~ \pmb{w_1^{m_2}} ~ w_2 ~ \pmb{w_3^{m_1}} ~ w_4
\end{equation}
where
\begin{itemize}
  \item $w_0$ is the output produced by the prefix of $\lambda^{(m,m)}$ 
        up to the left border of $L'_2$,
  \item $w_1$ is the output produced by the (unique) factor of $\lambda^{(m,m)}$ 
        intercepted by $L'_2$,
  \item $w_2$ is the output produced by the factor of $\lambda^{(m,m)}$ 
        between the right border of $L'_2$ and the left border of $L'_1$,
  \item $w_3$ is the output produced by the (unique) factor of $\lambda^{(m,m)}$ intercepted by $L'_1$,
  \item $w_4$ is the output produced by the suffix of
    $\lambda^{(m,m)}$ after the right border of $L'_1$.
\end{itemize}

For the second factorization, we consider $L'_1$ and $L'_2$ as loops of $\rho^{(m,m)}$.
We denote by $\ell'_1$ (resp.~$\ell'_2$) the anchor of the component $C_1$ (resp.~$C_2$) 
of $L'_1$ (resp.~$L'_2$). By assumption we have
$\ell'_1\leqtime\ell'_2$. Applying
Proposition \ref{prop:pumping} we get:
\begin{equation}\label{eq:two-way-eq}
  w^{(f(m_1),g(m_2))} ~=~ 
  v_0^{(m_1,m_2)} ~ \pmb{v_1^{m_1}} ~ v_2^{(m_1,m_2)} ~ \pmb{v_3^{m_2}} ~ v_4^{(m_1,m_2)}
\end{equation}
where 
\begin{itemize}
  \item $v_0^{(m_1,m_2)}$ is the output produced by the prefix of $\rho^{(m,m)}$ 
        that ends at $\ell'_1$
        (note that this word may depend on the parameters $m_1,m_2$, since the loops
        $L'_1$ and $L'_2$ may be traversed several times before reaching the location $\ell'_1$),
  \item $v_1=\out{\tr{C_1}}$
        (this word does not depend on $m_1,m_2$), 
  \item $v_2^{(m_1,m_2)}$ is the output produced by the factor of $\rho^{(m,m)}$ 
        between $\ell'_1$ and $\ell'_2$,
  \item $v_3=\out{\tr{C_2}}$,
  \item $v_4^{(m_1,m_2)}$ is the output produced by the suffix of $\rho^{(m,m)}$ 
        that starts at $\ell'_2$.
\end{itemize}

Putting together Eqs.~(\ref{eq:one-way-eq}) and (\ref{eq:two-way-eq}),
we get
\begin{equation}\label{eq:one-way-vs-two-way}
  v_0^{(m_1,m_2)} ~ \pmb{v_1^{m_1}} ~ v_2^{(m_1,m_2)} ~ \pmb{v_3^{m_2}} ~ v_4^{(m_1,m_2)}
  ~~=~~ 
  w_0 ~ \pmb{w_1^{m_2}} ~ w_2 ~ \pmb{w_3^{m_1}} ~ w_4 .
\end{equation}
We recall that the words $v_1,v_3$ are non-empty, since they are outputs
of traces of components that form an inversion.
This allows us to apply Lemma \ref{lemma:one-way-vs-two-way}, 
which shows that the word 
$\pmb{v_1} ~ \pmb{v_1^{m_1}} ~ v_2^{(m_1,m_2)} ~ \pmb{v_3^{m_2}} ~ \pmb{v_3}$ 
has period $\gcd(|v_1|,|v_3|)$, for all $m_1,m_2\in\bbN$.
Note that the latter period still depends on $\cT'$, since 
the words $v_1$ and $v_3$ were constructed from the loops $L'_1$ and
$L'_2$, that are both loops of 
the run $\lambda^{(m,m)}$ of $\cT'$. However,  
Proposition~\ref{prop:pumping} tells us that
the word $v_1$ (resp.~$v_3$) is an iteration of the output
$\out{\tr{C_1}}$ of the component $C_1$ of $L_1$
(resp.~the output $\out{\tr{C_2}}$ of the component $C_2$ of $L_2$). 
By Lemma~\ref{lemma:fine-wilf}, this implies that the period of 
$\pmb{v_1} ~ \pmb{v_1^{m_1}} ~ v_2^{(m_1,m_2)} ~ \pmb{v_3^{m_2}} ~ \pmb{v_3}$ 
divides both $|\out{\tr{C_1}}|$ and $|\out{\tr{C_2}}|$.

In a similar way, we recall from Proposition \ref{prop:pumping}
that all the words
$\pmb{v_1} ~ \pmb{v_1^{m_1}} ~ v_2^{(m_1,m_2)} ~ \pmb{v_3^{m_2}} ~ \pmb{v_3}$ 
are obtained by iterating suitable factors inside
$\out{\tr{C_1}} ~ \out{\rho[\an{C_1},\an{C_2}]} ~ \out{\tr{C_2}}$:
more precisely, by iterating $n_1$ (resp.~$n_2$) times the output 
traces of the components of $L_1$  (resp.~of $L_2$), where 
$n_1 = f(m_1)$ (resp.~$n_2 = g(m_2)$).
Since the periodicity property holds for infinitely many $n_1$
and, independently, for infinitely many $n_2$,
we know from Theorem \ref{thm:saarela} that 
it also holds for all $n_1,n_2\in\bbN$, and in particular, for $n_1=n_2=0$.
This allows us to conclude that the word
\[
  \outb{\tr{C_1}} ~ \outb{\rho[\an{C_1},\an{C_2}]} ~ \outb{\tr{C_2}}
\]
is periodic with period $p$ that divides both $|\out{\tr{C_1}}|$ and $|\out{\tr{C_2}}|$.

\medskip
It remains to prove the second claim of the proposition, which bounds 
the period by the constant $\bound$. 
This requires a refinement of the previous arguments that involves
pumping the run $\rho$ simultaneously on three different loops.

\begin{figure}[!t]

\centering
\begin{tikzpicture}[baseline=0, inner sep=0, outer sep=0, minimum size=0pt, scale=0.3, xscale=0.7]
  \tikzstyle{dot} = [draw, circle, fill=white, minimum size=4pt]
  \tikzstyle{fulldot} = [draw, circle, fill=black, minimum size=4pt]
  \tikzstyle{factor} = [->, shorten >=1pt, dotted, rounded corners=5]
  \tikzstyle{fullfactor} = [->, >=stealth, shorten >=1pt, very thick, rounded corners=5]

\begin{scope}
  \fill [pattern=north east lines, pattern color=gray!25]
        (4,-0.75) rectangle (10,9);
  \draw [dashed, thin, gray] (4,-0.75) -- (4,9);
  \draw [dashed, thin, gray] (10,-0.75) -- (10,9);
  \draw [gray] (4,-1) -- (4,-1.25) -- (10,-1.25) -- (10,-1);
  \draw [gray] (7,-1.5) node [below] {\footnotesize $L_2$};

  \fill [pattern=north east lines, pattern color=gray!25]
        (16,-0.75) rectangle (22,9);
  \draw [dashed, thin, gray] (16,-0.75) -- (16,9);
  \draw [dashed, thin, gray] (22,-0.75) -- (22,9);
  \draw [gray] (16,-1) -- (16,-1.25) -- (22,-1.25) -- (22,-1);
  \draw [gray] (19,-1.5) node [below] {\footnotesize $L_1$};

  \draw (0,0) node (node0) {};
  \draw (4,0) node [dot] (node1) {};
  \draw (10,0) node [dot] (node2) {};
  \draw (16,0) node [dot] (node3) {};
  \draw (22,0) node [dot] (node4) {};
  \draw (24,0) node (node5) {};
  \draw (24,1) node (node6) {};

  \draw (22,1) node [dot] (node7) {};
  \draw (19,1) node (node8) {};
  \draw (19,2) node (node9) {};
  \draw (22,2) node (node10) {};
  \draw (25,2) node (node11) {};
  \draw (25,3) node (node12) {};
  \draw (22,3) node [fulldot] (node13) {};
  \draw (16,3) node [dot] (node14) {};

  \draw (10,3) node [dot] (node15) {};
  \draw (4,3) node [dot] (node16) {};
  \draw (2,3) node (node17) {};
  \draw (2,4) node (node18) {};

  \draw (4,4) node [dot] (node19) {};
  \draw (7,4) node (node20) {};
  \draw (7,5) node (node21) {};
  \draw (4,5) node [dot] (node22) {};
  \draw (1,5) node (node23) {};
  \draw (1,6) node (node24) {};
  \draw (4,6) node [fulldot] (node25) {};
  \draw (10,6) node [dot] (node26) {};

  \draw (16,6) node [dot] (node27) {};
  \draw (20,6) node (node28) {};
  \draw (20,7) node (node29) {};
  \draw (16,7) node [dot] (node30) {};

  \draw (10,7) node [dot] (node31) {};
  \draw (6,7) node (node32) {};
  \draw (6,8) node (node33) {};
  \draw (10,8) node [dot] (node34) {};

  \draw (16,8) node [dot] (node35) {};
  \draw (22,8) node [dot] (node36) {};
  \draw (26,8) node (node37) {};

  \draw [factor] (node0) -- (node1);
  \draw [factor] (node1) -- (node2);
  \draw [factor] (node2) -- (node3);
  \draw [factor] (node3) -- (node4);
  \draw [factor] (node4) -- (node5.center) -- (node6.center) -- (node7);
  
  \draw [fullfactor] (node7) -- (node8.center) -- (node9.center) -- (node10);
  \draw [factor] (node10) -- (node11.center) -- (node12.center) -- (node13);
  \draw [fullfactor] (node13) -- (node14);

  \draw [factor] (node14) -- (node15);
  \draw [factor] (node15) -- (node16);
  \draw [factor] (node16) -- (node17.center) -- (node18.center) -- (node19);

  \draw [fullfactor] (node19) -- (node20.center) -- (node21.center) -- (node22);
  \draw [factor] (node22) -- (node23.center) -- (node24.center) -- (node25);
  \draw [fullfactor] (node25) -- (node26);
  
  \draw [factor] (node26) -- (node27);
  \draw [fullfactor, red] (node27) -- (node28.center) -- (node29.center) -- (node30);
  \draw [factor] (node30) -- (node31);
  \draw [fullfactor, red] (node31) -- (node32.center) -- (node33.center) -- (node34);

  \draw [factor] (node34) -- (node35);
  \draw [factor] (node35) -- (node36);
  \draw [factor] (node36) -- (node37);

  \draw (node13) node [above right=1.5mm] {\footnotesize $\an{C_1}$};
  \draw (node25) node [above left=1.5mm] {\footnotesize $\an{C_2}$};
\end{scope}

\begin{scope}[xshift=32cm]
  \fill [pattern=north east lines, pattern color=gray!25]
        (4,-0.75) rectangle (10,11);
  \draw [dashed, thin, gray] (4,-0.75) -- (4,11);
  \draw [dashed, thin, gray] (10,-0.75) -- (10,11);
  \draw [gray] (4,-1) -- (4,-1.25) -- (10,-1.25) -- (10,-1);
  \draw [gray] (7,-1.5) node [below] {\footnotesize $L_2$};

  \fill [pattern=north east lines, pattern color=gray!25]
        (16,-0.75) rectangle (22,11);
  \fill [pattern=north east lines, pattern color=gray!25]
        (22,-0.75) rectangle (28,11);
  \draw [dashed, thin, gray] (16,-0.75) -- (16,11);
  \draw [dashed, thin, gray] (22,-0.75) -- (22,11);
  \draw [dashed, thin, gray] (28,-0.75) -- (28,11);
  \draw [gray] (16,-1) -- (16,-1.25) -- (22,-1.25) -- (22,-1);
  \draw [gray] (22,-1) -- (22,-1.25) -- (28,-1.25) -- (28,-1);
  \draw [gray] (22,-1.5) node [below] {\footnotesize two copies of $L_1$};

  \fill [pattern=north east lines, pattern color=red!25]
        (17.5,-0.75) rectangle (19.5,11);
  \draw [dashed, thin, red!25] (17.5,-0.75) -- (17.5,11);
  \draw [dashed, thin, red!25] (19.5,-0.75) -- (19.5,11);
  \draw [red!25] (17.5,11.25) -- (17.5,11.5) -- (19.5,11.5) -- (19.5,11.25);
  \draw [red!50] (18.5,11.75) node [above] {\footnotesize $L_0$};
 
  \draw (0,0) node (node0) {};
  \draw (4,0) node [dot] (node1) {};
  \draw (10,0) node [dot] (node2) {};
  \draw (16,0) node [dot] (node3) {};
  \draw (22,0) node [dot] (node4) {};
  \draw (28,0) node [dot] (node5) {};
  \draw (30,0) node (node6) {};
  \draw (30,1) node (node7) {};

  \draw (28,1) node [dot] (node8) {};
  \draw (25,1) node (node9) {};
  \draw (25,2) node (node10) {};
  \draw (28,2) node [dot] (node11) {};

  \draw (31,2) node (node12) {};
  \draw (31,3) node (node13) {};
  \draw (28,3) node [fulldot] (node14) {};
  \draw (22,3) node [dot] (node15) {};
  \draw (19,3) node (node16) {};
  \draw (19,4) node (node17) {};
  \draw (22,4) node [dot] (node18) {};
  \draw (26,4) node (node19) {};
  \draw (26,5) node (node20) {};
  \draw (22,5) node [fulldot] (node21) {};
  \draw (16,5) node [dot] (node22) {};

  \draw (10,5) node [dot] (node23) {};
  \draw (4,5) node [dot] (node24) {};
  \draw (2,5) node (node25) {};
  \draw (2,6) node (node26) {};
  \draw (4,6) node [dot] (node27) {};

  \draw (7,6) node (node28) {};
  \draw (7,7) node (node29) {};
  \draw (4,7) node [dot] (node30) {};
  \draw (1,7) node (node31) {};
  \draw (1,8) node (node32) {};
  \draw (4,8) node [fulldot] (node33) {};
  \draw (10,8) node [dot] (node34) {};

  \draw (16,8) node [dot] (node35) {};
  \draw (20,8) node (node36) {};
  \draw (20,9) node (node37) {};
  \draw (16,9) node [dot] (node38) {};

  \draw (10,9) node [dot] (node39) {};
  \draw (6,9) node (node40) {};
  \draw (6,10) node (node41) {};
  \draw (10,10) node [dot] (node42) {};

  \draw (16,10) node [dot] (node43) {};
  \draw (22,10) node [dot] (node44) {};
  \draw (28,10) node [dot] (node45) {};
  \draw (32,10) node (node46) {};

  \draw [factor] (node0) -- (node1);
  \draw [factor] (node1) -- (node2);
  \draw [factor] (node2) -- (node3);
  \draw [factor] (node3) -- (node4);
  \draw [factor] (node4) -- (node5);
  \draw [factor] (node5) -- (node6.center) -- (node7.center) -- (node8);
  
  \draw [fullfactor] (node8) -- (node9.center) -- (node10.center) -- (node11);
  \draw [factor] (node11) -- (node12.center) -- (node13.center) -- (node14);
  \draw [fullfactor] (node14) -- (node15);
  \draw [fullfactor] (node15) -- (node16.center) -- (node17.center) -- (node18);
  \draw [fullfactor, red] (node18) -- (node19.center) -- (node20.center) -- (node21);
  \draw [fullfactor] (node21) -- (node22);

  \draw [factor] (node22) -- (node23);
  \draw [factor] (node23) -- (node24);
  \draw [factor] (node24) -- (node25.center) -- (node26.center) -- (node27);

  \draw [fullfactor] (node27) -- (node28.center) -- (node29.center) -- (node30);
  \draw [factor] (node30) -- (node31.center) -- (node32.center) -- (node33);
  \draw [fullfactor] (node33) -- (node34);

  \draw [factor] (node34) -- (node35);
  \draw [fullfactor, red] (node35) -- (node36.center) -- (node37.center) -- (node38);

  \draw [factor] (node38) -- (node39);
  \draw [fullfactor, red] (node39) -- (node40.center) -- (node41.center) -- (node42);

  \draw [factor] (node42) -- (node43);
  \draw [factor] (node43) -- (node44);
  \draw [factor] (node44) -- (node45);
  \draw [factor] (node45) -- (node46);
\end{scope}
\end{tikzpicture}
\caption{An inversion $(L_1,C_1,L_2,C_2)$ whose pairs $(L_i,C_i)$ are
  not output-minimal. The red parts produce long outputs, and lie
  outside $\rho[\an{C_1},\an{C_2}]$.}
\label{fig:non-output-minimal-inversion-full}
\end{figure}
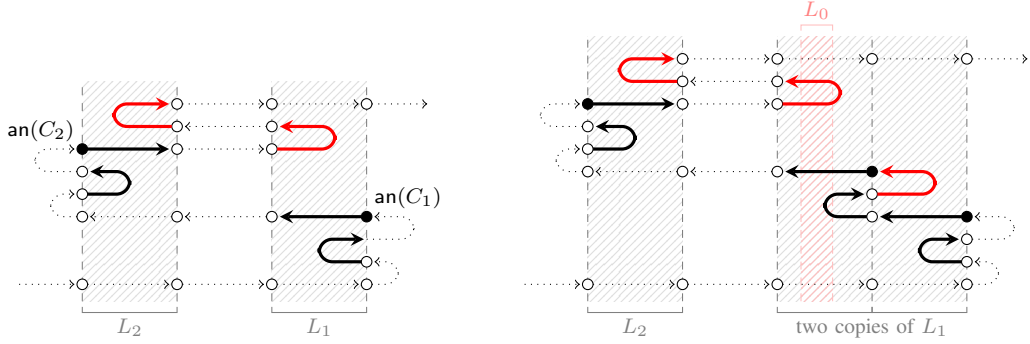

Recall that the period $p$ for
the word $\outb{\tr{C_1}} ~ \outb{\rho[\an{C_1},\an{C_2}]} ~ \outb{\tr{C_2}}$ was obtained
by considering a run $\rho^{(m_1,m_2)}$ where the loops $L_1$ and $L_2$ have been pumped
$m_1$ and $m_2$ times, respectively. To bound the period, we need to consider 
inversions that are formed by output-minimal pairs. 
As already explained, we cannot assume that the inversion $(L_1,C_1,L_2,C_2)$ 
contains an output-minimal pair.
For example, the left part of Fig.~\ref{fig:non-output-minimal-inversion-full}
represents a situation where both pairs $(L_1,C_1)$ and $(L_2,C_2)$ of the inversion 
are not output-minimal. 
Nonetheless, in the pumped run $\rho^{(2,1)}$ we do find inversions with output-minimal pairs.
For example, as suggested by the right part of Fig.~\ref{fig:non-output-minimal-inversion-full},
we can consider the leftmost and rightmost occurrences of $L_1$ 
in $\rho^{(2,1)}$, denoted as $\olft L_1$ and $\ort L_1$,  respectively.
Let $(L_0,C_0)$ be any {\sl output-minimal} pair such that $L_0$ is an 
idempotent loop, $\out{\tr{C_0}}\neq\emptystr$, and either $(L_0,C_0)=(\olft L_1,C_1)$ 
or $(L_0,C_0) \lesspair (\olft L_1,C_1)$
--- such a loop $L_0$ is suggestively represented in the figure 
by the red vertical stripe.

We claim that either $(L_0,C_0,L_2,C_2)$ or $(\ort L_1,C_1,L_0,C_0)$
is an inversion of the run $\rho^{(2,1)}$, depending on whether the 
anchor of $C_0$ inside $L_0$ occurs before or after the anchor of $C_2$ inside $L_2$.
First, note that all the loops $L_0$, $L_2$, $\ort L_1$ are idempotent
and non-overlapping; more precisely, we have 
$\max(L_2) \le \min(L_0)$ and $\max(L_0) \le \min(\ort L_1)$. 
Moreover, the trace outputs for the pairs $(L_0,C_0)$, $(L_2,C_2)$, 
$(\ort L_1,C_1)$ are non-empty.
So it remains to distinguish the two cases based on the ordering of the
anchors of $C_0$, $C_1$, $C_2$ inside the loops $L_0$, $\ort L_1$, $L_2$, respectively.
We denote those anchors by $\ell_0$, $\ell_1$, $\ell_2$.
If $\ell_0 \leqtime \ell_2$, then $(L_0,C_0,L_2,C_2)$ is clearly an inversion.
Otherwise, because $(\ort L_1,C_1,L_2,C_2)$ is an inversion, we know that 
$\ell_1 \leqtime \ell_2 \leqtime \ell_0$, and hence $(\ort L_1,C_1,L_0,C_0)$
is an inversion.

Now, we know that
$\rho^{(2,1)}$ contains the inversion $(\ort L_1,C_1,L_2,C_2)$, but also an
inversion with an output-minimal pair $(L_0,C_0)$, where $L_0$ is strictly 
between $\ort L_1$ and $L_2$.
For all $m_0,m_1,m_2$, we define $\rho^{(m_0,m_1,m_2)}$ as the run obtained from 
$\rho^{(2,1)}$ by pumping $m_0,m_1,m_2$ times the loops $L_0,\ort L_1,L_2$, respectively.
Since the output of the run $\rho^{(m_0,m_1,m_2)}$ contains many repetitions of 
the trace output $\out{\tr{C_0}}$ of $C_0$ inside $L_0$, and 
since these repetitions occur as factors of the output produced 
inside the inversion $(\ort L_1,C_1,L_2,C_2)$, their period $p'$ divides 
$|\out{\tr{C_0}}|$, $|\out{\tr{C_1}}|$, and $|\out{\tr{C_2}}|$ (due to Lemma~\ref{lemma:fine-wilf}). 
By Theorem~\ref{thm:saarela},  
we deduce that the word
$\outb{\tr{C_1} ~ \rho[\an{C_1},\an{C_2}] ~ \tr{C_2}}$
has period $p'$ as well.
To conclude the proof, it suffices to recall Lemma~\ref{lem:output-minimal},
saying that the length of $\out{\tr{C_0}}$, and hence the period $p'$, 
is bounded by $\bound$.
\end{proof}

\begin{replemma}{lemma:one-way-vs-two-way}
Consider a word equation of the form
\[
  v_0^{(m_1,m_2)} \: \pmb{v_1^{m_1}} \: v_2^{(m_1,m_2)} \: \pmb{v_3^{m_2}} \: v_4^{(m_1,m_2)}
  ~=~
  w_0 \: \pmb{w_1^{m_2}} \: w_2 \: \pmb{w_3^{m_1}} \: w_4
\]
where $m_1,m_2$ are the unknowns, $v_1,v_3$ are non-empty words,
and $v_0^{(m_1,m_2)},v_2^{(m_1,m_2)},v_4^{(m_1,m_2)}$ are words
that may contain some factors of the form $v^{m_1}$ or $v^{m_2}$,
for some $v$.
If the above equation holds for all $m_1,m_2\in\bbN$, 
then the words
$\pmb{v_1} ~ \pmb{v_1^{m_1}} ~ v_2^{(m_1,m_2)} ~ \pmb{v_3^{m_2}} ~ \pmb{v_3}$
are periodic with period $\gcd(|v_1|,|v_3|)$, for all $m_1,m_2\in\bbN$.
\end{replemma}

\begin{proof}
The idea of the proof is to let the parameters $m_1,m_2$ of the equation
grow independently, and exploit Fine and Wilf's theorem (Lemma \ref{lemma:fine-wilf}) 
a certain number of times to establish periodicities in overlapping factors of the 
considered words.

We begin by fixing $m_1$ large enough so that the factor
$\pmb{v_1^{m_1}}$ of the left hand-side of the equation is 
longer than $|w_0|+|w_1|$ (this is possible because $v_1$ is non-empty).
Now, if we let $m_2$ grow arbitrarily large, we see that the length of 
the periodic word $\pmb{w_1^{m_2}}$ is almost equal to the length of 
the left hand-side term
$v_0^{(m_1,m_2)} ~ \pmb{v_1^{m_1}} ~ v_2^{(m_1,m_2)} ~ \pmb{v_3^{m_2}} ~ v_4^{(m_1,m_2)}$:
indeed, the difference in length is given by the constant 
$|w_0| + |w_2| + m_1\cdot |w_3| + |w_4|$.
In particular, this implies that $\pmb{w_1^{m_2}}$ 
covers arbitrarily long prefixes of 
$\pmb{v_1} ~ v_2^{(m_1,m_2)} ~ \pmb{v_3^{m_2+1}}$,
which in its turn contains long repetitions of the word $v_3$.
Hence, by Lemma \ref{lemma:fine-wilf},
the word 
$\pmb{v_1} ~ v_2^{(m_1,m_2)} ~ \pmb{v_3^{m_2+1}}$ 
has period $|v_3|$.

We remark that the periodicity shown so far holds for infinitely many
$m_1$ and for all but finitely many $m_2$, where the threshold for
$m_2$ depends on $m_1$: once $m_1$ is fixed, $m_2$ needs to be larger
than $f(m_1)$, for a suitable function $f$.  
In fact, using Theorem~\ref{thm:saarela}, we can show that the periodicity 
holds even when $m_2$ ranges over all natural numbers. 
To see this, we introduce the following shorthand: 
given a word $w$ and a rational number $r=\frac{n}{|w|}$, with $n\in\bbN$, 
we denote by $w^r$ the word $w^{\lfloor r \rfloor}\:w'$, where
$w'$ is the prefix of $w$ of length $|w|\cdot (r - \lfloor r \rfloor)$.
We then state the periodicity property for 
$\pmb{v_1} ~ v_2^{(m_1,m_2)} ~ \pmb{v_3^{m_2+1}}$ 
as an equation of the form
\[
  \pmb{v_1} ~ v_2^{(m_1,m_2)} ~ \pmb{v_3^{m_2+1}} ~=~ \pmb{v^{\frac{g(m_2)}{|v|}}}
\] 
which, once $m_1$ is fixed, must hold for all but finitely many $m_2$, 
for a suitable word $v$ of the same length as $v_3$, and for a 
suitable linear function $g:\bbN\rightarrow\bbN$. 
More precisely, $g(m_2)$ gives the length of the 
left hand-side of the equation.
The above equation can be easily rewritten so as to highlight 
all the repetitions that depend on $m_2$, including those that
are hidden inside the term $v_2^{(m_1,m_2)}$. 
Note that we cannot apply Theorem~\ref{thm:saarela} yet, since
the repetitions in the right hand-side of the equation may be fractional.
If this is the case, however, it means that the left hand-side of the 
equation contains a repetition of the form $w^{m_2}$, for some word
$w$ whose length is not multiple of $|v|$. By Fine and Wilf's 
theorem (Lemma \ref{lemma:fine-wilf}), we know that the period 
of the left hand-side is in fact smaller, i.e.~$\gcd(|v_3|,|w|)$.
We can then replace the right hand-side of the equation 
with an exact repetition of a word $v'$ shorter than $v$.
This enables the application of Theorem~\ref{thm:saarela},
which implies that the equation holds for all $m_2\in\bbN$.
In this way we have shown that the word 
$\pmb{v_1} ~ v_2^{(m_1,m_2)} ~ \pmb{v_3^{m_2+1}}$ 
has period $|v_3|$ for all $m_2\in\bbN$.

We could also apply a symmetric reasoning, by fixing $m_2$ and 
by letting $m_1$ grow arbitrarily large. Doing so, we prove 
that for a large enough $m_2$ and for all but finitely many $m_1$, 
the word $\pmb{v_1^{m_1+1}} ~ v_2^{(m_1,m_2)} ~ \pmb{v_3}$ is periodic
with period $|v_1|$. As before, this can be strengthened to hold for all 
$m_1\in\bbN$, independently of the choice of $m_2$.

Putting together the results proven so far, we get that for all but finitely many $m_1,m_2$,
\[
  \rightward{ \underbracket[0.5pt]{ \phantom{ \pmb{v_1^{m_1}} \cdot \pmb{v_1} ~\cdot~
                                                   v_2^{(m_1,m_2)} ~\cdot~ \pmb{v_3} } }%
                                 _{\text{period } |v_1|} }
  \pmb{v_1^{m_1}} \cdot
  \overbracket[0.5pt]{ \pmb{v_1} ~\cdot~ v_2^{(m_1,m_2)} ~\cdot~ 
                       \pmb{v_3} \cdot \pmb{v_3^{m_2}} }%
                    ^{\text{period } |v_3|}
  .
\]
Finally, we observe that the prefix 
$\pmb{v_1^{m_1+1}}\cdot v_2^{(m_1,m_2)}\cdot \pmb{v_3}$ 
and the suffix
$\pmb{v_1}\cdot v_2^{(m_1,m_2)}\cdot\pmb{v_3^{m_2+1}}$ 
share a common factor of length at least $|v_1|+|v_3|$. 
By Lemma~\ref{lemma:fine-wilf}, 
we derive that $\pmb{v_1^{m_1+1}}\cdot v_2^{(m_1,m_2)}\cdot
\pmb{v_3^{m_2+1}}$ has period $\gcd(|v_1|,|v_3|)$.
Finally, by exploiting again Theorem \ref{thm:saarela}, 
we generalize this periodicity property to all $m_1,m_2\in\bbN$.
\end{proof}

\begin{replemma}{lem:output-minimal}
For every output-minimal pair $(L,C)$, $|\out{\tr{C}}| \le \bound$.
\end{replemma}

\begin{proof}
Consider a pair $(L,C)$ consisting of an idempotent loop $L=[x_1,x_2]$ and a component $C$ of $L$. 
We suppose that the length of $\out{\tr{C}}$ exceeds $\bound$ and we claim that $(L,C)$ is not output-minimal.

Recall that $\tr{C}$ is a concatenation of $(L,C)$-factors, say,
$\tr{C}=\beta_1\cdots\beta_k$. Let $\ell_1$ (resp.~$\ell_2$) be the first 
(resp.~last) location that is visited by these factors. Further let
$K = [\ell_1,\ell_2]$ and $Z = K \:\cap\: (L\times\bbN)$. 
By construction, the subrun $\rho\mid Z$ can be seen as a concatenation
of the factors $\beta_1,\dots,\beta_k$, possibly in a different order than
that of $\tr{C}$. This implies that $|\out{\rho\mid Z}| > \bound$.

By Theorem \ref{thm:simon2}, we know that there exist
an idempotent loop $L'\subsetneq L$ and a component $C'$ of $L'$
such that $\an{C'} \in K$ and $\out{\tr{C'}}\neq\emptystr$.
In particular, the $(L',C')$-factor that starts at the location $\an{C'}$ 
is entirely contained in some $(L,C)$-factor. This implies that 
$(L',C') \lesspair (L,C)$, and thus $(L,C)$ is not output-minimal.
\end{proof}

\begin{repproposition}{prop:block-periodicity}
If $\rho$ satisfies the periodicity property stated in \PR2
and $\ell \leqtime \ell'$ are two locations in the 
same $\simeq$-class, then $\outb{\rho[\ell,\ell']}$ 
has period at most $\bound$.
\end{repproposition}

\begin{proof}
The claim for $\ell=\ell'$ holds trivially, so we assume that $\ell\lesstime\ell'$. 
We know that $\ell,\ell'$ belong to the same non-singleton $\simeq$-class.
By definition of $\crossrel$, the run $\rho$ contains some inversions
$(L_0,C_0,L_1,C_1)$, $(L_2,C_2,L_3,C_3)$, \dots, $(L_{2k},C_{2k},L_{2k+1},C_{2k+1})$ 
such that $\an{C_0} \leqtime \ell \lesstime \ell' \leqtime \an{C_{2k+1}}$ and
$\an{C_{2i}} \leqtime \an{C_{2i+2}} \leqtime \an{C_{2i+1}} \leqtime \an{C_{2i+3}}$
for all $i=0,\dots,k-1$. Without loss of generality we can assume that every 
inversion $(L_{2i},C_{2i},L_{2i+1},C_{2i+1})$ is \emph{maximal} in the following sense:
there is no other inversion $(L,C,L',C') \neq (L_{2i},C_{2i},L_{2i+1},C_{2i+1})$ such 
that $\an{C}\leqtime\an{C_{2i}}\leqtime\an{C_{2i+1}}\leqtime\an{C'}$.

We introduce the following shorthands for all $i=0,\dots,2k+1$:
$\ell_i=\an{C_i}$, $v_i=\out{\tr{C_i}}$, and $p_i=|v_i|$. By Property P2, we know that $v_{2i} ~ \out{\rho[\ell_{2i},\ell_{2i+1}]} ~ v_{2i+1}$ 
has period at most $\bound$ that divides both $p_{2i}$ and $p_{2i+1}$.

In order to show that $\outb{\rho[\ell,\ell']}$ has period at most $B$, 
it suffices to prove the following claim by induction on $i$:

\begin{claim}
For all $i=0,\dots,k$, the period of
$\outb{\rho[\ell_0,\ell_{2i+1}]} \: v_{2i+1}$ divides $p_{2i+1}$
and is bounded by $\bound$.
\end{claim}

\begin{proof}[Proof of claim]
The base case $i=0$ follows from Property P2, 
since $(L_0,C_0,L_1,C_1)$ is an inversion.
For the inductive step, we assume that the claim holds for $i<k$ 
and we prove it for $i+1$. We factorize our word as follows:
\[
  \outb{\rho[\ell_0,\ell_{2i+3}]} ~ v_{2i+3} 
  ~=~
  \rightward{ \underbracket[0.5pt]{ \phantom{ 
                                              \outb{\rho[\ell_0,\ell_{2i+2}]} ~ 
                                              \outb{\rho[\ell_{2i+2},\ell_{2i+1}]} ~ } }%
                                 _{\text{period } p_{2i+1}} }
  \outb{\rho[\ell_0,\ell_{2i+2}]}
  \overbracket[0.5pt]{ ~ \outb{\rho[\ell_{2i+2},\ell_{2i+1}]} ~
                       \outb{\rho[\ell_{2i+1},\ell_{2i+3}]} ~
                       v_{2i+3} }%
                    ^{\text{periods } p_{2i+2} \text{ and } p_{2i+3}}
  .
\]
By the inductive hypothesis, the output produced between $\ell_0$ 
and $\ell_{2i+1}$, even extended to the right with the trace output $v_{2i+1}$, 
has period that divides $p_{2i+1}$. 
Moreover, because $(L_{2i+2},C_{2i+2},L_{2i+3},C_{2i+3})$ is an inversion, 
the output produced between the locations 
$\ell_{2i+2}=\an{C_{2i+2}}$ and $\ell_{2i+3}=\an{C_{2i+3}}$, 
extended to the left with $v_{2i+2}$ and to the 
right with $v_{2i+3}$, has period that divides both $p_{2i+2}$ and $p_{2i+3}$.
This does not suffice yet to apply Fine-Wilf's theorem 
so as to derive a suitable period of
$\outb{\rho[\ell_0,\ell_{2i+3}]} ~ v_{2i+3}$, 
since the common factor $\outb{\rho[\ell_{2i+2},\ell_{2i+1}]}$ 
might be too short.
The key argument here is that the interval $[\ell_{2i+2},\ell_{2i+1}]$ 
is covered by the inversion $(L_{2i+2},C_{2i+2},L_{2i+1},C_{2i+1})$,
which is different from the previous ones.

For this, we have to prove that the anchors $\an{C_{2i+2}}$ and $\an{C_{2i+1}}$ 
are correctly ordered w.r.t.~$\leqtime$ and the ordering of positions
(recall Definition~\ref{def:inversion}).
First, we have
$\an{C_{2i+2}} \leqtime \an{C_{2i+1}}$ by assumption.
Now we prove that $\an{C_{2i+1}}$ is strictly to the left of $\an{C_{2i+2}}$,
according to the ordering of positions. By way of contradiction,
suppose that this is not the case, namely,
$\an{C_{2i+1}}=(x_{2i+1},y_{2i+1})$, $\an{C_{2i+2}}=(x_{2i+2},y_{2i+2})$, 
and $x_{2i+1} > x_{2i+2}$.
Because $(L_{2i},C_{2i},L_{2i+1},C_{2i+1})$ and
$(L_{2i+2},C_{2i+2},L_{2i+3},C_{2i+3})$ are inversions, 
we know that $\an{C_{2i+3}}$ is to the left of $\an{C_{2i+2}}$ 
and $\an{C_{2i+1}}$ is to the left of $\an{C_{2i}}$. 
This implies that $\an{C_{2i+3}}$ is to the left of $\an{C_{2i}}$,
and hence $(L_{2i},C_{2i},L_{2i+3},C_{2i+3})$ is also an inversion.
But this would contradict the maximality of $(L_{2i},C_{2i},L_{2i+1},C_{2i+1})$,
which was assumed at the beginning of the proof.
\end{proof}

Now that we know that $\an{C_{2i+2}}$ and $\an{C_{2i+1}}$ 
are correctly ordered w.r.t.~$\leqtime$ and the ordering of positions,
we recall that the trace outputs $v_{2i+1}$ and $v_{2i+2}$ are non-empty.
This implies that $(L_{2i+2},C_{2i+2},L_{2i+1},C_{2i+1})$ is an inversion.
Moreover, the latter inversion covers the interval of locations 
$[\ell_{2i+2},\ell_{2i+1}]$. 
By Property P2, the word
$v_{2i+2} ~ \out{\rho[\ell_{2i+2},\ell_{2i+1}]} ~ v_{2i+1}$
has period that divides both $p_{2i+2}$ and $p_{2i+1}$.

Summing up, we have:
\begin{enumerate}
  \item $w_1 ~=~ \outb{\rho[\ell_0,\ell_{2i+1}]} ~ v_{2i+1}$ has period $p_{2i+1}$,
  \item $w_2 ~=~ v_{2i+2} ~ \outb{\rho[\ell_{2i+2},\ell_{2i+1}]} ~ v_{2i+1}$ 
        has period $p = \gcd(p_{2i+2},p_{2i+1})$, 
  \item $w_3 ~=~ v_{2i+2} ~ \outb{\rho[\ell_{2i+2},\ell_{2i+3}]} ~ v_{2i+3}$
        has period $p' = \gcd(p_{2i+2},p_{2i+3})$.
\end{enumerate}
We are now ready to exploit our slightly stronger variant of Fine-Wilf's theorem, 
that is, Lemma~\ref{lemma:fine-wilf}. 

Let $w = \outb{\rho[\ell_{2i+2},\ell_{2i+1}]}~ v_{2i+1}$ be the common suffix of
$w_1$ and $w_2$. From 1.~and 2., we know that the latter words have period 
$p_{2i+1}$ and $p=\gcd(p_{2i+2},p_{2i+1})$, respectively.
Moreover, since $p$ divides $|w_2|-|w|$ ($=|v_{2i+2}|$), 
$w$ is also a prefix of $w_2$. 
For the same reason, we also know that 
$|w| \ge |v_{2i+1}| = p_{2i+1} = p_{2i+1} + p - \gcd(p_{2i+1},p)$ (the latter equality
follows from the fact that $p$ divides $p_{2i+1}$).
Thus, by applying Lemma~\ref{lemma:fine-wilf} to $w_1=w'_1\,w$ and $w_2=w\,w''_2$, 
using $w$ as common factor, we obtain that
\begin{enumerate}
  \setcounter{enumi}{3}
  \item $w_4 ~=~ w'_1 \: w \: w''_2 
             ~=~ \outb{\rho[\ell_0,\ell_{2i+2}]} ~ v_{2i+2} ~ \outb{\rho[\ell_{2i+2},\ell_{2i+1}]} ~ v_{2i+1}$ 
        has period $p$.
\end{enumerate}
Now, from 2.~and 3., we know that the words $w_2$ and $w_3$
have periods $p$ and $p'$, respectively, and contain $v_{2i+2}$ as factor.
Moreover, the length of the factor $v_{2i+2}$ is a multiple of both periods $p$ and $p'$,
and hence $|v_{2i+2}| \ge p+p'-\gcd(p,p')$ (this is folklore, and follows from basic facts 
in number theory, such as $q\cdot q' \ge q + q' - 1$ for all $q,q'\in\bbN$). 
From Lemma~\ref{lemma:fine-wilf} we derive that 
\begin{enumerate}
  \setcounter{enumi}{4}
  \item $w_5 ~=~ v_{2i+2} ~ \outb{\rho[\ell_{2i+2},\ell_{2i+3}]} ~ v_{2i+3}$
        has period $p'' = \gcd(p_{2i+1},p_{2i+2},p_{2i+3})$.
\end{enumerate}
In a similar way, from 4.~and 5., using again $v_{2i+2}$ as common factor of $w_4$ and $w_5$, 
we derive
\begin{enumerate}
  \setcounter{enumi}{5}
  \item $w_6 ~=~ \outb{\rho[\ell_0,\ell_{2i+2}]} ~ v_{2i+2} ~ 
                 \outb{\rho[\ell_{2i+2},\ell_{2i+3}]} ~ v_{2i+3}$ 
        has period $p''$.
\end{enumerate}
Finally, the periodicity is not affected when we remove
factors of length multiple than the period. In particular, 
by removing the factor $v_{2i+2}$ from $w_6$, 
we obtain the word
$\outb{\rho[\ell_0,\ell_{2i+3}]} ~ v_{2i+3}$, whose period
still divides $p_{2i+3}$. This proves the claim for the inductive step,
and completes the proof of the proposition.
\end{proof}

\begin{replemma}{lemma:bounding-box}
If $K=[\ell,\ell']$ is a non-singleton $\simeq$-class, then $\rho[\ell_1,\ell_2]$ is a block,
where $[\ell_1,\ell_2]=\block{K}$.
\end{replemma}

\begin{proof}
Let $K=[\ell,\ell']$ and $\block{K}=[\ell_1,\ell_2]$, 
with $\ell_i=(x_i,y_i)$ for both $i=1,2$,
and let $\an{K}$ and $X_{\an{K}}$ be the sets given in Definition \ref{def:bounding-box}. 

We begin by observing that the factor $\rho[\ell_1,\ell]$ 
between the first location of the block
and the first location of the equivalence class lies entirely 
to the right of position $x_1$.
Indeed, if this were not the case, there would exist another location 
$\ell'_1=(x_1,y_1+1)$, on the same position $x_1$ as $\ell_1$ but at a higher level,
such that $\ell_1 \lesstime \ell'_1 \leqtime \ell$. But this would contradict 
Definition \ref{def:bounding-box}.
In a similar way one verifies that the factor $\rho[\ell',\ell_2]$ 
lies to the left of $x_2$.

Next, we prove that the output produced by the factor 
$\rho[\ell_1,\ell_2]$ is quasi-periodic. 
By Definition \ref{def:bounding-box}, we have 
$\ell_1\leqtime\ell\lesstime\ell'\leqtime\ell_2$,
and by Proposition \ref{prop:block-periodicity} 
we know that $\out{\rho[\ell,\ell']}$ is periodic with 
period at most $\bound$. So it suffices to bound the length 
of the words $\out{\rho[\ell_1,\ell]}$ and $\out{\rho[\ell',\ell_2]}$. 
We shall focus on the former word, as the arguments for the latter 
are similar. 
As usual, the idea is to apply a Ramsey-type argument.

Suppose, by way of contradiction, 
that the length of $|\out{\rho[\ell_1,\ell]}| > \bound$.
We head towards finding a location $\ell'' \lesstime \ell$ 
that is $\simeq$-equivalent to $\ell$,
thus contradicting the fact that $\ell$ is the first location of the equivalence class $K$.
Recall that the factor $\rho[\ell_1,\ell]$ lies entirely to the right of the
position $x_1$ of $\ell_1$, so $|\out{\rho[\ell_1,\ell]}| > \bound$ is 
equivalent to saying 
$|\out{\rho\mid Z}| > \bound$, where $Z = [\ell_1,\ell] \:\cap\: \big([x_1,\omega]\times\bbN\big)$.
Theorem \ref{thm:simon2} implies the existence of an idempotent loop $L$ and a
component $C$ such that
\begin{itemize}
  \item $\min(L) > x_1$,
  \item $\ell_1\lesstime \an{C} \lesstime \ell$,
  \item $\out{\tr{C}}\neq\emptystr$.
\end{itemize}
Let $\ell''=\an{C}$. By construction, $x_1$ is the leftmost position of all the 
locations of the class $K=[\ell,\ell']$ that are also anchors of components of inversions.
Thus there exist an inversion $(L_1,C_1,L_2,C_2)$ 
and a location $\ell'''=(x_1,y''')\in K$ such that 
$\ell'''=\an{C_i}$ for some $i\in\{1,2\}$. 
Since $\ell'' \lesstime \ell \leqtime \ell'''$
and the position of $\ell''$ is to the right of $x_1$, 
we know that $(L,C,L_i,C_i)$ is also an inversion, 
and hence $\ell'' \simeq \ell''' \simeq \ell$.
But since $\ell''\neq \ell$, we get a contradiction with the 
assumption that $\ell$ is the first location of a $\simeq$-class. 
In this way we have shown that $|\out{\rho[\ell_1,\ell]}| \le \bound$.

It remains to bound the lengths of the outputs produced 
by the subruns $\rho\mid Z^\leftshort$ and $\rho\mid Z^\rightshort$, 
where $Z^\leftshort=[\ell_1,\ell_2] \:\cap\: \big([0,x_1]\times\bbN\big)$ 
and $Z^\rightshort=[\ell_1,\ell_2] \:\cap\: \big([x_2,\omega]\times\bbN\big)$.
As usual, we consider only one of the two symmetric cases.
Suppose, by way of contradiction, that $|\out{\rho\mid Z^\leftshort}| > \bound$.
By Theorem \ref{thm:simon2}, there exist an idempotent loop $L$ 
and a component $C$ of $L$ such that
\begin{itemize}
  \item $\max(L) < x_1$,
  \item $\ell_1 \lesstime \an{C} \lesstime \ell_2$,
  \item $\out{\tr{C}}\neq\emptystr$.
\end{itemize}
Let $\ell''=\an{C}$.
By following the same line of reasoning as before, we recall that
$\ell$ is the first location of the non-singleton class $K$.
From this we derive the existence an inversion $(L_1,C_1,L_2,C_2)$ 
such that $\ell=\an{C_1}$.
We claim that $\ell \leqtime \ell''$.
Indeed, if this were not the case, then, because $\ell''$ is strictly to the 
left of $x_1$ and $\ell$ is to the right of $x_1$, there would exist a location 
$\ell'_1$ between $\ell''$ and $\ell$ that lies at position $x_1$. 
But $\ell_1 \lesstime \ell'' \leqtime \ell'_1 \leqtime \ell$ would 
contradict the fact that $\ell_1$ is the {\sl latest} location before $\ell$ 
that lies at the position $x_1 = \min(X_{\an{K}})$.
Now that we know that $\ell\leqtime\ell''$ and that $\ell''$ is to the left of $x_1$, 
we observe that $(L_1,C_1,L,C)$ is also an inversion, and hence $\ell''\in \an{K}$. 
Since $\ell''$ is strictly to the left of $x_1$,
we get a contradiction with the definition of $x_1$ as leftmost 
position of the locations of $K$ that are anchors of components of inversions.
We must conclude that $|\out{\rho\mid Z^\leftshort}| \le \bound$.

This completes the proof that $\rho\mid\block{K}$ is a block.
\end{proof}

\begin{replemma}{lem:consecutive-blocks}
Suppose that $K_1$ and $K_2$ are two different non-singleton $\simeq$-classes
such that $\ell \lesstime \ell'$ for all $\ell \in K_1$ and $\ell' \in
K_2$.
Let $\block{K_1}=[\ell_1,\ell_2]$ and $\block{K_2}=[\ell_3,\ell_4]$,
with $\ell_2=(x_2,y_2)$ and $\ell_3=(x_3,y_3)$. 
Then $x_2 < x_3$.
\end{replemma}

\begin{proof}
Suppose by contradiction  that  $K_1$ and $K_2$  are as in the
statement, but $x_2 \ge x_3$. 
By Definition \ref{def:bounding-box}, 
$x_2=\max(X_{\an{K_1}})$ and $x_3=\min(X_{\an{K_2}})$.
This implies the existence of some inversions
$(L_1,C_1,L_2,C_2)$ and $(L_3,C_3,L_4,C_4)$ such that
$\an{C_i}=(x_2,y)$ for some $i\in\{1,2\}$
and $\an{C_j}=(x_3,y')$ for some $j\in\{3,4\}$.
Moreover, since $\an{C_i} \leqtime \an{C_j}$ and $x_2 \ge x_3$, 
we know that $(L_i,C_i,L_j,C_j)$ is also an inversion.
But this means that $K_1=K_2$.
\end{proof}

\begin{replemma}{lem:diagonal}
Let $\rho[\ell_1,\ell_2]$ be a
factor of $\rho$ that does not overlap any $\simeq$-block,
with $\ell_1=(x_1,y_1)$, $\ell_2=(x_2,y_2)$, and $x_1<x_2$.
Then $\rho[\ell_1,\ell_2]$ is a diagonal.
\end{replemma}

\begin{proof}
Suppose by contradiction that there is some $x\in[x_1,x_2]$
such that, for all locations $\ell=(x,y)$ between $\ell_1$ and $\ell_2$, 
one of the following conditions holds:
\begin{enumerate}
  \item $|\out{\rho\mid Z_\ell^\upperleft}| > \bound$, 
        where $Z_\ell^\upperleft = [\ell,\ell_2] \:\cap\: \big([0,x]\times\bbN\big)$,
  \item $|\out{\rho\mid Z_\ell^\lowerright}| > \bound$, 
        where $Z_\ell^\lowerright = [\ell_1,\ell] \:\cap\: \big([x,\omega]\times\bbN\big)$.
\end{enumerate}
We claim first that for each condition above
there is some level $y$ at which it holds. Observe that for the
highest location $\ell$ of the run at position $x$, the set
$Z_\ell^\upperleft$ is empty, since the outgoing transition at $\ell$
is rightward. So condition 1 is trivially violated at $\ell$ as above, hence
condition 2 holds by the initial assumption. Symmetrically, condition
1 holds at the lowest location  of the run at position $x$.
Let us now compare, for each condition, the levels where it holds.

Clearly, the lower the level of the location $\ell$, 
the easier it is to satisfy condition 1, and symmetrically for condition 2.
So, let $\ell=(x,y)$ (resp.~$\ell'=(x,y')$) be the highest (resp.~lowest) 
location at position $x$ that satisfies condition 1 (resp.~condition 2). 

We claim that $y \ge y'$.
For this, we first observe that $y \ge y'-1$, since otherwise there would 
exist a location $\ell=(x,y'')$, with $y < y'' < y'$, violating both conditions 1 and 2.
Moreover, $y$ must be odd, otherwise the transition departing from $\ell=(x,y)$ 
would be rightward oriented and the location $\ell''=(x,y+1)$ would still 
satisfy condition 1, contradicting the fact that $\ell=(x,y)$ was chosen to 
be the highest location. 
For similar reasons, $y'$ must also be odd, otherwise there would be 
a location $\ell''=(x,y'-1)$ that precedes $\ell'$ and satisfies condition 2.
But since $y \ge y'-1$ and both $y$ and $y'$ are odd, we need to have $y\ge y'$.

From the previous arguments we know that in fact $\ell=(x,y)$ satisfies both conditions
1 and 2. We can thus apply Theorem \ref{thm:simon2} to the sets 
$Z_\ell^\lowerright$ and $Z_\ell^\upperleft$, deriving the existence of
two idempotent loops $L_1,L_2$ and two components $C_1,C_2$ of $L_1,L_2$, 
respectively, such that
\begin{itemize}
  \item $\max(L_2) < x < \min(L_1)$,
  \item $\ell_1 \lesstime \an{C_1} \lesstime \ell \lesstime \an{C_2} \lesstime \ell_2$,
  \item $\out{\tr{C_1}},\out{\tr{C_2}}\neq\emptystr$.
\end{itemize}
In particular, since $\an{C_1}$ is to the right of $\an{C_2}$ w.r.t.~the order
of positions, we know that $(L_1,C_1,L_2,C_2)$ is an inversion, and hence
$\an{C_1} \simeq \an{C_2}$. But this contradicts the assumption that 
$\rho[\ell_1,\ell_2]$ does not overlap with any $\simeq$-block.
\end{proof}

\begin{repproposition}{prop:sufficiency}
Given a functional two-way transducer $\cT$, 
one can construct in $\threeexptime$ 
a one-way transducer $\cT'$ such that 
$\cT' \subseteq \cT$ and $\dom(\cT') \supseteq U$.
\end{repproposition}

\begin{proof}
Given an input $u$, the transducer $\cT'$ will guess (and check) 
a successful run $\rho$ of $\cT$ on $u$, together with a decomposition 
$(\rho[\ell_i,\ell_{i+1}])_i$ of $\rho$ into blocks and diagonals.
The decomposition will be used by $\cT'$ to simulate the output of 
$\rho$ left-to-right, thus proving that $\cT' \subseteq \cT$.
Moreover, $u\in U$ implies the existence of a successful run 
that can be decomposed, thus proving that $\dom(\cT') \supseteq U$.
We now provide some details of the construction of $\cT'$.

Guessing the run $\rho$ is standard (see, for instance, \cite{she59,HU79}): 
it amounts to guess the crossing sequences $\rho|x$ for 
each position $x$ of the input. Recall that this is a bounded
amount of information for each $x$, since the run is normalized.
As concerns the decomposition of $\rho$, it can be encoded
by the endpoints $\ell_i$ of its factors, that is, by annotating 
the position of each $\ell_i$ as the level of $\ell_i$.
In a similar way $\cT'$ guesses the information of whether
each factor $\rho[\ell_i,\ell_{i+1}]$ is a diagonal or a block.

Thanks to the definition of decomposition (Def.~\ref{def:decomposition}), 
every two distinct factors span across non-overlapping intervals of positions. 
This means that each position $x$ is covered by exactly one factor of 
the decomposition. We call this factor the \emph{active factor at position $x$}.
The mode of computation of the transducer will depend on 
the type of active factor: if the active factor is a diagonal
(resp.~a block), then we say that $\cT'$ is in \emph{diagonal mode} 
(resp.~\emph{block mode}). 
Below we describe the behaviour for these two modes of computation.

\smallskip
\par\noindent\emph{Diagonal mode.}~
We recall the key condition satisfied by the diagonal 
$\rho[\ell_i,\ell_{i+1}]$ that is active at position $x$ 
(cf.~Def.~\ref{def:factors} and Figure~\ref{fig:factors}):
there exists a location $\ell_x=(x,y_x)$ 
between $\ell_i$ and $\ell_{i+1}$ such that the words
$\out{\rho\mid Z_{\ell_x}^\upperleft}$ and $\out{\rho\mid Z_{\ell_x}^\lowerright}$ 
have length at most $\bound$, where 
$Z_{\ell_x}^\upperleft = [{\ell_x},\ell_2] \:\cap\: \big([0,x]\times\bbN\big)$
and $Z_{\ell_x}^\lowerright = [\ell_1,{\ell_x}] \:\cap\: \big([x,\omega]\times\bbN\big)$.

Besides the run $\rho$ and the decomposition, the transducer $\cT'$ will
also guess the locations $\ell_x=(x,y_x)$, that is, will annotate each $x$
with the corresponding $y_x$.
Without loss of generality, we can assume that the function that 
associates each position $x$ with the guessed location $\ell_x=(x,y_x)$ 
is monotone, namely, $x\le x'$ implies $\ell_x\leqtime\ell_{x'}$.
While the transducer $\cT'$ is in diagonal mode, the goal is to preserve 
the following invariant: 

\begin{quote}
\em
After reaching a position $x$ covered by the active diagonal, 
$\cT'$ must have produced the output of $\rho$ up to location $\ell_x$. 
\end{quote}

\noindent
To preserve the above invariant when moving from $x$ to the next 
position $x+1$, the transducer should output the word 
$\out{\rho[\ell_x,\ell_{x+1}]}$. This word consists of
the following parts:
\begin{enumerate}
  \item The words produced by the single transitions of $\rho[\ell_x,\ell_{x+1}]$
        with endpoints in $\{x,x+1\}\times\bbN$. 
        Note that there are at most $\hmax$ such words, 
        each of them has length at most $\cmax$, and they can all be determined 
        using the crossing sequences at $x$ and $x+1$ and the information
        about the levels of $\ell_x$ and $\ell_{x+1}$.
        We can thus assume that this information is readily available
        to the transducer.
  \item The words produced by the factors of $\rho[\ell_x,\ell_{x+1}]$ 
        that are intercepted by the interval $[0,x]$. 
        Thanks to the definition of diagonal, we know that
        the total length of these words is at most $\bound$.
        These words cannot be determined from the information
        on $\rho|x$, $\rho|x+1$, $\ell_x$, and $\ell_{x+1}$
        alone, so they need to be constructed while scanning the input.
        For this, it is important to store additional information. 
        
        More precisely, at each position $x$ of the input, 
        the transducer stores all the outputs produced by the factors of 
        $\rho$ that are intercepted by $[0,x]$ and that occur {\sl after} 
        a location of the form $\ell_{x'}$, for any $x'\ge x$ that is 
        covered by a diagonal.
        This clearly includes the previous words when $x'=x$, but also 
        other words that might be used later for processing other diagonals.
        Moreover, by exploiting the properties of diagonals,
        one can prove that those words have length at most $\bound$, 
        so they can be stored with triply exponentially many states.
        Using classical techniques, the stored information
        can be maintained while scanning the input $u$ using the
        guessed crossing sequences of $\rho$.
  \item The words produced by the factors of $\rho[\ell_x,\ell_{x+1}]$ 
        that are intercepted by the interval $[x+1,\omega]$. 
        These words must be guessed, since they depend on a portion
        of the input that has not been processed yet. 
        Accordingly, the guesses need to be stored into memory,
        so that they can be checked later. Formally, the transducer 
        stores, for each position $x$, the guessed words that correspond 
        to the outputs produced by the factors of $\rho$ intercepted by 
        $[x,\omega]$ and occurring {\sl before} a location of the form 
        $\ell_{x'}$, for any $x'\le x$ that is covered by a diagonal.
\end{enumerate}

\smallskip
\par\noindent\emph{Block mode.}~
Suppose that the active factor $\rho[\ell_i,\ell_{i+1}]$ is a block.
Let $I=[x_i,x_{i+1}]$ be the set of positions covered by this factor.
Moreover, for each position $x\in I$, let 
$Z^\leftshort_x = [\ell_i,\ell_{i+1}] \:\cap\: \big([0,x]\times \bbN\big)$
and $Z^\rightshort_x = [\ell_i,\ell_{i+1}] \:\cap\: \big([x,\omega]\times \bbN\big)$.
We recall the key property of a block
(cf.~Definition~\ref{def:factors} and Figure~\ref{fig:factors}): 
the word $\out{\rho[\ell_{i_x},\ell_{i_x+1}]}$ is almost periodic with bound $\bound$, 
and the words $\out{\rho\mid Z^\leftshort_{x_i}}$ and $\out{\rho\mid Z^\rightshort_{x_{i+1}}}$ 
have length at most $\bound$.

For the sake of simplicity, suppose that $\out{\rho[\ell_i,\ell_{i+1}]} = w_1\,w_2,w_3$,
where $w_2$ is periodic with period $\bound$ and $w_1,w_2$ have length at most $\bound$.
Similarly, let $w_0 = \out{\rho\mid Z^\leftshort_{x_i}}$ and $w_4 = \out{\rho\mid Z^\rightshort_{x_{i+1}}}$.
The invariant preserved by $\cT'$ in block mode is the following: 

\begin{quote}
\em
After reaching a position $x$ covered by the active block $\rho[\ell_i,\ell_{i+1}]$, 
$\cT'$ must have produced the output of the prefix of $\rho$
up to location $\ell_i$, followed by a prefix of $\out{\rho[\ell_i,\ell_{i+1}]} = w_1\,w_2\,w_3$
of the same length as $\out{\rho\mid Z^\leftshort_x}$.
\end{quote}

\noindent
The initialization of the invariant is done when reaching the left 
endpoint $x_i$ of the interval $I$. At this moment, it suffices that $\cT'$ outputs 
a prefix of $w_1\,w_2\,w_3$ of the same length as $w_0 = \out{\rho\mid Z^\leftshort_{x_i}}$,
thus bounded by $\bound$.
Symmetrically, when reaching the right endpoint $x_{i+1}$ of $I$,
$\cT'$ will have produced almost the entire word 
$\out{\rho[\ell_1,\ell_i]} \, w_1 \, w_2 \, w_3$,
but without the suffix of length $|w_4| \le \bound$. 
Thus, before moving to the next factor of the decomposition, the transducer will 
have to produce the remaining suffix, so as to complete the output 
of $\rho$ up to location $\ell_{i_x+1}$.

It remains to describe how the above invariant can be maintained
when moving from a position $x$ to the next position $x+1$ inside $I$.
For this, it is convenient to succinctly represent the word $w_2$ 
by its repeating pattern, say $v$, of length at most $\bound$. 
To determine the symbols that have to be output at each step,
the transducer will maintain a pointer on either $w_1\,v$ or $w_3$.
The pointer is increased in a deterministic way, and precisely
by the amount $\big|\out{\rho\mid Z^\leftshort_{x+1}}\big| - \big|\out{\rho\mid Z^\leftshort_x}\big|$.
The only exception is when the pointer lies in $w_1\,v$, but its 
increase would go over $w_1\,v$: in this case the transducer has 
the choice to either bring the pointer back to the beginning of $v$ 
(representing a periodic output inside $w_2$), or move it to $w_3$. 
Of course, this is a non-deterministic choice, but it can be 
validated when reaching the right endpoint of $I$.
Concerning the number of symbols that need to be emitted at each
step, this can be determined from the crossing sequences at
$x$ and $x+1$, and from the knowledge of the lowest and highest 
levels of locations that are at position $x$ and between 
$\ell_i$ and $\ell_{i+1}$. We denote the latter levels by
$y^-_x$ and $y^+_x$, respectively.

Overall, this shows how to maintain the invariant of the block mode,
assuming that the levels $y^-_x,y^+_x$ are known, as well as
the words $w_0,w_1,v,w_3,w_4$ of bounded length.
Like the mapping $x\mapsto\ell_x=(x,y_x)$ used in diagonal mode, 
the mapping $x\mapsto (y^-_x,y^+_x)$ can be guessed and checked 
using the crossing sequences.
Similarly, the words $w_1,v,w_3$ can be guessed just before
entering the active block, and can be checked along the process.
As concerns the words $w_0,w_4$, these can be guessed and checked 
in a way similar to the words that we used in diagonal mode.
More precisely, for each position $x$ of the input, the
transducer stores the following additional information:
\begin{enumerate}
  \item the outputs produced by the factors of $\rho$ that are 
        intercepted by $[0,x]$ and that occur after the beginning
        $\ell_j$ of a block, where $\ell_j=(x_j,y_j)$ and $x_j\ge x$;
  \item the outputs produced by the factors of $\rho$ that are 
        intercepted by $[x,\omega]$ and that occur before the ending
        $\ell_{j+1}$ of a block, where $\ell_{j+1}=(x_{j+1},y_{j+1})$ 
        and $x_{j+1}\le x$.
\end{enumerate}
Thanks to the properties of blocks, the above words have length 
at most $\bound$ and can be maintained while processing the input
and the crossing sequences.
Finally, we observe that the words, together with the information 
given by the lowest and highest levels $y^-_x,y^+_x$, for both $x=x_i$ and
$x=x_{i+1}$, are sufficient for determining the content of $w_0$ and $w_4$.

\smallskip
The above constructions give a one-way transducer $\cT'$ of size 
triple exponential in $\cT$. 
\end{proof}

\begin{reptheorem}{thm:sweeping}
A functional two-way transducer $\cT$ is sweeping definable iff
it is $k$-pass sweeping definable, for $k=2\hmax\cdot (2^{3\emax}+1)$.
\end{reptheorem}

\begin{proof}
Suppose that $\cT$ is not $k$-pass sweeping definable for $k=2\hmax\cdot (2^{3\emax}+1)$.
We aim at proving that $\cT$ is not $m$-pass sweeping definable for all $m>0$.
By Theorem \ref{thm:k-pass-sweeping}, we know that there exist
a successful run $\rho$ and a $k$-inversion $\bar\cI = (\cI_0,\dots,\cI_{k-1})$ 
of it, with $\cI_i=(L_i,C_i,L'_i,C'_i)$, that is not safe. 
We consider the locations of $\rho$ that are visited between the beginning 
of an inversion $\cI_i$ and the ending of the next co-inversion $\cI_{i+1}$.
Formally, for all even indices $i=0,2,\dots,k-1$, we let
\[
  K_i ~=~ \big[\an{C_i},\an{C'_{i+1}}\big].
\]
We then project each $K_i$ on the $x$-coordinates:
\[
  X_i ~=~ \big\{ x ~:~ \exists \: \ell=(x,y)\in K_i \big\}.
\]
Since $K_i$ is an interval of locations and the transducer $\cT$ can only move
its head between consecutive positions, we know that each $X_i$ is an interval 
of positions. Hereafter, we often use the term ``interval'' to denote a set
of the form $X_i$, for some even index $i\in\{0,2,\dots,\k-1\}$.

Below we prove that there is a large enough set of pairwise non-overlapping intervals:

\begin{claim}
There is a set $\cX=\{X_i\}_{i\in I}$ of cardinality $n=2^{3\emax}+1$ 
such that $X \cap X' = \emptyset$ for all $X\neq X'\in \cX$.
\end{claim}

\begin{claimproof}
In this proof, we consider an ordering on the intervals $X_i$ 
different from the one induced by the indices $i$. This is given
by the lexicographic order on the endpoints, where the dominant 
element is the rightmost endpoint, namely,
we let $X_i < X_j$ if either $\max(X_i) < \max(X_j)$, or 
$\max(X_i)=\max(X_j)$ and $\min(X_i)<\min(X_j)$.

We construct the set $\cX$ inductively, by following the lexicographic ordering. 
Formally, for all $j=0,\dots,n$, we construct:
\begin{itemize}
  \item a set $\cX_j$ of size $j$ such that 
        $X \cap X' = \emptyset$ for all $X\neq X' \in \cX_j$
  \item a set $\cX'_j$ of size at least $\hmax\cdot(2^{3\emax}+1 - j)$
        such that, for all $X\in\cX_j$ and all $X'\in\cX'_j$,
        $\max(X) < \min(X')$ 
        (namely, all intervals of $\cX'_j$ are strictly to the right of 
         the intervals of $\cX_j$).
\end{itemize}
The base case $j=0$ of the induction is easy: we let $\cX_0=\emptyset$
and $\cX'_0$ be the set of all intervals. It only suffices to observe 
that $\cX'_0$ has cardinality $\frac{k}{2}=\hmax\cdot(2^{3\emax}+1)$.

For the inductive step, suppose that $j<n=2^{3\emax}+1$ and that we 
constructed $\cX_j$ and $\cX'_j$ satisfying the inductive hypothesis. 
We let $X$ be the least element in $\cX'_j$
according to the lexicographic order 
(note that $\cX'_j\neq\emptyset$ since $j<n$). 
Accordingly, we define $\cX_{j+1} = \cX_j \cup \{X\}$ and
$\cX'_{j+1}$ as the subset of $\cX'_j$ that contains the intervals 
strictly to the right of $X$.
It remains to verify that $\cX'_{j+1}$ has cardinality at least 
$\hmax\cdot\big(2^{3\emax}+1 - (j+1)\big)$. 
For this we recall that the run $\rho$ is normalized. This implies that 
there are at most $\hmax$ intervals in $\cX'_j$ that cover the position 
$x=\max(X)$.
All other intervals of $\cX'_j$ are necessarily to the right of $X$:
indeed, because $X$ is minimal in the lexicographic ordering,
we know that every interval of $\cX'_j$ 
has the right endpoint 
to the right of $x$, and as they do not cover the position $x$, their left endpoint too.
This shows that there are at most $\hmax$ intervals in
$\cX'_j \:\setminus\: \cX'_{j+1}$, so 
$|\cX'_{j+1}| \ge \hmax\cdot\big(2^{3\emax}+1 - (j+1)\big)$. 
\end{claimproof}

Turning back to the proof of the theorem, we consider the left endpoints
of the intervals in $\cX$, say
\[
  \olft X ~=~ \{ \min(X) ~:~ X\in\cX \}.
\]
Since $|\olft X| > 2^{3\emax}$, we can use Theorem \ref{th:simon} 
to derive the existence of three distinct positions $x < x' <x'' \in \olft X$ 
such that $[x,x']$ and $[x',x'']$ are consecutive idempotent loops of $\rho$ with 
the same effect (see also the proof of Theorem~\ref{thm:simon2} for a similar claim).
We let $L=[x,x'']$ be the union of those two loops, and we consider the intermediate
position $x'$. We recall that $x'$ is the left endpoint of an interval of $\cX$,
which we denote by $X_i$ for simplicity. We also recall that $X_i$ is the set of
positions visited by a factor of the run $\rho$ that goes from the first anchor
$\an{C_i}$ of the inversion $\cI_i=(L_i,C_i,L'_i,C'_i)$ to the second anchor
$\an{C'_{i+1}}$ of the co-inversion $\cI_{i+1}=(L_{i+1},C_{i+1},L'_{i+1},C'_{i+1})$.

We claim that the inversion $\cI_i$ and the co-inversion $\cI_{i+1}$ 
occur in the same factor intercepted by $L$.
Indeed, the factor $\rho[\an{C_i},\an{C'_{i+1}}]$ visits only positions 
inside the interval $X_i$. Moreover, the endpoints of $X_i$ are strictly 
between the endpoints of $L$, namely,
\[
  \min(L) ~=~ x ~<~ x' ~=~ \min(X_i) ~\le~ \max(X_i) ~<~ x'' ~=~ \max(L).
\]
This shows that the inversion $\cI_i=(L_i,C_i,L'_i,C'_i)$
and the co-inversion $\cI_{i+1}=(L_{i+1},C_{i+1},L'_{i+1},C'_{i+1})$ 
occur in the same factor intercepted by $L$, which we denote by $\alpha$.

Now, we can easily introduce new copies of the factor $\alpha$, and hence new
copies of the (co)-invesions $\cI_i$ and $\cI_{i+1}$, by pumping the idempotent 
loop $L$. Formally, for all $m>0$, we denote by 
$\cI_i^{(1)},\dots,\cI_i^{(m)}$ (resp.~$\cI_{i+1}^{(1)},\dots,\cI_{i+1}^{(m)}$)
the $m$ copies of the inversion $\cI_i$ (resp.~the $m$ copies of the co-inversion 
$\cI_{i+1}$) that appear in the pumped run $\pump_L^m(\rho)$.
For the sake of simplicity, we assume that those copies are listed
according to their order of occurrence in the pumped run, namely,
\[ 
  \cI_i^{(1)} ~\lesstime~ \cI_{i+1}^{(1)} ~\lesstime~
  \cI_i^{(2)} ~\lesstime~ \cI_{i+1}^{(2)} ~\lesstime~
  \dots ~\lesstime~
  \cI_i^{(m)} ~\lesstime~ \cI_{i+1}^{(m)}
\]
(the order $\lesstime$ is extended from locations to (co-)inversions in the natural way).

Towards a conclusion, we observe that 
$\big(\cI_i^{(1)},\cI_{i+1}^{(1)},\dots,\cI_i^{(m)},\cI_{i+1}^{(m)}\big)$ 
is a $2m$-inversion of the successful run $\pump_L^m(\rho)$ of $\cT$. 
Moreover, this $2m$-inversion is not safe, since it consists of (co-)inversions
that do not generate periodic outputs --- more formally, the period of the 
word $\out{\tr{C_i}} \: \out{\rho[\an{C_i},\an{C'_i}]} \: \out{\tr{C'_i}}$
(resp.~$\out{\tr{C_{i+1}}} \: \out{\rho[\an{C_{i+1}},\an{C'_{i+1}}]} \: \out{\tr{C'_{i+1}}}$)
is larger than $\bound$ or does not divide $|\out{\tr{C_i}}|$ and $|\out{\tr{C'_i}}|$
(resp.~$|\out{\tr{C_{i+1}}}|$ and $|\out{\tr{C'_{i+1}}}|$). 
By Theorem \ref{thm:k-pass-sweeping}, this proves that $\cT$ is not
$m$-pass sweeping definable.
Finally, since the above holds for all $m>0$, we conclude that $\cT$ 
is not sweeping definable.
\end{proof}

\end{document}